\documentclass{article}

\usepackage{arxiv}
\usepackage{amsmath,amssymb}
\usepackage{amsthm}
\usepackage{algpseudocode}
\usepackage{enumitem}
\usepackage{comment}
\usepackage{caption}
\usepackage{subfigure}
\usepackage{subcaption}
\usepackage{algorithm}\usepackage{mathtools}
\usepackage[utf8]{inputenc}
\usepackage{graphicx}
\usepackage{bbm}
\usepackage[utf8]{inputenc} 
\usepackage[T1]{fontenc}    
\usepackage{hyperref}       
\usepackage{url}            
\usepackage{booktabs}       
\usepackage{amsfonts}       
\usepackage{nicefrac}       
\usepackage{microtype}      
\usepackage{lipsum}
\usepackage{graphicx}
\graphicspath{ {./images/} }

\newcommand\myeqa{\stackrel{\mathclap{\normalfont\mbox{optimize $Q_2$}}}{\leq}}
\newcommand\myeqaa{\stackrel{\mathclap{\normalfont\mbox{optimize $Q_1$}}}{\leq}}
\newcommand\myeqb{\stackrel{\mathclap{\normalfont\mbox{Lemma 1}}}{\leq}}
\DeclareMathOperator*{\argmax}{arg\,max}
\newcommand{\NB}{\text{NB}}

\newtheorem{theorem}{Theorem}
\newtheorem{corollary}{Corollary}[theorem]

\title{Flexible Empirical Bayesian Approaches  to Pharmacovigilance for Simultaneous Signal Detection and Signal Strength Estimation in   Spontaneous Reporting Systems Data\protect}

\date{}

\author{
 Yihao Tan \\
  Department of Biostatistics\\
  School of Public Health and Health Professions\\
  State University of New York at Buffalo\\
  Buffalo, New York, USA\\
   \And
 Marianthi Markatou\thanks{Joint senior and corresponding authors} \\
  Department of Biostatistics\\
  School of Public Health and Health Professions\\
  State University of New York at Buffalo\\
  Buffalo, New York, USA\\
  \texttt{markatou@buffalo.edu} \\
  \And
 Saptarshi Chakraborty$^*$ \\
  Department of Biostatistics\\
  School of Public Health and Health Professions\\
  State University of New York at Buffalo\\
  Buffalo, New York, USA\\
  \texttt{chakrab2@buffalo.edu} \\
}

\shortauthors{Tan, Markatou, and Chakraborty}

\begin{document}
\maketitle
\begin{abstract}
Inferring adverse events (AEs) of medical products from Spontaneous Reporting Systems (SRS) databases is a core challenge in contemporary pharmacovigilance. Bayesian methods for pharmacovigilance are attractive for their rigorous ability to simultaneously detect potential AE signals and estimate their strengths/degrees of relevance. However, existing Bayesian and empirical Bayesian methods impose restrictive parametric assumptions and/or demand substantial computational resources, limiting their practical utility. This paper introduces a suite of novel, scalable empirical Bayes methods for pharmacovigilance that utilize flexible non-parametric priors and custom, efficient data-driven estimation techniques to enhance signal detection and signal strength estimation at a low computational cost. Our highly flexible methods accommodate a broader range of data and achieve signal detection performance comparable to or better than existing Bayesian and empirical Bayesian approaches. More importantly, they provide coherent and high-fidelity estimation and uncertainty quantification for potential AE signal strengths, offering deeper insights into the comparative importance and relevance of AEs. Extensive simulation experiments across diverse data-generating scenarios demonstrate the superiority of our methods in terms of accurate signal strength estimation, as measured by replication root mean squared errors. Additionally, our methods maintain or exceed the signal detection performance of state-of-the-art techniques, as evaluated by frequentist false discovery rates and sensitivity metrics. Applications on FDA FAERS data for the statin group of drugs reveal interesting insights through Bayesian posterior probabilities.
\end{abstract}


\section{INTRODUCTION}\label{sec:intro}

Postmarket medical product safety assessment is a core challenge of contemporary pharmacovigilance. It focuses on the detection and analysis of adverse events (AEs) of medical products once they enter the market \cite{amery1999there}. This surveillance plays a major role in ensuring the ongoing safety and efficacy of medical products in real-world settings, as it can uncover AEs that may not have been evident during clinical trials, which are often limited by size, duration, and scope. Large Spontaneous Reporting Systems (SRS) databases established worldwide---such as the US Food and Drug Administration (FDA)'s Adverse Event Reporting System (FAERS) database---constitute key resources in this endeavor, curating data from both mandatory reports by pharmaceutical companies and voluntary submissions by healthcare professionals and patients.

However, SRS data are observational and present challenges such as underreporting of AEs, lack of proper controls, inaccuracies in measuring drug use, and the presence of selection bias and confounding \cite{markatou2014pattern}. These challenges hinder the ability to infer causal relationships between drugs and AEs from SRS data. Instead, pharmacovigilance commonly focuses on detecting AE signals using SRS data mining methods that analyze AE-drug associations. A common theme underlying many of these methods is to compare the observed frequencies ($\{O\}$) of specific AE-drug pairs against their \textit{null baseline expected} frequencies ($\{E\}$), where $E$ represents the theoretical expected count if there had been no association between the AE and the drug.  The analysis is performed either across all drugs and AEs in the database or within specific subsets, such as all reports from a particular year. AE-drug combinations with substantially higher than expected reporting rates (i.e., $O/E \gg 1$) are identified as potential AE safety signals in these methods.

Various SRS data mining methods for AE identification have been proposed over the past decades, including proportional reporting ratios (PRR) \cite{evans2001use}, reporting odds ratio (ROR) \cite{rothman2004reporting}, likelihood ratio test (LRT) based methods \cite{ding2020evaluation, huang2011likelihood, chakraborty2022use, zhao2018extended, huang2017zero}, and Bayesian methods \cite{huang2013likelihood, bate1998bayesian, dumouchel1999bayesian, hu2015signal}. Some of these approaches are heuristic, employing ad hoc thresholds directly for the $\{O/E\}$ values to identify AE-drug combinations as signals. More formal approaches parametrize the $\{O/E\}$ values, or some functions thereof, using probabilistic models for the observed report counts and suggest more principled signal determination based on hypothesis tests with controlled frequentist type I errors and false discovery rates, or on prespecified posterior probabilities of being a signal from an appropriately articulated Bayesian model for the data.

A central focus of these SRS data mining methods is the \textit{detection of signals}, which categorizes AE-drug pairs as either ``signals'' or ``non-signals.'' This includes the class of all LRT-based approaches that utilize frequentist uncertainty to rigorously identify AE signals. While identifying AE signals is critical, contemporary science and biomedicine increasingly recognize the limitations of a simple signal/non-signal dichotomy for statistical inference, which can obscure data nuances \cite{amrheinScientistsRiseStatistical2019, gelmanStatisticalCrisisScience2016, rothmanDisengagingStatisticalSignificance2016, wassersteinASAStatementPValues2016}. Transcending to a more comprehensive analysis that permits signal strength estimation and uncertainty quantification, in addition to signal detection, can provide substantially deeper insights, particularly in pharmacovigilance. First, the estimated signal strengths can quantify the relevance of specific AE-drug combinations rather than merely identifying the statistically significant combinations, thereby avoiding the "statistical significance filter" problem \cite{gelman2018failure, vasishthStatisticalSignificanceFilter2018, vanzwetSignificanceFilterWinner2021}. Second, inferred signal strengths can facilitate a coherent comparison of the relevance of multiple AEs across different drugs or data subsets (e.g., from different years). Third, the quantified uncertainties, such as interval estimates for the signal strengths, allow for a more nuanced assessment of the statistical guarantees of the results.

Bayesian approaches can be particularly effective for principled estimation and uncertainty quantification of signal strengths across all AE-drug combinations in SRS data. These methods utilize a prior distribution---a probability distribution defined over the model parameters describing the O/E ratios---and combine it with the information obtained solely from the data (the likelihood) to obtain the posterior distribution for the model parameters. The resulting posterior permits rigorous probabilistic inference on model parameters while coherently accounting for modeling uncertainty by formally combining the prior and the data likelihood.  It is important to note that specifying the prior distribution is a critical step in any Bayesian analysis that reflects the underlying considerations of the analysis \cite{efronBayesTheorem21st2013} (see Section \ref{sec:general-gamma} for a discussion on our approach to prior selection for our method). For example, \textit{objective priors} aim for purely data-driven inference analogous to a conventional frequentist inference \cite{welchFormulaeConfidencePoints1963, bergerCaseObjectiveBayesian2006}; \textit{subjective priors} incorporate existing and/or historical knowledge outside the data into the analysis \cite{yiBayesianModelChoice2003a, ohaganUncertainJudgementsEliciting2006, vieleUseHistoricalControl2014, ibrahimPowerPriorTheory2015a}; and \textit{data-informed priors}---which are influenced by specific aspects of the data such as sample size, measurement scale, or predictors in regression models---aim to improve model estimation in specific data analysis settings \cite{zellner1986assessing, piironenHyperpriorChoiceGlobal2017, gelmanPriorCanOften2017}. In the context of SRS data mining for pharmacovigilance, where the $O/E$ parameters are often high-dimensional due to the large number of possible AE-drug pairs, using prior distributions can yield shrinkage in the resulting Bayesian estimates for the parameters. Specifically, the analysis may produce an estimated $O/E$ parameter for each AE-drug combination that is influenced both by the observed counts for that pair and by a summarized aggregate of counts for all other AE-drug pairs. Since most AE-drug pairs are ``non-signals'' with O/E parameter values close to 1, the final estimates of the O/E parameters are typically shrunken towards 1 for pairs with relatively higher observed O/E values.

Extensive theoretical research over the past decades has provided strong formal justifications for shrinkage estimation in high-dimensional data-driven problems from both frequentist and Bayesian perspectives \cite{jamesEstimationQuadraticLoss1992, tibshiraniRegressionShrinkageSelection1996, parkBayesianLasso2008, carvalhoHorseshoeEstimatorSparse2010a,bhadraLassoMeetsHorseshoe2019a}. We provide a conceptual overview of shrinkage estimation in pharmacovigilance from a frequentist viewpoint, where shrinkage arises as a consequence of regularized estimation to enhance efficiency. For this, we leverage the frequentist replicability framework:  we consider the existence of unobserved ``true'' signal strength parameters in the population that can generate an infinite ensemble of random SRS datasets, with the observed dataset being one single sample from this infinite ensemble. We further assume that the ``true'' O/E ratios for a handful of AE-drug pairs (the "signals") exceed 1, while for the vast majority of pairs (the "non-signals"), they equal 1. However, due to frequentist sampling variability, some non-signal pairs may exhibit noisy observed counts in any individual SRS dataset. Indeed, as the total number (dimension) of AE-drug pairs increases, the probability of at least one pair exhibiting a large observed count---leading to a large (>1) O/E ratio despite an underlying ``true'' O/E value of 1---also increases. Additionally, some AE-drug pairs (unknown in advance) may be biologically or physically impossible to be reported together; these combinations, termed "structural zeros," are often modeled using zero-inflated count models, resulting in zero "true values" for the corresponding signal strength parameters.

To illustrate, Figure 1 shows a histogram of the observed O/E values from 46 relevant AEs across 6 statin drugs (Atorvastatin, Pravastatin, Simvastatin, Rosuvastatin, Lovastatin) obtained from the FDA FAERS database for the quarters of 2014Q3 - 2020Q4. The expected (E) values are calculated assuming independence between AEs and drugs. These empirical O/E values demonstrate a noisy pattern that can be stabilized through shrinkage estimation, but the overall distribution shows clear multi-modality for the O/E values. Specifically, we see (a) a clear peak around zero arising from the zero observed counts, (b) a distinct peak around 1 from the ``non-signal'' pairs, and (c) several other smaller peaks around values larger than one.

\begin{figure}[htp]
    \centering
    \includegraphics[width=14cm]{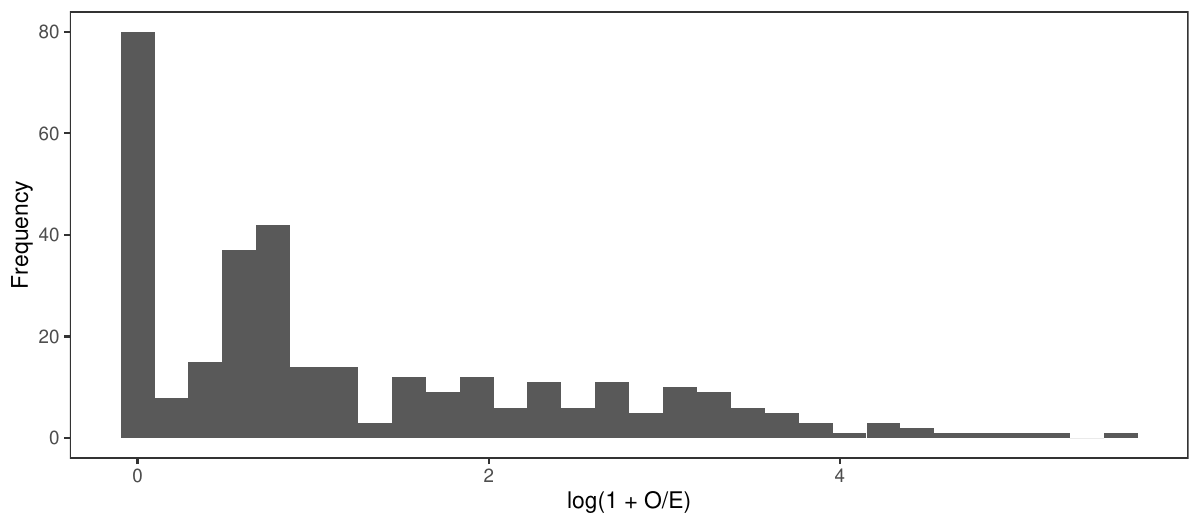}
    \caption{Histogram of ``observed'' O/E ratios for AEs associated with six statin drugs curated from the FDA FAERs database.}
\end{figure}

The discussion underscores the critical importance of selecting an appropriate prior distribution for O/E values in Bayesian SRS data mining for informed yet objective inference. Specifically, the prior distribution must meet several key criteria, as discussed below (see Section \ref{sec:np-empirical-bayes} for a detailed discussion). First, it should be sufficiently ``vague'' to reliably reflect the data-driven (likelihood) information in the posterior while providing enough structure to facilitate shrinkage for principled inference. Second, it must be flexible enough to accommodate multiple modes in the distribution of O/E values arising from the signal, non-signal, and possibly zero-inflation inducing ``structural zero'' AE-drug pairs. Finally, for formal inference on signal strengths, the prior distribution should jointly consider all AE-drug pairs for coherent uncertainty quantification for formal Bayesian inference. 

Although several Bayesian SRS data mining approaches currently exist, few meet all these criteria. Simpler approaches often lack the flexibility to capture multiple modes in the O/E value distribution (e.g., the Poisson-single-gamma model \cite{huang2013likelihood}; see Section \ref{sec:single-gamma} for a review) or do not account for the joint variability among all O/E values (e.g., the Bayesian confidence propagation neural networks (BCPNN) model \cite{bate1998bayesian}; Section \ref{sec:bcpnn}). More advanced methods like the gamma Poisson shrinker (GPS) \cite{dumouchel1999bayesian} and the multi-item gamma Poisson shrinker (MGPS) \cite{dumouchel2001empirical} use a more flexible two-component mixture of gamma distributions---one component for ``signals'' and another for ``non-signals''---as the prior. However, even this flexibility may be inadequate, as shown in our simulation experiments (Section \ref{sec:simulation}). Finally, fully non-parametric Bayesian approaches, such as Dirichlet process mixture models \cite{hu2015signal}, offer rigorous flexible inference but often rely heavily on model hyperparameters that are not straightforward to determine and require significant computational resources for implementation, limiting their scalability for large datasets. These challenges hinder the ability of these approaches to permit accurate inference on AE signal strengths, even if they aid accurate signal detection.

The remainder of the article is organized as follows. Section \ref{sec:our-contributions} summarizes our contributions, highlighting key novelties. Section \ref{sec:statin-data-description} describes our motivating dataset of AEs from statin drugs, used for both simulations and real data analysis. Section \ref{sec:review} reviews existing Bayesian and empirical Bayesian SRS data mining approaches. Section \ref{sec:np-empirical-bayes} introduces a general framework of non-parametric empirical Bayesian approaches for signal detection and signal strength estimation in SRS data mining, suggesting three prominent methods within this framework. Two of these are adapted from recent methodologies in fields outside pharmacovigilance, while the third is our original contribution, with a highly efficient expectation-conditional-maximization-based implementation strategy (Algorithm~\ref{alg:algorithm1}) with guaranteed convergence (Theorem~\ref{thm:thm1}). This section also proposes a novel estimator for the null baseline expected counts $\{E\}$, with proven superiority over existing estimators in certain settings (Theorem~\ref{thm2:Eestimator}).  Section \ref{sec:simulation} presents extensive simulation results comparing the proposed and existing Bayesian and empirical Bayesian SRS data mining methods across a range of settings based on the motivating dataset. Section \ref{sec:real-data-analysis} demonstrates the application of our approaches in real-world SRS data mining, comparing the results with existing methods. Finally, Section \ref{sec:discussion} concludes with a brief discussion of our contributions and future directions. Technical details, including derivations, theorem proofs, additional computational notes, and supplementary simulation results, are provided in the Appendix.

\subsection{Our Contributions} \label{sec:our-contributions}

In this paper, we introduce a suite of flexible empirical Bayesian approaches that meet all the criteria for prior distributions discussed above, while remaining computationally tractable and scalable for large datasets. Our primary approach, called the general-gamma mixture model, begins with a Poisson model layer and employs a novel and highly flexible non-parametric prior based on sparse finite mixtures of gamma distributions \cite{fruhwirth2019here} and a novel and more principled (compared to the state-of-the-art) estimate of the expected count, E, for the O/E parameter. Additionally, we adapt two other non-parametric empirical Bayes methods introduced in the last decade and propose their use in the current context for SRS data mining in pharmacovigilance: the Koenker-Keifer-Wolfowitz discrete mixture prior-based method \cite{koenker2014convex}, referred to as the KM model, and the Efron smooth $g$-mixture prior-based method \cite{efron2016empirical}, referred to as the Efron model.

Our approaches offer several desirable properties for practical use. First, as with most Bayesian pharmacovigilance approaches, they allow for intuitive and interpretable posterior-based probabilistic statements on signal positions and strengths. Second, the non-parametric nature of these methods, owing to the general mixture priors, ensures they can adaptively capture multiple modes for O/E values, allowing accommodation for non-signal, signal, and structural zero AE-drug combinations in the data. Consequently, there is no need for a separate zero-inflation component for structural zero report counts in the model; a mixture component in the prior distribution of O/E values around zero will coherently capture such zero-inflations. This adaptability allows our methods to handle a richer set of SRS data while still enabling regularized estimation of the model parameters for statistically stable inference. 

Third, the empirical Bayes aspect of our approach involves scalable, consistent point estimation of the mixture prior for efficient implementation, contrasting with a full Bayesian MCMC-based approach that requires substantially more computational power. Specifically, we develop a novel expectation-maximization (EM) algorithm to estimate the underlying sparse general-gamma mixture prior within an empirical Bayes framework. Simultaneously, we leverage existing non-parametric maximum likelihood estimation (NPMLE) and penalized likelihood maximization frameworks for the KM and Efron models, respectively. All three implementations are computationally efficient, in contrast to the existing non-parametric full Bayesian Dirichlet process mixture-based approach \cite{hu2015signal}, which is substantially more computationally demanding, as observed in our experiments. Our resulting empirical Bayes posterior maintains high accuracy for both signal detection—matching or surpassing competing approaches—and signal strength estimation, often outperforming all existing methods, including the full Bayesian Dirichlet process mixture approach. Fourth, the multi-modality in the estimated prior can propagate into the marginal posterior distribution of the signal strength of an individual AE-drug pair. This implies that, with high \textit{posterior} probability, the corresponding signal strength can take values in two or more distinct regions, such as those associated with structural zero and possibly signal or non-signal combinations. Consequently, our approach offers nuanced estimations of signal strengths for AE-drug pairs, enabling deeper insights.

In pharmacovigilance, there is a general wariness of Bayesian methods due to the potential bias introduced by informative priors in AE signal detection and strength estimation. To address this, we adopt a frequentist approach to evaluate our methods and compare them with existing Bayesian approaches. We theoretically demonstrate the superiority of our proposed estimator for the expected count from a frequentist mean squared error perspective. Additionally, we conduct extensive simulation experiments with replicated data generated from a wide range of ``true'' signal parameters. We then compare our methods to existing ones using frequentist metrics such as type I error, power, false discovery rate, sensitivity (for signal detection), and replication mean squared error (for signal strength estimation). To account for nuances like multimodality in the empirical Bayes posterior for O/E parameters, we use Wasserstein-2 distances (posterior mean squared errors) as a formal comparison metric between the computed posterior and the ``true'' signal strength parameters in each replicated dataset. The frequentist evaluation is then performed through replication-based averages of these Wasserstein-2 distances.

We note that while our approach has a distinct Bayesian feel, it also has a valid frequentist interpretation. Specifically, our method can be cast into the framework of $g$-modeling \cite{efron2014two}, which assumes a general mixture model for the observed AE-drug occurrence counts. In our approach, the observed counts are modeled using a mixture of Poisson kernels---analogous to marginal likelihood in Bayesian analysis---with flexible, non-parametric mixing densities (priors). The non-parametric mixing density is estimated using marginal maximum likelihood, and the resulting Poisson mixture model is then used for inference on the signal strength parameters. We believe this frequentist non-parametric model interpretation is particularly helpful in arguing the ``objectivity'' of our framework, which is also evidenced by our theoretical and extensive simulation results.

In addition to introducing the non-parametric empirical Bayes modeling framework for pharmacovigilance, we propose several improvements to existing approaches. Our contributions in this paper are:

\begin{enumerate}
  \item We introduce a suite of non-parametric empirical Bayesian approaches for pharmacovigilance. These methods are highly adaptable, enabling fully data-driven and computationally efficient implementation while producing interpretable inferences. The methods we discuss include KMs non-parametric discrete mixture \cite{koenker2014convex}, Efron's non-parametric $g$-mixture \cite{efron2016empirical}, and our proposed new general sparse gamma mixture models with an EM-like algorithm for fitting.
  
  \item We extend the existing MGPS method \cite{dumouchel1999bayesian} to a mixture of three gamma priors, with one component capturing structural zeros. We propose three different approaches to select initial values for the iterative optimization of MGPS and the extended MGPS model. We also provide a method to select appropriate grid values necessary to employ the KM model on pharmacovigilance SRS data.
  
  \item We suggest informative metrics to perform frequentist evaluation of Bayesian pharmacovigilance methods that acknowledge nuances in the Bayesian posterior distributions for the signal O/E parameters. Our approach employs the Wasserstein-2 distance to assess `error' in the Bayesian posterior relative to `true' values, and then uses frequentist sampling/replication averages of these distances as an informative metric of estimation error. The Wasserstein distance for our sparse general-gamma mixture method attains closed-form expressions, which we derive.
  
  \item We perform extensive simulations to compare our approaches against several other existing Bayesian and empirical Bayesian approaches. Our results suggest impressive signal detection performance---comparable to or better than all existing approaches---and signal estimation performance---notably superior to virtually all existing methods.
\end{enumerate}

\section{Motivating Dataset: statin drugs} \label{sec:statin-data-description}

Our motivating dataset used to exemplify the approaches discussed in this paper is based on SRS data for 6  statin drugs and $6,039$ AEs associated with these drugs.  Statins are a class of prescription drugs used in conjunction with diet and exercise to reduce levels of low-density lipoprotein (LDL) or ``bad'' cholesterol in the blood. The dataset focuses on the following six statin drugs (across columns): Atorvastatin, Fluvastatin, Lovastatin, Pravastatin, Rosuvastatin, and Simvastatin. In addition to these six drugs that are included in the dataset, a seventh statin drug, Cerivastatin, was previously available but was withdrawn from the United States market in 2001, and is thus not considered in our analyses. We focus on the statin dataset for several reasons:  1) statin drugs are widely prescribed; 2) statins are associated with a well-studied set of AEs linked with muscle disorders with varying severity\cite{LAW2006S52, wolfe2004dangers}; and  3) the dataset has been previously analyzed on multiple occasions to demonstrate the use of statistical pharmacovigilance and SRS data mining methods\cite{chakraborty2022use, ding2020evaluation, hu2015signal, huang2011likelihood}.  

In addition to the six statin drugs, the dataset also contains a reference category (column) of `Other drugs' obtained by collapsing drugs not belonging to the statin group as comparator. The AE-drug occurrences cataloged in the data are curated from  the FDA FAERS database for the quarters 2014Q3-2020Q4 with  $N_{\bullet \bullet} = 63,976,610$ total report counts. The dataset is available in the R package \texttt{pvLRT} \cite{pvlrt}.  In this paper, we focus on three subsets containing 42, 46, and 1,491 AEs (rows) of specific interest from all the 6,039 AEs stored in the full dataset; in each subset, we collapse the remaining AEs into a reference AE category (row) called `Other AEs.'  The resulting ``statin-42'' and ``statin-46'' datasets are provided in Appendix Section S6 in the form of contingency tables, while the ``statin1491'' dataset can be accessed from the package \texttt{pvLRT}\cite{pvlrt}.

\section{Notation and a Review of Existing Bayesian/empirical Bayesian Approaches to Pharmacovigilance}\label{sec:review}

This section introduces our notation and briefly reviews existing Bayesian and empirical Bayesian approaches to pharmacovigilance. We consider an SRS dataset cataloging reports on $I$ AEs across $J$ drugs. Let $N_{ij}$ denote the number of reported cases for the $i$-th AE and the $j$-th drug, where $i = 1, \dots, I$ and $j = 1, \dots, J$. We assume that the $J$-th drug serves as a reference drug/category---commonly referred to as ``Other drugs''---against which the preponderance of AEs in the remaining drugs ($j = 1, \dots, J-1$) is measured. Similarly, the $I$-th AE is assumed to correspond to a reference category termed ``Other AEs.'' These reference drug and AE categories often appear as natural comparators in SRS databases; if absent, they can be constructed by collapsing or grouping some of the existing AEs or drugs. We summarize these AE-drug pairwise occurrences in an $I \times J$ contingency table, where the $(i, j)$-th cell catalogs the observed count $N_{ij}$ arising from the $i$-th AE and the $j$-th drug. The row, column (marginal), and grand totals in the table are denoted as $N_{i\bullet} = \sum_{j=1}^J N_{ij}$ for $i=1,\dots,I$; $N_{\bullet j} = \sum_{i=1}^I N_{ij}$ for $j=1,\dots,J$; and $N_{\bullet \bullet} = \sum_{i=1}^I \sum_{j=1}^J N_{ij}$. Further, we denote by $E_{ij}$ the \textit{null baseline expected count}---the population-level occurrences of reports for the $(i, j)$ AE-drug pair---had there been no association or dependence between the pair $(i, j)$. In practice, $E_{ij}$ is commonly estimated/approximated by its natural estimator $N_{i \bullet} N_{\bullet j} / N_{\bullet \bullet}$, obtained through the marginal row and column proportions $N_{i \bullet} / N_{\bullet \bullet}$ and $N_{\bullet j} / N_{\bullet \bullet}$ along with the grand total $N_{\bullet \bullet}$.

Our base model assumption for the observed count $N_{ij}$ conditional on $E_{ij}$ is that
\begin{equation} \label{eqn:model-poisson-likelihood}
N_{ij} \mid E_{ij} \sim \operatorname{Poisson}(E_{ij}\lambda_{ij}),  
\end{equation}
where the parameter $\lambda_{ij} \geq 0$ represents the relative reporting ratio---\textit{the signal strength}---for the $(i, j)$-th pair, measuring the ratio of the actual expected count arising due to dependence to the null baseline expected count. Thus, $\{\lambda_{ij}\}$ are our key parameters of interest, with a large $\lambda_{ij}$ indicating a strong association between a drug and an AE. Formally, a AE-drug combination with a corresponding  $\lambda > 1$ is defined as a potential signal, while the combination is understood as a non-signal if $\lambda \leq 1$. Additionally, if $\lambda = 0$, the AE-drug combination is considered impossible to co-occur in the population and is deemed a \textit{structural zero} or a \textit{true zero}. This is different from an observed  $N_{ij} = 0$, which may occur with positive probability the under the above Poisson law even when $\lambda_{ij} > 0$.  Furthermore, the assumption regarding the baseline AE ($I$) and baseline drug level ($J$) implies that AE-drug combinations in the last row or last column of the contingency table are necessarily non-signals, with $\lambda_{iJ} = 1$ for $i = 1, \dots, I$ and $\lambda_{Ij} = 1$ for $j = 1, \dots, J$. 

In a Bayesian or empirical Bayesian framework, a prior distribution for the parameters $\{\lambda_{ij}\}$ is considered, and inference on $\{\lambda_{ij}\}$ for signal detection is made through the resulting posterior distribution of $\lambda_{ij}$ given the observed SRS data.  In this paper, we utilize the gamma distribution $\operatorname{Gamma}(\text{shape } = \ \alpha, \text{ rate =} \ \beta)$ with probability density function:
\[ 
f_{\text{gamma}}(x\mid \alpha,\beta) = \frac{\beta^{\alpha}}{\Gamma(\alpha)}x^{\alpha-1}\exp(-\beta x) \quad \text{for } x>0\ \ \alpha, \beta >0, 
\]
and mixtures of the gamma distributions to construct priors for $\lambda_{ij}$. The prior hyperparameters $\{\alpha\}$ and $\{\beta\}$ can either be pre-specified/elicited or can be inferred from the SRS data itself using point estimation (in an empirical Bayesian framework) or posterior approximation/drawing (in a full Bayesian framework) before evaluating the posterior distribution of $\lambda_{ij}$ given the observed SRS data. 

Signal detection subsequently relies on the evaluated posterior of $\{\lambda_{ij}\}$. Specifically, a pair $(i, j)$ is deemed as a signal if the corresponding posterior probability $\Pr(\lambda_{ij} \geq 1 + \varepsilon \mid \text{data})$ is obtained to be sufficiently high, e.g., being larger than $0.95$ for some small tolerance $\varepsilon > 0$ (we used $\varepsilon = 0.001$ in our numerical experiments).  Under an empirical Bayesian setup, a data-based \textit{point estimate} of the prior distribution (or equivalently, of the prior distribution hyperparameters) for  $\{\lambda_{ij}\}$  is obtained and used to derive an estimated posterior. By contrast, in a full Bayesian setup, the prior distribution hyperparameters for  $\{\lambda_{ij}\}$  are either specified beforehand, or their posterior distribution is also obtained alongside $\{\lambda_{ij}\}$, often using random simulation-based methods such as Markov chain Monte Carlo (MCMC). Below, we review existing Bayesian and empirical Bayesian frameworks for pharmacovigilance that use this or an equivalent framework.

\subsection{The single-gamma prior model}\label{sec:single-gamma}

Introduced by Huang et al. \cite{huang2013likelihood} as a simple benchmark for more complex Bayesian approaches, the single-gamma prior model assumes a common parametric gamma prior for all $\{\lambda_{ij}\}$: $\lambda_{ij} \sim \text{Gamma}(\alpha, \alpha)$, where $\alpha \in (0,1)$ is a hyper-parameter that can be either pre-specified or estimated from the data. When $\alpha$ is estimated, the joint variability of the $(N_{ij}, E_{ij})$ values across cells is moderately acknowledged, while pre-specifying $\alpha$ ignores this variability. Huang et al. \cite{huang2013likelihood} recommend setting $\alpha = 0.5$. The posterior distribution of $\lambda_{ij}$ given the data is:
$$\lambda_{ij} \mid N_{ij} \sim \text{Gamma}(\alpha + N_{ij}, \alpha + E_{ij}).$$
The gamma prior \textit{shrinks} the estimate of $\lambda_{ij}$ from its maximum likelihood estimate, $N_{ij}/E_{ij}$, reducing the risk of spurious identification of AE-drug combinations with low observed counts $N_{ij}$ and low expected counts $E_{ij}$. However, the parameter $\alpha$ plays a crucial role in determining the method's performance, particularly regarding the false discovery rate and sensitivity. A posterior probability-based false discovery adjustment method \cite{muller2006fdr} can be employed to control the false discovery rate. In our simulation in Section \ref{sec:simulation}, we incorporate this adjustment and compare the results with other methods.

\subsection{Bayesian Confidence Propagating Neural Network (BCPNN)}\label{sec:bcpnn}

Bate et al. \cite{bate1998bayesian} introduced the Bayesian Confidence Propagation Neural Network (BCPNN), utilizing the information-theoretic measure of mutual information\cite{bate1998bayesian} for AE signal detection. Unlike the relative reporting ratio-based Poisson model described above, this model is formulated using a binomial model. Let $p_{ij}$ denote the probability of the occurrence of the $(i, j)$-th cell, and let $p_{\bullet j}$ and $p_{i\bullet}$ be the marginal column and row probabilities, each with a Uniform(0, 1) prior distribution. The hierarchical model assumptions for the cell counts and marginal totals are: $N_{ij} \mid p_{ij} \sim \text{Bin}(N_{\bullet\bullet}, p_{ij})$, $N_{i\bullet} \mid p_{i\bullet} \sim \text{Bin}(N_{\bullet\bullet}, p_{i\bullet})$, $N_{\bullet j} \mid p_{\bullet j} \sim \text{Bin}(N_{\bullet\bullet}, p_{\bullet j})$; with the prior $p_{ij} \sim \text{Beta}(1,\beta_{ij})$. The prior parameters $\{\beta_{ij}\}$ are determined/estimated by setting the prior mean of $p_{ij}$ equal to the product of the posterior means of $p_{i\bullet}$ and $p_{\bullet j}$.

The strength of association between the $i$-th AE and $j$-th drug in this approach is quantified through the information component\cite{bate1998bayesian}---the term measuring the contribution of the $(i, j)$-th pair in the expression for mutual information between the AE and drug categories in the contingency table:
\[\text{IC}_{ij} = \log_2 \frac{p_{ij}}{p_{i\bullet} p_{\bullet j}}.\]
Inferences on $\{\text{IC}_{ij}\}$ are based on an asymptotic normal approximation of their posterior distributions, characterized by the estimated asymptotic mean and variance as follows\cite{bate1998bayesian}:
$$\hat{\text{E}}(\text{IC}_{ij}) \approx \log_2 \frac{(N_{ij}+1)(N_{\bullet\bullet}+2)^2}{(N_{\bullet\bullet}+\hat{\beta}_{ij})(N_{i\bullet}+1)(N_{\bullet j} +1)},$$
$$\hat{\text{Var}}(\text{IC}_{ij}) \approx \frac{1}{(\log 2)^2}\left(\frac{N_{\bullet\bullet}-N_{ij}+\hat{\beta}_{ij} -1}{(N_{ij}+1)(1+N_{\bullet\bullet}+\hat{\beta}_{ij})} + \frac{N_{\bullet\bullet}-N_{i\bullet}+1}{(N_{i\bullet}+1)(N_{\bullet\bullet}+3)}+\frac{N_{\bullet\bullet}-N_{\bullet j}+1}{(N_{\bullet j}+1)(N_{\bullet\bullet}+3)}\right).$$
Similar to the single-gamma model with a pre-specified $\alpha$, the posterior inference from BCPNN does not account for the joint variability among all O/E values. Moreover, the model may also fail to control the FDR unless a posterior probability-based adjustment, as suggested for the single-gamma prior model above, is employed. Our simulation experiments in Section \ref{sec:simulation} utilize this FDR adjustment.

\subsection{The two-component gamma (2-gamma) and two-component-gamma-zi (2-gamma-zi) mixture prior models}\label{sec:2-gamma}

The single-gamma model described above makes highly restrictive parametric assumptions that may lead to significant model violations in applications, resulting in poor performance. To increase flexibility, several generalizations have been proposed, the most notable being the Multi-item Gamma Poisson Shrinker (MGPS) model by DuMouchel (1999) \cite{dumouchel1999bayesian}. This model assumes a semi-parametric mixture of two gamma distribution components as a prior (called a ``2-gamma'' prior henceforth) for $\{\lambda_{ij}\}$:
\[
g(\lambda \mid \alpha_1, \beta_1, \alpha_2, \beta_2, \omega) = \omega f_{\operatorname{gamma}}(\lambda \mid \alpha_1, \beta_1) + (1-\omega)f_{\operatorname{gamma}}(\lambda \mid \alpha_2, \beta_2),
\]
where $\alpha_j > 0$, $\beta_j > 0$, for $j=1, 2$ are the mixture component-specific prior parameters, $\omega \in [0, 1]$ is the prior mixing weight, and $f_{\operatorname{gamma}}(\bullet \mid \alpha,\beta)$ is the probability density function of a gamma distribution with parameters $\alpha$ and $\beta$. As is commonly observed in many SRS databases, most AE-drug combinations $\{(i, j)\}$ are non-signals with little to no association---implying that the underlying values of the corresponding signal parameters $\lambda_{ij}$ are close to 1. In contrast, a few signal pairs exhibit substantially large associations with $\lambda_{ij} > 1$. This pattern can be well captured by the above mixture of two-component gamma distributions: one gamma component with a high mixing probability centering around the value 1 for the non-signal $\{\lambda_{ij}\}$ values,  while the other component describing all signal $\{\lambda_{ij} > 1\}$ values---provided there is not substantial heterogeneity in the $\lambda_{ij}$ values within each component. 

If there is zero inflation among the observed counts, potentially due to the presence of structural zero AE-drug pairs, they may be rigorously handled in this setup through a straightforward extension of the prior model. This entails creating an additional ``zero-inflation'' component $\lambda_{ij}$ to accommodate small $\lambda_{ij} \approx 0$ values: 
\[
g(\lambda \mid \alpha_1, \beta_1, \alpha_2, \beta_2, \omega_1, \omega_2) = \omega_1 f_{\operatorname{gamma}}(\lambda \mid \alpha_1, \beta_1) + \omega_2 f_{\operatorname{gamma}}(\lambda \mid \alpha_2, \beta_2) + 
(1 - \omega_1 - \omega_2) f_{\operatorname{gamma}, \operatorname{zi}}(\lambda \mid \alpha_{\text{zi}}, \beta_{\text{zi}}),
\]
We will call the resulting model a two-component gamma zero-inflation (''2-gamma-zi'') mixture prior model henceforth to distinguish from the 2-gamma prior model introduced above. In this model, the extra gamma component $f_{\operatorname{gamma}, \operatorname{zi}}$ is elicited with \textit{pre-specified} small shape $\alpha_{\text{zi}}$ and rate $\beta_{\text{zi}}$ parameters, such that the prior mean and variance of the component are both close to $0$, thus allowing the entire density to concentrate around zero.

An empirical Bayes approach is suggested in DuMouchel (1999) \cite{dumouchel1999bayesian} for model estimation.  This entails maximizing the marginal (i.e., $\lambda_{ij}$-integrated) likelihood of the prior parameters $\{\alpha_1, \beta_1, \alpha_2, \beta_2, \omega\}$ given observed counts $\{N_{ij}\}$ to produce their maximum marginal likelihood (point) estimates $\{\hat \alpha_1, \hat \beta_1, \hat \alpha_2, \hat \beta_2, \hat \omega\}$. Subsequently the \textit{estimated} posterior distribution of $\lambda_{ij}$ conditional on $N_{ij}$ and $\{\hat \alpha_1, \hat \beta_1, \hat \alpha_2, \hat \beta_2, \hat \omega\}$--which has an analogous gamma mixture shape--is obtained for inference (signal detection) on $\lambda_{ij}$. Under the Poisson model \eqref{eqn:model-poisson-likelihood}, the $\lambda_{ij}$-integrated marginal probability mass function of $N_{ij}$ is a mixture of two negative binomial distributions:
\[
p(N_{ij} = n_{ij} \mid \alpha_1, \beta_1, \alpha_2, \beta_2, \omega) 
= \omega f_{\text{NB}}(n_{ij} \mid \alpha_1,\beta_1, E_{ij}) + (1-\omega) f_{\text{NB}}(n_{ij} \mid \alpha_2,\beta_2, E_{ij}),
\]
where $f_{\text{NB}}(n \mid \alpha,\beta, E) = \frac{\Gamma(n+\alpha)}{\Gamma(\alpha)n!}\left( \frac{\beta}{E+\beta} \right) ^{n} \left(\frac{E}{E+\beta} \right) ^{\alpha}$; $n=1, 2, \dots$, is the negative binomial probability mass function. The parameters $(\alpha_1, \beta_1, \alpha_2, \beta_2, \omega)$ are estimated by maximizing their marginal likelihood:
\[
(\hat \alpha_1, \hat \beta_1, \hat \alpha_2, \hat \beta_2, \hat \omega) = \argmax_{\alpha_1, \beta_1, \alpha_2, \beta_2, \omega} \prod_{i}\prod_j p(N_{ij} = n_{ij} \mid \alpha_1, \beta_1, \alpha_2, \beta_2, \omega).
\]
Closed-form analytic expressions for $(\hat \alpha_1, \hat \beta_1, \hat \alpha_2, \hat \beta_2, \hat \omega)$ are not available; instead, numerical approaches are used for the maximization. The R package ``openEBGM'' \cite{openEBGM} provides an open-source implementation of the 2-gamma model. Implementation of the 2-gamma-zi model, which has a similar mixture negative binomial structure but includes a third ``zero-inflation'' component with an unknown mixing weight parameter (to be estimated) and pre-specified shape and rate parameters is operationally similar.

The 2-gamma and 2-gamma-zi models can permit flexible signal detection and signal strength estimation in applications, provided that the groups comprising the $\lambda_{ij}$-parameter values for the signal cells, non-signal cells, and structural zero cells (for the 2-gamma-zi model) are not too heterogeneous within themselves, allowing each group to be adequately explained by a single (separate) gamma distribution. Care is needed for effective numerical estimation of the prior parameters, as the objective function (marginal likelihood) described above is non-concave. Consequently, the final numerical estimates may heavily depend on their initial values; however, there is no clear method for selecting appropriate initial values for the parameter estimates, and ad hoc choices may lead to unstable/sub-optimal results, as we observed in our experimentation with the openEBGM package. To address this, we have developed two approaches to generate reasonable starting points for the above maximization based on the method of moment estimation and $k$-means clustering; see Appendix S2 for details.

\subsection{Hierarchical Dirichlet Process (HDP) model}\label{sec:HDP}


A final generalization of the two-component gamma mixture prior model described above was introduced by Hu et al.\cite{hu2015signal} in the framework of a non-parametric Bayesian model, leveraging a hierarchical Dirichlet process mixture of gamma distributions. The resulting model induces an infinite mixture of gamma component distributions for $\{\lambda_{ij}\}$ with appropriately articulated diminishing (under expectation) prior mixing weights over components. This allows the model to flexibly accommodate multiple subgroups or `subclusters' for the $\lambda_{ij}$ values---even within the signal, non-signal, and zero-inflation groups---while still keeping the total number of effective or non-empty components relatively small---in a rigorous data-driven way. The model\cite{hu2015signal} is described as follows. For each drug $j$,
\begin{gather*}
    \lambda_{ij}\mid G_j \sim G_j \equiv \operatorname{DP}(\rho_j, G_{0j}), \\
    G_{0j} \equiv \operatorname{Gamma}(\alpha_j, \alpha_j), \\
    \alpha_j \sim \operatorname{Uniform}(0,1), \\
    \rho_j\sim \operatorname{Uniform}(0.2,10).
\end{gather*}
Here $\operatorname{DP}(\rho_j, G_{0j})$ denotes a dirichlet process with precision parameter $\rho_j$ and baseline distribution $G_{0j}$ defined as a unit-mean gamma distribution of the form $\operatorname{Gamma}(\alpha_j, \alpha_j)$.  For implementation, the authors consider a stick-breaking representation of the Dirichlet process, which characterizes the prior $G_j$ as an infinite mixture of gamma distributions of the form:
\begin{gather*}
  G_j = \sum_{l=1}^{\infty} w_l \delta_{\{\theta_l\}}; \quad \theta_l\sim G_{0j};  \\
  w_1 = v_1,\ w_l = v_l \prod_{l'=1}^{l-1}(1 - v_{l'}) \ \text{with} \ v_l \sim \operatorname{Beta}(1,\rho_j) \ \text{for} \ l = 1, 2, \dots.
\end{gather*}
Here $\delta_{\{\theta\}}$ denotes the degenerate distribution with a point mass at $\theta$.  In applications, the above infinite mixture can be approximated by a finite mixture of $K$ components for some large $K$, which is the approach taken by Hu et al. For implementation of the model, the authors consider Markov chain Monte Carlo (MCMC) sampling from the posterior distribution. This enables flexible estimation of the model but requires substantially heavy computations whose convergence can be difficult to justify in applications. Our method proposed in the following section is inspired by this model; however, we leverage the empirical Bayes framework for implementation to aid substantial computation gains. Section \ref{sec:real-data-analysis} provides a comparison of computation times for various models, including the HDP model and our non-parametric empirical Bayes model.

\section{Flexible Non-Parametric Empirical Bayes Models for Pharmacovigilance} \label{sec:np-empirical-bayes}

This section introduces a non-parametric empirical Bayes framework to flexibly estimate $\{\lambda_{ij}\}$ within the Poisson model \eqref{eqn:model-poisson-likelihood}. As an extension of the parametric empirical Bayes, non-parametric empirical Bayes methods use priors that are not restricted to a specific parametric form. Instead, these priors are often represented as general mixtures of parametric distributions for model parameters and take a data-driven approach to estimating the prior from observed data. In the following, we describe a general structure for non-parametric empirical Bayes methods for the Poisson model \eqref{eqn:model-poisson-likelihood}.

Let $g$ be a prior density function for signal strength parameters for all AE-drug pairs: $\lambda_{ij} \sim g$. Then, in the context of the Poisson model \eqref{eqn:model-poisson-likelihood}, the marginal probability mass function of $N_{ij}$ is given by:
\[
p(N_{ij}) = \int_0^{\infty} g(\lambda_{ij}) \ f_{\text{pois}}(N_{ij} \mid \lambda_{ij}E_{ij}) \ d\lambda_{ij},
\]
where $f_{\text{pois}}(N \mid \lambda)$ is the probability mass function of a Poisson random variable with mean $\lambda$ evaluated at $N$. Under the empirical Bayes framework, the prior distribution is consequently estimated from the data by maximizing the log marginal likelihood:
\[
\hat g = \mathop{\arg \max}\limits_{g} \sum_{i=1}^I \sum_{j=1}^J \log p(N_{ij}).
\]
Then, the estimated posterior density of $\lambda$ given $N_{ij}$ is:
\[
\hat{\text{p}}(\lambda \mid N_{ij}) = \frac{\hat g(\lambda) f_{\text{pois}}(N_{ij} \mid \lambda E_{ij})}{\hat{\text{p}}(N_{ij})},
\]
where $\hat{\text{p}}(N_{ij}) = \int_0^{\infty} \hat g(\lambda_{ij})f_{\text{pois}}(N_{ij}\mid \lambda_{ij}E_{ij}) \ d\lambda_{ij}$.

We note that the discussion above makes no specific assumption about the structure of $g$. By letting $g$ have a flexible non-parametric structure---typically expressed as an arbitrary mixture of parametric densities---we may obtain rigorous (empirical) Bayesian inference on $\{\lambda_{ij}\}$ without the prior $g$ overshadowing the information provided by the SRS data table. The mixture structure of $g$ permits multiple modes/clusters in the underlying distribution and therefore can flexibly accommodate multiple subgroups of signal, non-signal, or zero-inflation $\{\lambda_{ij}\}$ values. Care is needed, however, in determining the parametric form for the mixture component densities and the number of mixture components. This is because, depending on the form of mixture density used, too large a number of mixture components would provide inadequate regularization needed for stable posterior estimation of $\{ \lambda_{ij} \}$ particularly in high dimensions (large $I$ and/or $J$; see Figure 1 and the associated discussion in Introduction); thus producing noisy inference. Below, we discuss four non-parametric empirical Bayesian methods with different choices of priors that appropriately balance this flexibility and regularization and are generally applicable to pharmacovigilance problems.

\subsection{The Koenker and Mizera (KM) approach} \label{sec:KM}

The Koenker and Mizera non-parametric empirical Bayes method (2014)\cite{koenker2014convex} for Poisson model \eqref{eqn:model-poisson-likelihood}

assumes the prior $g$ for $\{\lambda_{ij}\}$ to have a finite discrete support of size $K$, $\lambda \in \{ v_1,...,v_K\}$, $K<\infty$, with associated prior probability masses: $\{ g_1,...,g_K\}$ with $g_k \geq 0$; $\sum_{k=1}^K g_k = 1$. No further assumption is made on $\{g_k\}$. Under this prior assumption, the marginal probability mass for the $(i,j)$-th observation is obtained as:
\[
f_{ij} = \sum_{k=1}^K f_{\text{pois}}(N_{ij} \mid  v_k E_{ij}) \times g_k.
\]
For empirical Bayes estimation, the parameters $\{ g_1, \dots, g_K \}$ are estimated through the following constrained maximization problem:
\[
\max \left\{ \sum_{i=1}^I\sum_{j=1}^J\log(f_{ij}): \quad f_{ij} = \sum_{k=1}^K f_{\text{pois}}(N_{ij}\mid  v_k E_{ij})\times g_k,\quad g_k \geq 0, \quad \sum_{k=1}^K g_k=1\right \}.
\] 
This is a non-convex optimization problem\cite{kiefer1956consistency} that does not have a guaranteed global solution and is computationally highly inefficient. However, Koenker and Mizera\cite{koenker2014convex} identify and leverage the convexity of its dual problem, and propose an efficient estimation strategy using existing convex optimization software (Mosek\cite{mosek}).

The resulting non-parametric empirical Bayes approach for the Poisson model \eqref{eqn:model-poisson-likelihood} is referred to as the KM approach hereinafter. This method is highly flexible; however, its performance depends on the grid values $\{ v_1, \ldots, v_K \}$, which are not straightforward to determine a priori. In our applications, we employ a histogram-based approach to generate reasonable grid values based on the computed $\tilde{\lambda}_{ij} = \max \{N_{ij}/E_{ij}, \epsilon\}$, for some small, prespecified $\epsilon > 0$. Specifically, given a grid size $K$, we derive a histogram estimator based on the computed $\{\tilde{\lambda}_{ij}\}$ and generate a grid $\{ v_1, \ldots, v_K \}$ by random sampling from the histogram estimator (performed on the log-transformed $\{\tilde{\lambda}_{ij}\}$). This approach generates a reasonable grid to use as the discrete support for the prior $g$. However, the grid size $K$ also needs to be judiciously chosen. A small $K$ would concentrate a significant mass of the estimated prior $\hat{g}$ on only a handful of grid points, causing the posteriors of $\{\lambda_{ij}\}$ to also be concentrated around these grid values. While this may still allow adequate signal discovery in some applications, inference on $\{\lambda_{ij}\}$ could suffer due to large jumps in their estimated posterior distribution functions. As a remedy, a large grid size of $K$ should be used, yet there is no general way to determine the optimal size of $K$. Additionally, when $K$ is very large, the computational demands for the optimization may become extremely high, rendering the implementation practically challenging, if not infeasible.

\subsection{The Efron approach} \label{sec:efron}

Similar to KM, the Efron non-parametric empirical Bayes\cite{efron2016empirical} approach also considers a finite discrete support of size $K$: $\lambda \in \{ v_1, \ldots, v_K \}$, $K < \infty$. However, the associated prior probability masses are assumed to have an exponential form: 
\[
g = g(\alpha) = \exp\{ Q \alpha - \phi(\alpha) \},
\]
parametrized by a $p$-dimensional parameter vector $\alpha$, a known $K \times p$ structure matrix $Q$, and an appropriately determined normalizing constant $\phi > 0$ that makes $g$ a proper mass function. The default choice for $Q$ in Efron\cite{efron2016empirical} is considered to be a natural spline for $\{ v_1, \ldots, v_K \}$ with $p$ degrees of freedom. The exponential structure in $g$ helps stabilize its estimation by reducing the problem to that of estimating $\alpha$ from its ridge-penalized log marginal likelihood. Specifically,  
\[
\hat \alpha = \argmax_\alpha \left\{ \sum_{i=1}^I\sum_{j=1}^J \log P_{ij}^T g(\alpha) - c_0\left( \sum_{l=1}^p \alpha_l^2\right)^{1/2} \right\},
\]
where $P_{ij} = (f_{\text{pois}}(N_{ij} \mid v_k E_{ij}): k= 1, \dots, K)$. This optimization problem can be easily solved by a gradient-based approach; Efron suggests a Fisher-scoring type approach\cite{efron2016empirical}. Similar to KM, the Efron approach also requires appropriate specification of the grid values $\{ v_1, \ldots, v_K \}$ as the support of $g$. In our applications, we consider the same histogram-based grid generation technique as discussed in section \ref{sec:KM}.

\subsection{The $K$-component gamma mixture prior ($K$-gamma) model with a pre-specified $K$} \label{sec:K-gamma}
As an alternative formulation avoiding the discrete nature of the priors considered in the KM and the Efron approaches, we may consider a continuous prior density $g$ defined by a mixture of $K$ continuous densities for a prespecified $K \geq 3$. Retaining the analogy with the 2-gamma model (see Section \ref{sec:2-gamma}), we use component-specific gamma densities. The consequent mixture density is defined as:
\[
g(\lambda\mid R, H, \Omega) = \sum_{k=1}^K \omega_k f_{\text{gamma}} \left(\lambda\mid \alpha = r_k, \beta = \frac{1}{h_k}\right),
\]
where $\Omega = \{\omega_1, \dots, \omega_K\}$, $R = \{r_1, \dots, r_K\}$, and $H = \{h_1, \dots, h_K\}$ are component-specific parameters. The derivation of the marginal likelihood and the posterior distribution is analogous to the 2-gamma and 2-gamma-zi prior models that we discussed in Section \ref{sec:2-gamma}. The resulting approach will be called the $K$-gamma mixture prior modeling approach henceforth. 

Introducing $K \geq 3$ gamma mixture components in the prior model for $g$ enhances its flexibility to capture complex structures in the SRS data more faithfully, as compared to the 2-gamma and 2-gamma-zi prior models, provided an adequate number of components $K$ is used. Specifically, $K$ needs to match the number of distinct sub-groups/sub-clusters present in the 'true' prior distribution of the $\{\lambda_{ij}\}$ when explained as a mixture of gamma distributions. An incorrect $K$ may lead to underfitting or overfitting, making the model either too restrictive or too noisy, respectively. In applications where some knowledge is available on $K$, this approach may be used. When $K$ is entirely unknown---as often is the case in practice---the general-gamma mixture model, as proposed below, should be used.

\subsection{Proposed approach: the general-gamma mixture prior model with an improved estimator of null baseline expected counts} \label{sec:general-gamma}

We make a two-fold improvement over the $K$-gamma mixture prior model with a pre-specified $K$ described in Section~3.3 above by first adaptively determining a reasonable $K$ from the data and then proposing and utilizing an improved estimator of the null baseline counts $\{E_{ij}\}$. Specifically, we first extend the framework of the $K$-gamma mixture prior model to handle situations where no information on the number of components $K$ is available a priori and thus must be inferred from the data. To achieve this, we employ the framework of sparse finite mixture models \cite{fruhwirth2019here}\cite{malsiner2016model}\cite{malsiner2017identifying} that determines an optimal number of \textit{active} mixture components in a data-driven way. In particular, it begins with an overfitted mixture model with a large $K$---e.g., set to $100$ in our computations---and assigns an appropriately defined Dirichlet hierarchical prior distribution on the mixture weight parameters: $\Omega = \{\omega_1, \dots, \omega_K\} \sim \text{Dirichlet}(\alpha,\alpha, \dots, \alpha)$ with $\alpha < 1$. 

This specific choice of $\alpha \in (0, 1)$ for the Dirichlet prior asymptotically (in $IJ$) ensures that the consequent fitted mixture model would be sparse---with many \textit{inactive} components (with zero mixture weights) and will have the number of \textit{active} components (with positive mixture weights) equal to the number of distinct mixture components in the population. More precisely, under certain regularity conditions, the sparse overfitted mixture is known to asymptotically converge to the ``true'' population mixture (Rousseau and Mengersen (2011)\cite{rousseau2011asymptotic}) provided the hyperparameters in the Dirichlet prior $\alpha$ are smaller than $d/2$, where $d$ is the dimension of the component-specific parameter $\{R, H \}$ ($d = 2$ in our setting). Because the initial $K$ can be very large in practice, a customized estimation strategy is needed for a successful implementation of the model. We propose a bi-level expectation conditional maximization (ECM \cite{meng1993maximum}) algorithm for an efficient implementation in Section \ref{sec:general_gamma_implementation} below.    

Additionally, we propose (Section \ref{sec:estimating-Eij}) an improved estimator of the null baseline expected count parameter $E_{ij}$ and employ it in our model. To our knowledge, all existing approaches consider $E_{ij}$ to be a known value fixed at its natural estimator $N_{i \bullet} N_{\bullet j}/N_{\bullet \bullet}$ and treat it as a pre-specified offset term in the Poisson model. While this may permit reasonable signal discovery (detecting if $\lambda_{ij} > 1$), we show that by treating $E_{ij}$ as a parameter and subsequently using an improved estimate thereof, substantial gain may be achieved for estimating $\lambda_{ij}$, particularly when some of the underlying $\lambda_{ij} \gg 1$. Our estimator stems from a more controlled variance-bias trade-off consideration and exploits the assumed structure of the reference AEs (rows) and reference drugs (columns), each of which is postulated to be a non-signal, i.e., has $\lambda = 1$.

\subsubsection{Implementation via a novel bi-level Expectation Conditional Maximization (ECM) algorithm} \label{sec:general_gamma_implementation}

Under the Poisson model (\ref{eqn:model-poisson-likelihood}), and given the mixture weight parameters $\Omega = \{\omega_1,...,\omega_K\}$, the marginal probability mass function for $N_{ij}$ has the following mixture of negative binomial shape:
\[
p(N_{ij} \mid \omega_1, \dots, \omega_K) = \sum_{k = 1}^K \omega_k f_{\NB}\left(N_{ij} \mid r_k, \ \theta_{ijk} = \frac{1}{1 + E_{ij} h_k}\right).
\]
We provide a fast and stable two-level ECM algorithm to estimate the parameters $\{ \omega_k, r_k, h_k\}$ by maximizing their joint marginal likelihood. Our algorithm leverages two different types of data augmentation in its two levels to aid the ECM estimation. The first level augments latent mixture component indicators as auxiliary data: $S = \{S_{11}, S_{12}, \dots, S_{IJ}\}$, $S_{ij} = (S_{ijk}: k = 1, \dots, K)$ with $S_{ijk} = 1$ if observation $(i, j)$ belongs to component $k$ and is $0$ otherwise, to estimate $\{ \omega_k, r_k \}$. The second level utilizes a Poisson-logarithmic series representation of negative binomial random variables\cite{quenouille1949relation} and employs augmented variables $M$ and $Y$ denoting latent Poisson variables and logarithmic series elements, respectively, along with $S$ to estimate $\{r_k\}$. Together, these two data augmentation steps produce two separate expected complete data log-likelihood (``$Q$'') functions, which we use to update $\{\omega_k, r_k\}$ and $\{h_k\}$ respectively for iterative conditional maximization (``CM steps'' of the algorithm). A common expectation (``E'') step is employed to iteratively update the two complete data log-likelihood functions. Algorithm \ref{alg:algorithm1} below shows the steps involved. Detailed derivations for these steps are provided in Appendix S4.  Similar to an ordinary single-level ECM algorithm\cite{meng1993maximum}, our proposed bi-level ECM algorithm is guaranteed to converge under standard regularity conditions; we formally establish this in Theorem \ref{thm:thm1} below.

\begin{algorithm}[ht] 
\caption{The $(u+1)$-th iteration of the proposed bi-level ECM algorithm for implementation of the general-gamma model \eqref{sec:general-gamma}.} \label{alg:algorithm1}
\begin{algorithmic}
\Require Current iteration of $\phi^{(u)} = \{ \Omega^{(u)}, R^{(u)},H^{(u)}\}$. Do for all $k = 1, \dots, K$: \\
\textbf{E step}: Compute $\tau_{ijk}^{(u+1)} = E(\mathbbm{1}_{\{ S_{ijk}=1\}}\mid N,\phi^{(u)})$, $\delta_{ijk}^{(u+1)}=E(M_{ij}\mid S_{ijk}=1,\phi^{(u)})$. \\
\textbf{CM step 1}: Given $h^{(u)}_k$, update $\omega^{(u+1)}_k$ and $r^{(u+1)}_k$ as follows: \\
$$\omega_k^{(u+1)} = \max\left\{0,\frac{\alpha-1 + \sum_{i=1}^I\sum_{j=1}^J\tau_{ijk}^{(u+1)}}{I*J+K(\alpha-1)}\right\}$$
$$r_k^{(u+1)} = \frac{\sum_{i=1}^I\sum_{j=1}^J\tau_{ijk}^{(u+1)}\delta_{ijk}^{(u+1)}}{\sum_{t=1}^T\tau_{ijk}^{(u+1)}\log\theta_{ijk}^{(u)}}$$
\textbf{CM step 2}: Given $r^{(u+1)}_k$ update $h^{(u+1)}_k$ by solving the following equation:
\[
\sum_{i=1}^I\sum_{j=1}^J \tau_{ijk}^{(u+1)}\left[ \frac{N_{ij}}{h_k} - \frac{E_{ij}(N_{ij}+r_k^{(u+1)})}{1+E_{ij} h_k}\right] = 0.
\]
A simple iterative process for solving this equation is provided in Appendix S3.1. 
\end{algorithmic}
\end{algorithm}

\begin{theorem} \label{thm:thm1}
    Algorithm~\ref{alg:algorithm1} is guaranteed to converge to a (local) maxima of the marginal likelihood for $\{\omega_k, r_k, h_k: k=1, \dots, K\}$ for any $K$ and any Dirichlet prior hyper-parameter $\alpha \in (0, 1)$. 
\end{theorem}

\begin{proof}
We only provide a brief outline of the proof here; detailed arguments are provided in Appendix S3.2 We show that the expected conditional log-likelihoods $Q_1$ and $Q_2$ for the two levels both monotonically increase as the iteration progresses for any $K$ and $\alpha \in (0, 1)$. From this, it follows (Appendix Lemmas 1 and 2) that the target log marginal likelihood for the parameters $\{\omega_k, r_k, h_k: k=1, \dots, K\}$ given the observed data also increases as iteration progresses, which completes the proof. 
\end{proof}

\paragraph{Initialization of model parameters and selection of hyperparameters} 

Care is needed for the selection of the initial values of the model parameters to ensure reasonable convergence. We initialize our algorithm with an over-fitted model with a large $K$; e.g., $K = 100$ in our examples with $IJ ~ 400$. Owing the Dirichlet prior structure for $\omega_k$ with concentration parameter $0 < \alpha < 1$, certain mixture components start to vanish (i.e., get mixture weights equal to zero) as iteration progresses and, upon convergence, produce a sparse fitted mixture model. Initialization for parameters $\{ r_k, h_k: k=1, \dots, K\}$ utilizes the grid value generating process that we proposed for the discrete non-parametric empirical Bayes methods (KM and Efron). Specifically, given an initial large $K$, we generate the grid values $\{v_1, \dots, v_{K}\}$. We then choose $\{r_k\}$ and $\{h_k\}$ such that the consequent gamma component distributions individually concentrate around each of these grid values $\{v_k\}$. This is achieved by setting $r_k h_k = v_k$ and $r_k h_k^2 = \epsilon$ for some small $\epsilon > 0$ such as $\epsilon = 10^{-6}$. For $\{ \omega_k\}$ we consider a uniform initialization, i.e., $\{ \omega_k = 1/K: k=1,\dots, K\}$. 

We note that the model can be sensitive to the choice of the Dirichlet hyperparameter $\alpha$: specifically, the smaller $\alpha$ is, the fewer the number of non-empty components in the fitted mixture. An appropriate sensitivity analysis is therefore needed to select a reasonable $\alpha$. We suggest using the approximate leave-one-out cross-validation information criterion (LOOIC \cite{vehtari2017practical}) to determine an optimal $\alpha$ for estimation. This entails fitting the model with several different choices of $\alpha$, then evaluating LOOIC for each fit, and finally selecting the model with an optimal LOOIC; see Appendix S3.3 for more details.

\subsection{Improved estimation of the expected null baseline count $E$} \label{sec:estimating-Eij}

We note that the null baseline expected count $E_{ij}$ under the independence of AE-$i$ and drug-$j$ is a model parameter whose natural estimator $\hat E_{ij} = N_{\bullet \bullet} \frac{N_{i\bullet}}{N_{\bullet \bullet}}\frac{N_{\bullet}j}{N_{\bullet \bullet}}= \frac{N_{i\bullet} N_{\bullet j}}{N_{\bullet \bullet}}$ is used as a plug-in value in virtually every SRS data mining method\cite{huang2013likelihood, chakraborty2022use, huang2017zero, hu2015signal, huang2011likelihood, dumouchel1999bayesian, dumouchel2001empirical}. This natural estimator can be derived as the maximum likelihood estimator of $E_{ij}$ under a variety of model assumptions on $N_{ij}$ under the assumption of independence of the AEs and the drugs, including the Poisson model \eqref{eqn:model-poisson-likelihood} when all $\lambda_{ij} = 1$. While this simple estimator can provide reasonable signal detection, as we show in this section, it does not always yield good estimation of $\lambda_{ij}$, particularly when the underlying $\lambda_{ij}$ are ``large'' (much larger than 1; see Theorem \ref{thm2:Eestimator}). In such cases, $\hat E_{ij}$ becomes highly biased for estimating $E_{ij}$, leading to high mean squared error. Instead, we propose an improved estimator of $E_{ij}$: 
\[
\tilde E_{ij} = \frac{N_{iJ}N_{Ij}}{N_{IJ}},
\]
which aids a better control over the bias-variance tradeoff, leading to a smaller mean squared error than that of $\hat E_{ij}$. 

Utilizing independence, we write $E_{ij} = \tilde N p_{i\bullet}^* p_{\bullet j}^*$, where $p_{i\bullet}^*$ and $p_{\bullet j}^*$ are the marginal probabilities of AE-$i$ drug-$j$, respectively, and $N_{\bullet \bullet}$ is the total number of report counts (grand total) in the observed SRS dataset. With this form for $E_{ij}$ and conditional on $N_{\bullet \bullet}$, the Poisson model \eqref{eqn:model-poisson-likelihood} produces the following conditional multinomial model for $N_{ij}$:
\begin{equation} \label{eqn:multinomial_model}
    N_{ij} \mid N_{\bullet \bullet} \sim \text{Multinomial}\left(N_{\bullet \bullet}, p_{ij} = \frac{\lambda_{ij}p_{i\bullet}^*p^*_{\bullet j}}{\sum_{k=1}^I\sum_{l=1}^J \lambda_{kl}p_{k\bullet}^*p_{\bullet l}^*}\right).
\end{equation}

Using the asymptotic normality of the above multinomial distribution, straightforward analysis shows that the asymptotic mean squared error (AMSE) for $\hat E_{ij}/E_{ij}$ obtained by the delta method for large $N_{\bullet \bullet}$ is given by:
\begin{equation}  \label{eqn:AMSE_hat}
\text{AMSE}\left( \frac{\hat E_{ij}}{E_{ij}} \right)\\ 
= \frac{\bar \lambda_{i\bullet} \bar \lambda_{\bullet j}}{\bar \lambda_{\bullet \bullet}^4N_{\bullet \bullet}}\left[\bar \lambda_{\bullet \bullet}\left( 2\lambda_{ij}+ \frac{\bar \lambda_{i \bullet}}{p_{i\bullet}^*} + \frac{\bar \lambda_{\bullet j}}{p_{\bullet j}^*}\right) - 4\bar \lambda_{i \bullet}\bar \lambda_{\bullet j} \right] + \left[ \frac{\bar \lambda_{i \bullet}\bar \lambda_{\bullet j}}{\bar \lambda_{\bullet \bullet}^2} -1\right]^2,\\
\end{equation}
where $\bar \lambda_{i \bullet} = \sum_{l=1}^J\lambda_{il}p_{\bullet l}^*$, $\bar \lambda_{\bullet j} = \sum_{k=1}^I\lambda_{kj}p_{k\bullet}^*$ and $\bar \lambda_{\bullet \bullet} = \sum_{k=1}^I\sum_{l=1}^J\lambda_{kl}p_{k\bullet}^*p_{\bullet l}^*$. By the definition of $\bar \lambda_{\bullet \bullet}$, we have: $\bar \lambda_{\bullet \bullet} < \max\{\lambda_{ij}: i = 1,\dots, I-1; j = 1, \dots, J-1\} = \lambda_{\max} < \infty$. Similarly, the AMSE for the $\tilde E_{ij}/E_{ij}$ for large $N_{\bullet \bullet}$ is:
\begin{equation} \label{eqn:AMSE_tilde}
\text{AMSE} \left(\frac{\tilde E_{ij}}{E_{ij}} \right) =\frac{1}{N_{\bullet \bullet}}\left[\frac{1}{p_{I\bullet}^*p_{\bullet J}^*\bar \lambda_{\bullet \bullet}}+\frac{1}{p_{i\bullet}^*p_{\bullet J}^*\bar \lambda_{\bullet \bullet}}+\frac{1}{p_{I\bullet}^*p_{\bullet j}^*\bar \lambda_{\bullet \bullet}}-\frac{1}{\bar \lambda_{\bullet \bullet}^2} \right] +\left(\frac{1}{\bar \lambda_{\bullet \bullet}}-1\right)^2. \\
\end{equation}

The details of the steps in deriving these AMSEs are provided in Appendix S1. Theorem \ref{thm2:Eestimator} and the Corollaries below guarantee that when the total report counts $N_{\bullet \bullet}$ is large enough, ensuring the asymptotic properties hold reliably, the estimation error for $\tilde E_{ij}$ is smaller than that of $\hat E_{ij}$ provided at least one of the underlying $\lambda_{ij}$ values is sufficiently larger than one, which is precisely the situation that practical pharmacovigilance encounters.  When all $\lambda_{ij} = 1$ for all $i, j$, i.e., there is no signal, and all drugs and AEs are independent then $\text{AMSE}\left( \frac{\hat E_{ij}}{E_{ij}} \right)<\text{AMSE}\left( \frac{\tilde E_{ij}}{E_{ij}} \right)$.  However, this null baseline situation is rarely, if ever, expected to be observed in practice, and virtually all relevant large-scale SRS datasets are expected to harbor at least one signal, indicating potential superiority of the $\tilde E_{ij}$ estimator over $E_{ij}$ for such datasets. 

\begin{theorem} \label{thm2:Eestimator}
Suppose the following regularity conditions hold: 
\begin{enumerate}[label=(\Alph*)]
    \item The underlying data-generating model \eqref{eqn:model-poisson-likelihood} for the AE-drug report counts $N_{ij}$ stored in an $I \times J$ contingency table holds for some $1 \leq \lambda_{ij} < \infty$ for all $i = 1, \dots, I$, and $j = 1, \dots, J$. 
    
    \item There exist a comparator reference AE ($I$-th row) and a comparator reference drug ($J$-th column) such that $\lambda_{iJ} = 1$ for all $i = 1, \dots, I$ and $\lambda_{Ij} = 1$ for all $j = 1, \dots, J$.
    
    \item The total number of AEs $I$ and the total number of drugs $J$ are fixed and do not change with the grand total count $N_{\bullet \bullet}$.
    
    \item The smallest expected counts associated with the reference row and the reference column $E_{\min}$ is greater or equal to 4: $E(N_{iJ}) = E_{iJ} \geq 4$ for all $i= 1, \dots, I$ and $E(N_{Ij}) = E_{Ij} \geq 4$ for all $j = 1, \dots, J$. 
\end{enumerate}
Define
\[
A_{i^*j^*} =\left\{1 -  \sqrt{\left(1-\frac{1}{\bar \lambda_{\bullet \bullet}} \right)^2 + \left( \frac{1}{E_{i^* J}}+\frac{1}{E_{I j^*}}+\frac{1}{E_{IJ}}\right) \frac{1}{\bar \lambda_{\bullet \bullet}}}\right\} \bar \lambda_{\bullet \bullet}^2
\]
and 
\[
B_{i^*j^*} =\left\{1 +  \sqrt{\left(1-\frac{1}{\bar \lambda_{\bullet \bullet}} \right)^2 + \left( \frac{1}{E_{i^* J}}+\frac{1}{E_{I j^*}}+\frac{1}{E_{IJ}}\right) \frac{1}{\bar \lambda_{\bullet \bullet}}}\right\} \bar \lambda_{\bullet \bullet}^2
\] 
Then, for a \textit{non-reference} cell $(i^*,j^*)$, the estimation performance of $\tilde E_{i^*j^*}$ is better than $\hat E_{i^*j^*}$, i.e., 
\[
\text{AMSE}(\hat E_{i^* j^*}/E_{i^* j^*}) > \text{AMSE}(\tilde E_{i^* j^*}/E_{i^* j^*}),
\]
if
\begin{equation}\label{eqn:sufficient_condition_thm2}
    \begin{cases}
   \bar \lambda_{i^* \bullet} \bar \lambda_{\bullet j^*} \in (B_{i^*j^*}, +\infty) & \text{ when } A_{i^*j^*}\leq 1 \\
   \bar \lambda_{i^* \bullet} \bar \lambda_{\bullet j^*} \in (1, A_{i^*j^*})\cup (B_{i^*j^*}, +\infty) & \text{ when } A_{i^*j^*} > 1
\end{cases}.
\end{equation}

\end{theorem}

\begin{proof}
We only provide a brief outline here; detailed steps are given in Appendix section S1.1. We consider the scaled difference $S_{i^* j^*} \equiv S_{i^* j^*}(\lambda_{i^* j^*}, \bar \lambda_{i\bullet}, \bar \lambda_{\bullet j^*}, \bar \lambda_{\bullet \bullet})$  between the asymptotic MSE of $\hat E_{i^* j^*}/E_{i^* j^*}$ and $\tilde E_{i^* j^*}/E_{i^* j^*}$. Then, we decompose $S_{i^* j^*}$ into two parts and show that the first part is always positive, and through involved analysis, establish that the second part is also positive when the sufficient condition \eqref{eqn:sufficient_condition_thm2} holds.  
\end{proof}

Theorem \ref{thm2:Eestimator} guarantees that under mild regularity conditions (A)-(D) the proposed estimator $\tilde{E}_{i^* j^*}$ for the null baseline counts outperforms the natural estimator $\hat{E}_{i^* j^*}$ for non-reference AE-drug pairs $(i^*, j^*)$ with AEs $\{i^*\}$ (and/or drugs $\{j^*\}$) that, on average, have ``large'' signal strengths $\{\bar{\lambda}_{i^* \bullet}\}$ (resp., $\{\bar{\lambda}_{\bullet j^*}\}$) across all drugs (resp., AEs), such that the products $\{\bar{\lambda}_{i^* \bullet} \bar{\lambda}_{\bullet j^*}\}$ are also ``large'' (as determined by the sufficient condition \eqref{eqn:sufficient_condition_thm2}). Clearly, this condition is satisfied for at least one $(i^*, j^*)$ in situations where the overall average signal strength $\bar{\lambda}_{\bullet \bullet}$ for the entire table is large, implying the superiority of $\tilde{E}$ over $\hat{E}$ in such scenarios. In Appendix S1.2, we provide a few corollaries based on Theorem \ref{thm2:Eestimator} and the sufficient condition \eqref{eqn:sufficient_condition_thm2} to further elucidate situations with guaranteed superiority of $\tilde E$ over $\hat E$.

We emphasize the sufficient nature of the condition \eqref{eqn:sufficient_condition_thm2}: $\tilde{E}_{i^* j^*}$ is guaranteed to outperform $\hat{E}_{i^* j^*}$ if \eqref{eqn:sufficient_condition_thm2} holds; however, $\tilde{E}_{i^* j^*}$ may still perform as well as, or even better than, $\hat{E}_{i^* j^*}$ in applications where \eqref{eqn:sufficient_condition_thm2} does not hold. Indeed, under simulations emulating real SRS data (see Section \ref{sec:computed_AMSE}), we demonstrate that there is virtually no efficiency loss---and even some potential improvement---in using $\tilde{E}_{i^* j^*}$ instead of $\hat{E}_{i^* j^*}$ in situations that violate \eqref{eqn:sufficient_condition_thm2}, while there can be clear and significant efficiency gains when \eqref{eqn:sufficient_condition_thm2} does hold. In practice, formal verification of \eqref{eqn:sufficient_condition_thm2} is not feasible, as it involves the underlying unknown parameters $\{\lambda_{ij}\}$, $\{p_{i\bullet}^*\}$, and $\{p_{\bullet j}^*\}$. However, the mild nature of the regularity conditions (A)-(D), coupled with our simulation-based results (Section~\ref{sec:computed_AMSE}), provides substantial confidence in using $\tilde{E}$ as a general-purpose replacement for $\hat{E}_{i^* j^*}$ across a wide range of applications.

\section{Simulation Results} \label{sec:simulation}

This section compares the performances of the proposed non-parametric Bayesian approaches against the existing Bayesian/empirical Bayesian approaches reviewed in Section \ref{sec:review} in both signal discovery and signal strength estimation under extensive, controlled simulations, to provide strong frequentist validation for our methods.  We considered our motivating statin dataset described in Section~\ref{sec:statin-data-description} with $I = 43$ AEs---including a reference AE ``other AEs''---and $J=7$ drugs---including a reference drug ``other drugs''. Random datasets similar to this motivating statin dataset in terms of the row (AE) and column (drug) marginal total counts were generated. However, the  \textit{true} relative reporting ratio parameter values $\{\lambda_{ij}\}$ for each generated dataset were fully specified and controlled.  In particular, we used a multinomial-based random data generating process (see Algorithm 2 below) given a  ``true'' signal strength matrix $((\lambda_{ij}^{\text{true}}))$ with the convention that $\lambda_{ij}^{\text{true}} > 1$ for all \textit{true signal} cells $\{(i, j)\}$, $\lambda_{ij}^{\text{true}} = 1$ for all \textit{true non-signal} cells, and $\lambda_{ij}^{\text{true}} = 0$ for all \textit{true structural zero cells} (if present).  We considered three separate settings with different levels of heterogeneity among the signal cell $\lambda_{ij}^{\text{true}}$ values: (I) homogeneous signals: all signal cell $\{\lambda_{ij}^{\text{true}}\}$ are the same, (II) moderately heterogeneous signals: $50\%$ of all signal cell $\{\lambda_{ij}^{\text{true}}\}$ are fixed at a specific value and the remaining $50\%$ are fixed at another specific value (different from the first $50\%$), and (III) highly heterogeneous signals: all signal cell $\{\lambda_{ij}^{\text{true}}\}$ are fixed at different values. For setting I, we varied the number of true signals to be 1, 3, and 6; for settings II and III, the number of true signal cells was fixed at 6 and 12, respectively. Positions for the signal cells excluded the reference ($I$-th) row and the reference ($J$-th) column; the specific signal positions employed in the three simulation settings are detailed in Appendix Tables 6. 

A range of specific values for the signal cell $\lambda_{ij}^{\text{true}}$ was considered in each setting/number of signal cell combinations. In setting I, each individual signal $\lambda_{ij}$ was varied in $\{1.2, 1.4, 1.6, 2.0, 2.5, 3.0, 4.0\}$. In setting II, three signal cell $\lambda_{ij}^{\text{true}}$ were fixed at 2, and the remaining 3 signal cell $\lambda_{ij}$ were varied in $\{1.2, 1.4, 1.6, 2.0, 2.5, 3.0, 4.0\}$. Finally, in setting 3, all 12 signal cells $\lambda_{ij}$ were set at separate, distinct values; see Appendix Table 6 for details. 

For each choice of $((\lambda_{ij}^{\text{true}}))$ matrix derived from the process described above under settings I and II, three levels of structural zero inflation were considered: (a) no zero-inflation: the $\{\lambda_{ij}^{\text{true}}\}$ values obtained through the process above were kept unaltered, (b) moderate zero-inflation: $25\%$ of all $\{\lambda_{ij}^{\text{true}}\}$ were replaced by $0$ to be the true structural zeros, and (c) high zero-inflation: $50\%$ of all $\{\lambda_{ij}^{\text{true}}\}$ were replaced to be structural zeros. For (b) and (c), the structural zero positions were distributed randomly among the non-signal cells (with $\{\lambda_{ij}^{\text{true}} = 1\}$) \textit{outside} of the reference row $I$ and the reference column $J$.  Afterward, the structural zero positions were kept fixed prior to generating $M=1000$ replicated datasets for the setup using the multinomial-based random data generation described in Algorithm \ref{alg:multinomial_data_generation}.

\begin{algorithm}
\caption{Multinomial AE-drug report count data $\{N_{ij}\}$ generation process}  \label{alg:multinomial_data_generation}
\begin{algorithmic}
\Require A signal strength matrix $((\lambda_{ij}: i = 1, \dots, I; j = 1, \dots, J))$,  an exemplar dataset $((N_{ij}: i = 1, \dots, I; j = 1, \dots, J))$, and a structural zero probability $\omega$.

\Statex

\State 1. Compute the grand total $\tilde N_{\bullet \bullet} = \sum_{i=1}^I\sum_{j=1}^J \tilde N_{ij}$, row totals $\{\tilde N_{i\bullet} = \sum_{j=1}^J \tilde N_{ij}: i = 1, \dots, I\}$, and column totals $\{\tilde N_{\bullet j} = \sum_{i=1}^I \tilde N_{ij}: j = 1, \dots, J\}$. Also compute the corresponding row and column marginal proportions $p_{i \bullet}$ and $p_{\bullet j}$ where $p_{i \bullet}$ $(p_{1\bullet}^*, p_{2\bullet}^*, \dots, p_{I \bullet}^*)$, where $p_{i \bullet}^* = \tilde N_{i \bullet}/\tilde N_{\bullet \bullet}$ and $p_{\bullet j}^* = \tilde N_{\bullet j}/ \tilde N_{\bullet \bullet}$ .

\State 2. Generate structural zero position indicators $z_{ij} \sim \text{Bernoulli}(\omega)$ such that $z_{ij} = 1$ implies cell $(i, j)$ is a structural zero.

\State 3. Compute the cell probabilities $P = (p_{11}, p_{12}, \dots, p_{IJ})$ such that:
\[
p_{ij} = \frac{(1-z_{ij})\lambda_{ij}p_{i\bullet}^*p_{\bullet j}^*}{\sum_{i=1}^I\sum_{j=1}^J (1-z_{ij})\lambda_{ij}p_{i\bullet}^*p_{\bullet j}^*}.
\]

\State 4. Generate a random table $((N_{ij}))$ with $(N_{11}, N_{12}, \dots, N_{IJ}) \sim \text{Multinomial}(\tilde N_{\bullet \bullet}, P)$.
\end{algorithmic}
\end{algorithm}

For each of the above $85$ total simulation scenarios---$63$ from setting I, $21$ from setting II, and $1$ from setting III (see Table \ref{table:simulation-setting-1} and \ref{table:simulation-setting-2})---each characterized by a specific $((\lambda_{ij}^{\text{true}}))$ matrix, we employed two separate approaches for random data generation using Algorithm \ref{alg:multinomial_data_generation}: the ``fixed truth'' approach and the ``randomly perturbed truth'' approach. The fixed truth approach drew directly upon the framework of replication-based frequentist uncertainty in the simulations, assuming an underlying data-generating model with fixed ``true'' parameters. By contrast, the ``randomly perturbed truth'' approach added small noises to the constructed $((\lambda_{ij}^{\text{true}}))$ matrices to allow some deviations from the underlying frequentist assumption of fixed true parameters in the simulations. More specifically, the fixed truth approach used the final $((\lambda_{ij}^{\text{true}}))$ matrix as obtained from the process above to generate $M = 1000$ replicated datasets for each simulation scenario.  In the randomly perturbed truth approach, we added small random (varying with replicates) noises---independently generated from a truncated normal distribution with mean $0$, standard deviation $0.05$, truncated within $(-0.4, 0)$ for non-signal cells and truncated within $(-0.2, 0.2)$ for signal cells ---to $((\lambda_{ij}^{\text{true}}))$ before generating $M = 1000$ replicated random data. 

In each randomly generated dataset, we fitted the three proposed non-parametric empirical Bayes models, viz., the general-gamma model (Section \ref{sec:general-gamma}), the Efron model (Section \ref{sec:efron}), and the KM model (Section \ref{sec:KM}) with the proposed null baseline expected count estimates $\{\tilde E_{ij}\}$ (Section \ref{sec:estimating-Eij}). For the single-gamma model, the hyperparameter $\alpha$ was set to $0.99$ to encourage flexibility  while still ensuring sparsity of the final mixtures. To facilitate comparison, we also fitted all the existing empirical Bayes and hierarchical Bayes models/approaches reviewed in Section \ref{sec:review}, viz., BCPNN, single-gamma, 2-gamma, 2-gamma-zi, and HDP. These competing approaches used the natural estimator $\hat E_{ij}$ of the null baseline counts, as suggested in the literature. From each model fit (except BCPNN), a signal discovery analysis was first performed through its computed (estimated) posterior probabilities of being a signal: $\{\Pr(\lambda_{ij} \geq 1.001 \mid \text{data})\}$: a cell $(i, j)$ with $\Pr(\lambda_{ij} \geq 1.001 \mid \text{data}) > 0.95$ was deemed as a \textit{discovered signal}. For BCPNN, analogous posterior probabilities of being a signal were obtained through the information content parameter $2^{IC_{ij}}$ and the resulting estimated posterior probabilities were corrected for false discovery rates (FDR) \cite{muller2006fdr}. A similar FDR adjustment was also made to posterior probabilities obtained from the single-gamma prior model. Following the signal discovery analysis, performances of each model/method were assessed through the simulation/replication-based FDRs and sensitivities defined as follows:
\begin{equation} \label{eqn:FDR}
    \text{FDR} = \frac{1}{M} \sum_{m=1}^M \frac{\# \left\{(i, j): d_{ij}^{(m)} =  1, \lambda_{ij}^{\text{true}} \leq  1 \right\}}{\# \left\{(i, j): d_{ij}^{(m)} =  1\right\} }
\end{equation}

and
\begin{equation} \label{eqn:sensitivity}
    \text{Sensititvity} = \frac{1}{M} \sum_{m=1}^M \frac {\# \left\{ (i, j): d_{ij}^{(m)} =  1, \lambda_{ij}^{\text{true}} >  1 \right\}}{\# \left\{(i, j): \lambda_{ij}^{\text{true}} >  1\right\}}
\end{equation}

where $d_{ij}^{(m)} =  1$ indicates that the cell $(i, j)$ is deemed as a discovered signal in $m$-th replicated dataset, and $\#A$ denotes the cardinality of a set $A$. 

Next, for each model/method other than BCPNN (which does not use a $\lambda_{ij}$ parametrization) and single-gamma (which showed very poor $\lambda_{ij}$ estimation performance), we focused on the entire computed posterior distributions for $\{\lambda_{ij}\}$ obtained from that method for signal strength estimation. However, evaluating these estimates is difficult because existing Bayesian/empirical Bayesian approaches to pharmacovigilance primarily focus on signal detection--which is evaluated via FDR and sensitivity--rather than signal strength estimation. To our knowledge, no unified approach for evaluating the performance of Bayesian signal strength estimators exists in the literature. This is particularly important, as the computed posterior distributions, when used as Bayesian/empirical Bayesian signal strength estimates, may exhibit considerable asymmetry, multimodality, and uncertainty (see, e.g., Figure \ref{fig:real_data_analysis} under real data analysis). These nuances are not well captured by a single point estimate of $\lambda_{ij}$, and hence any point estimator performance metric--such as mean squared errors (MSEs) or mean absolute errors (MAEs) assessing frequentist variability in point estimates  of  $\lambda_{ij}$ across simulations/replications--is likely to be uninformative. To address this, we propose and consider the two Wasserstein distance-based performance assessment metrics in Section~\ref{sec:Evaluation_metrics} below that reflect variability in computed posterior distributions (i.e., Bayesian uncertainty) while also accounting for variability across simulation replicates (i.e., frequentist uncertainty).  

In Sections \ref{sec:simu_setting1}, \ref{sec:simu_setting2} below, we report signal discovery and signal strength estimation performance results evaluated through the replicated datasets simulated under the ``fixed truth'' setup. Analogous results obtained under the ``randomly perturbed truth'' setup are mostly similar; detailed descriptions are provided in Appendix S5.

\subsection{Evaluation metrics for Bayesian signal strength estimators}\label{sec:Evaluation_metrics}

Let $f_{ij}$ be the posterior density of the signal strength parameter $\lambda_{ij}$ given an SRS dataset. To derive evaluation metrics for $f_{ij}$ to be used as a Bayesian estimator of the signal strength parameter $\lambda_{ij}$, we consider the general scaled $p$-th Wasserstein distance between the posterior density function $f_{ij}$ and the degenerate distribution characterizing $\lambda_{ij}^{\text{true}}$, the true signal strength of $(i,j)$-th cell, as:
\[
\text{Scaled-Wasserstein}_p(f_{ij}, \lambda_{ij}^{\text{true}}) = \frac{1}{\lambda_{ij}^{\text{true}}}\left[\int_{0}^{1} \left|F_{ij}^{-1}(q)-F_{\lambda_{ij}^{\text{true}}}^{-1}(q)\right|^p dq\right]^{1/p}.
\]
Here $F_{ij}$ and $F_{\lambda_{ij}^{\text{true}}}$ denote the cumulative distribution functions associated with $f_{ij}$ and the degenerate distribution characterized by a point-mass at $\lambda_{ij}^{\text{true}}$, respectively, and $F_{ij}^{-1}$ and $F_{\lambda_{ij}^{\text{true}}}^{-1}$ are the corresponding quantile functions. For computation of these Wasserstein distances, we leverage the degeneracy of $F_{\lambda_{ij}^{\text{true}}}$ and utilize the consequent equivalent formulation defined in terms of expectation with respect to the posterior density $f_{ij}$:
\[
\text{Scaled-Wasserstein}_p(f_{ij}, \lambda_{ij}^{\text{true}}) =
\frac{1}{\lambda_{ij}^{\text{true}}}\left[\int_{0}^{\infty} \left|\lambda-\lambda_{ij}^{\text{true}}\right|^p f_{ij}(\lambda)d\lambda\right]^{1/p}.
\]
See Appendix S4 for a formal proof of this equivalence. In this paper, we focus on the cases $p=1$ and $p=2$. Analytic expressions for the resulting distances for these two specific choices of $p$ can be derived using the gamma-mixture form (for general-gamma, single-gamma, 2-gamma, and 2-gamma-zi) or the discrete mixture form (for Efron and KM) based on the estimated posterior distributions from all $\lambda_{ij}$-based empirical Bayesian approaches considered in this paper (see Appendix S4.1). For the hierarchical Bayesian approach (HDP), an analytic form for the posterior and thus for these distances is not available; however, these distances can still be efficiently computed given posterior MCMC draws for $\{\lambda_{ij}\}$ from the HDP model.

The Wasserstein distances for $p=1$ and $2$ are henceforth called the scaled posterior mean absolute error (Scaled posterior MAE) and the scaled posterior root mean squared error (Scaled posterior RMSE), respectively. Based on these, we define the following two evaluation metrics that combine \textit{all signal cells} in the entire table:
\begin{equation*}
    \begin{split}
        & \text{Average-Scaled-Wasserstein}_p(\{f_{ij}\}, \{\lambda_{ij}^{\text{true}}\}) = \frac{1}{\# C_{\text{sig}}} \sum_{(i,j)\in C_{\text{sig}}} \text{Scaled-Wasserstein}_p(f_{ij}, \lambda_{ij}^{\text{true}}), \\
        & \text{Max-Scaled-Wasserstein}_p(\{f_{ij}\}, \{\lambda_{ij}^{\text{true}}\}) = \max_{(i,j)\in C_{\text{sig}}} \left[ \text{Scaled-Wasserstein}_p(f_{ij}, \lambda_{ij}^{\text{true}}) \right],
    \end{split}
\end{equation*}
where $C_{\text{sig}} = \{(i, j): \lambda_{ij}^{\text{true}} > 1 \}$ is the set of all \textit{true} signal cells. We use these metrics to summarize the signal strength estimation performance of a method on a given simulated dataset. As the final performance assessment metrics summarizing performances across all replicated datasets, we consider the replication-based (i.e., frequentist) averages of these distances. Specifically,  given (estimated) posterior density functions $f_{ij}^{(m)}$ from the $m$-th replicate, $m = 1, \dots, M=1000$, we obtain the summary metrics:
\begin{equation} \label{eqn:Average-Scaled-RMSE&Max-Scaled-RMSE}
    \begin{split}
        &\text{Average-Scaled-RMSE} = \frac{1}{M}\sum_{m=1}^M\text{Average-Scaled-Wasserstein}_2(\{f_{ij}^{(m)}\}, \{\lambda_{ij}^{\text{true}}\}), \\
        &\text{Max-Scaled-RMSE} = \frac{1}{M}\sum_{m=1}^M\text{Max-Scaled-Wasserstein}_2(\{f_{ij}^{(m)}\}, \{\lambda_{ij}^{\text{true}}\}).
    \end{split}
\end{equation}
In the ``randomly perturbed truth'' scenarios where different, randomly perturbed true signal strength matrices considered in the replicated data generation, $\{\lambda_{ij}^{\text{true}}\}$ in Equation \eqref{eqn:Average-Scaled-RMSE&Max-Scaled-RMSE} is replaced by $\{\lambda_{ij}^{\text{true}(m)}\}$---the true signal strength matrix in the $m$-th replicated table. Similar Average-Scaled-MAE and Max-Scaled-MAE metrics are constructed by setting $p=1$. These metrics are used to monitor and assess the performance of the different methods discussed. 

\subsection{Results from simulation setting I (homogeneous signal strengths)} \label{sec:simu_setting1}

This simulation setting assesses the performances of the models when signal strengths are homogeneous. Table \ref{table:simulation-setting-1} shows the individual setups/levels for the different factors considered and varied in this setting, namely, the number of signals, signal strength values for the true signals, level of zero-inflation, and whether or not random perturbation in the true signal values is considered.  For each data configuration, $M=1000$ replicated tables are generated using Algorithm \ref{alg:multinomial_data_generation}. The Max-Scaled-RMSE results, as defined in Equation \eqref{eqn:Average-Scaled-RMSE&Max-Scaled-RMSE} with no random perturbation in $\{\lambda_{ij}^{true}\}$ are provided below. Similar results with random perturbations in $\{\lambda_{ij}^{true}\}$ yield comparable conclusions and can be found in the Appendix S5.1. Appendix S5.2 displays results in terms of scaled-MAE for both fixed and randomly perturbed $\{\lambda_{ij}^{true}\}$, and these results are largely concordant with their Max-Scaled-RMSE counterparts.

\begin{table}[ht]
\centering
\begin{tabular}{lll}
\hline
Factors        & Level & Number of levels \\ \hline
Number of signal cells & 1 (Case 1), 3 (Case 2), 6 (Case 3)           & 3 \\
Signal strength & 1.2, 1.4, 1.6, 2.0, 2.5, 3.0, 4.0 & 7 \\
Level of zero-inflation  & None, Low, High           & 3 \\
Random perturbation in $\{ \lambda_{ij}^{\text{true}}\}$   & Yes, No    & 2 \\\hline
\end{tabular}
\caption{Random table configurations in simulation setting I (homogeneous signal strengths)}
\label{table:simulation-setting-1}
\end{table}

\begin{figure}[ht]
    \centering
    \includegraphics[width=\textwidth]{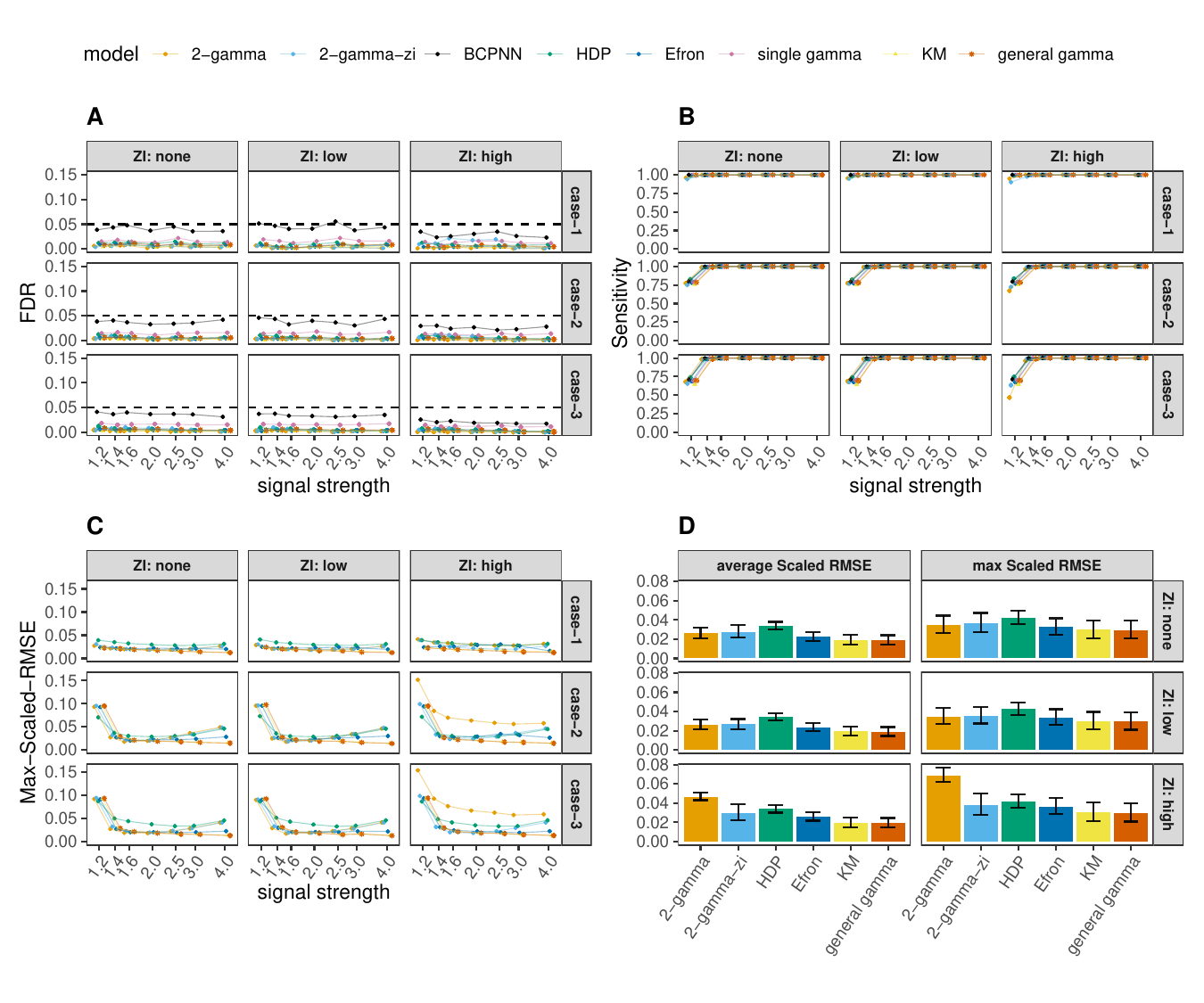}
    \caption{Results from simulation setting I (homogeneous signal strength). Panel A, B, C: line and dot plot matrices display replication-based FDR (Equation \eqref{eqn:FDR}), sensitivity (Equation \eqref{eqn:sensitivity}), and scaled RMSE (Equation \eqref{eqn:Average-Scaled-RMSE&Max-Scaled-RMSE}), respectively, plotted along the vertical axes against grid values of signal strength $\lambda_{ij}$ for the signal cells, across different zero inflation levels (along plot columns) and number of signals (cases; along the plot rows). The lines and dots are color coded by the different methods/methods;  dots for KM and general-gamma are highlighted via triangle and asterisks. Panel D: The overall means of  average (left column) and maximum (right column) scaled RMSEs obtained across all $\lambda_{ij}$ grid values and all choices of the number of signal cells are plotted as vertical bars for different levels of zero inflation (along rows). Each bar represents the replication-based mean of the average (left column) or the maximum (right column) scaled RMSE of a specific method computed over the signal cells of an entire table; the associated error whiskers represent 5th and 95th percentile points across  replicates. } 
    \label{fig:comb_set1}
\end{figure}

Figure \ref{fig:comb_set1} visualizes the signal discovery (panels A, B) and signal strength estimation (panels B, C) results for this setting. Panels A and B are plot matrices that document FDRs and sensitivities, respectively, along the vertical axes against different $\lambda_{ij}$ grid values for the signal cells for the different models (color-coded lines/dots) across varying numbers of signal cells (cases; along rows) and different zero inflation (ZI) levels (across columns). The results indicate that all models/methods strongly control the FDR within the nominal 0.05 level and exhibit similar FDR profiles across all signal strengths, numbers of signal cells, and ZI levels, except BCPNN, which incurs somewhat higher FDRs---mostly controlled within $0.05$ but with some situations with certain $\lambda$ grid values where computed FDR exceeds 0.05 despite the FDR adjustment during model fitting. The sensitivity profiles for the methods also appear similar overall, except for the 2-gamma model, which shows somewhat lower sensitivities than the others in scenarios with a large number of signal cells and a high level of ZI. This is not particularly surprising, as all models except BCPNN (including the 2-gamma-zi model) are flexible enough to accommodate zero inflation when present. BCPNN demonstrates similar sensitivities to the proposed non-parametric empirical Bayes approaches; however, this comes at the cost of possibly higher FDR. Together, Figure \ref{fig:comb_set1}A and \ref{fig:comb_set1}B demonstrate that the proposed non-parametric empirical Bayes approaches perform comparably to state-of-the-art approaches in terms of signal discovery, as evidenced by the strong control of FDR and high sensitivity profiles.

Panels C and D display the estimation performance results for all methods except BCPNN (which does not use a $\lambda_{ij}$-based formulation) and single-gamma (which is too inflexible for reasonable estimation). In Panel C, the Max-Scaled-RMSE (Equation \eqref{eqn:Average-Scaled-RMSE&Max-Scaled-RMSE}) is plotted against signal cell $\lambda_{ij}$ grid values separately for the different models (color-coded lines/dots) across varying numbers of signal cells (cases; along rows) and different zero inflation (ZI) levels (across columns). The bars in Panel D show the overall mean of the average scaled RMSE (left) and the maximum scaled RMSE (right) obtained across different signal cell $\lambda_{ij}$ grid values and numbers of signal cells (cases) for each ZI level (rows). This provides an average estimation performance measurement for the different levels of zero inflation. The associated error whiskers represent the 5th and 95th percentile points across replicates.

The following observations can be made from these figures. First, all methods across all scenarios show a general pattern of the Max-Scaled-RMSE decreasing (i.e., the estimation performance improving) as the signal cell $\lambda_{ij}$ values increase (Panel C). However, the magnitude of the decrease in Max-Scaled-RMSE varies among the methods, and some methods (namely, 2-gamma, 2-gamma-zi, and HDP) exhibit a U-shaped pattern, where the Max-Scaled-RMSE increases when the underlying true $\lambda_{ij}$ is very high. This discrepancy between the methods showing a U-shaped pattern and our proposed non-parametric empirical Bayes approaches, which show a monotonic decreasing pattern, can be largely attributed to the improvements in estimation accuracy provided by the proposed null baseline expected count $\tilde{E}_{ij}$ estimator (see Figure \ref{sec:computed_AMSE}), which is utilized in the latter group but not in the former.

Second, for each method other than the 2-gamma model, the overall estimation performance pattern remains consistent across the different levels of ZI for each number of true signals (each row in Panel C and each column in Panel D). However, the 2-gamma model shows noticeable deterioration in signal strength estimation as the ZI level increases, due to its inability to flexibly accommodate zero inflation, unlike the other models. Finally, across all scenarios—for all ZI levels and numbers of signal cells—the general-gamma model and the KM model produce the most accurate estimators of the true $\lambda_{ij}$ values, with the smallest Max-Scaled-RMSEs. The Efron model results are close competitors, with Max-Scaled-RMSEs slightly higher than those obtained for the KM and general-gamma models. However, the performance of the Efron model could potentially be improved with further tuning of model parameters (specifically $p$ and $c_0$), which future research may explore.

\subsection{Estimation performance of the expected null baseline count E under simulation setting I (homogeneous signal strengths)} \label{sec:computed_AMSE}

We provide a comparison of the estimation performances of the two estimators $\hat{E}$ and $\tilde{E}$ of the expected null baseline count $E$ in terms of scaled MSE under simulation setting I, as defined in Section \ref{sec:simu_setting1}, with no zero-inflation (ZI:none) and 1 embedded signal (Case-1). We considered a fixed $((E_{ij}^{\text{true}} = (n_{i\bullet} n_{\bullet j}/n_{\bullet \bullet}))$ with $n_{i\bullet}$, $n_{\bullet j}$, and $n_{\bullet \bullet}$ as obtained from the statin-42 dataset, and a range of values for the $\lambda_{ij}^{\text{true}}$ parameter for the signal cell, ranging between 1 and 10, while keeping the $\lambda_{ij}^{\text{true}}$ parameters for all non-signal cells fixed at 1. Subsequently, with each specified $((E_{ij}^{\text{true}}))$ and $((\lambda_{ij}^{\text{true}}))$, we generated $M=1000$ replicated datasets $\{((n_{ij}^{(m)})): m = 1, \dots, 1000\}$ using Algorithm~\ref{alg:multinomial_data_generation}, focusing on estimating $\{E_{ij}\}$ for the entire table (excluding the reference row $I$ and reference column $J$) in each replicate $m$ using the two estimators $\{\hat{E}_{ij}\}$ and $\{\tilde{E}_{ij}\}$. 

For each $((\lambda_{ij}^{\text{true}}))$, we then evaluated the replication-based scaled MSEs for the two estimators of $E_{ij}$ using the formulas:
\[
\text{Scaled-RMSE}(\hat{E}_{ij}) = \sqrt{\frac{1}{M} \sum_{m=1}^M \frac{(\hat{E}_{ij} - E_{ij}^{\text{true}})^2}{(E_{ij}^{\text{true}})^2}}
\]
and
\[
\text{Scaled-RMSE}(\tilde{E}_{ij}) = \sqrt{\frac{1}{M} \sum_{m=1}^M \frac{(\tilde{E}_{ij} - E_{ij}^{\text{true}})^2}{(E_{ij}^{\text{true}})^2}}
\]
for each cell $(i, j)$ under each choice of $((\lambda_{ij}^{\text{true}}))$. We then considered the scaled RMSE ratios $\text{Scaled-RMSE}(\hat{E}_{ij})/\text{Scaled-RMSE}(\tilde{E}_{ij})$. These scaled RMSE ratios are plotted (along the vertical axis) against the values of $\lambda_{ij}^{\text{true}}$ for the signal cell (horizontal axis) as dots in Figure \ref{fig:Eestimators_comparison}, separately for signal and non-signal cells (panels). The red vertical lines in the figure separate the setting where the sufficient condition \eqref{eqn:sufficient_condition_thm2} holds (right) and does not hold (left).  The dashed horizontal line indicates where the ratio equals 1.

\begin{figure}[htp]
    \centering
    \includegraphics[width=0.8\textwidth]{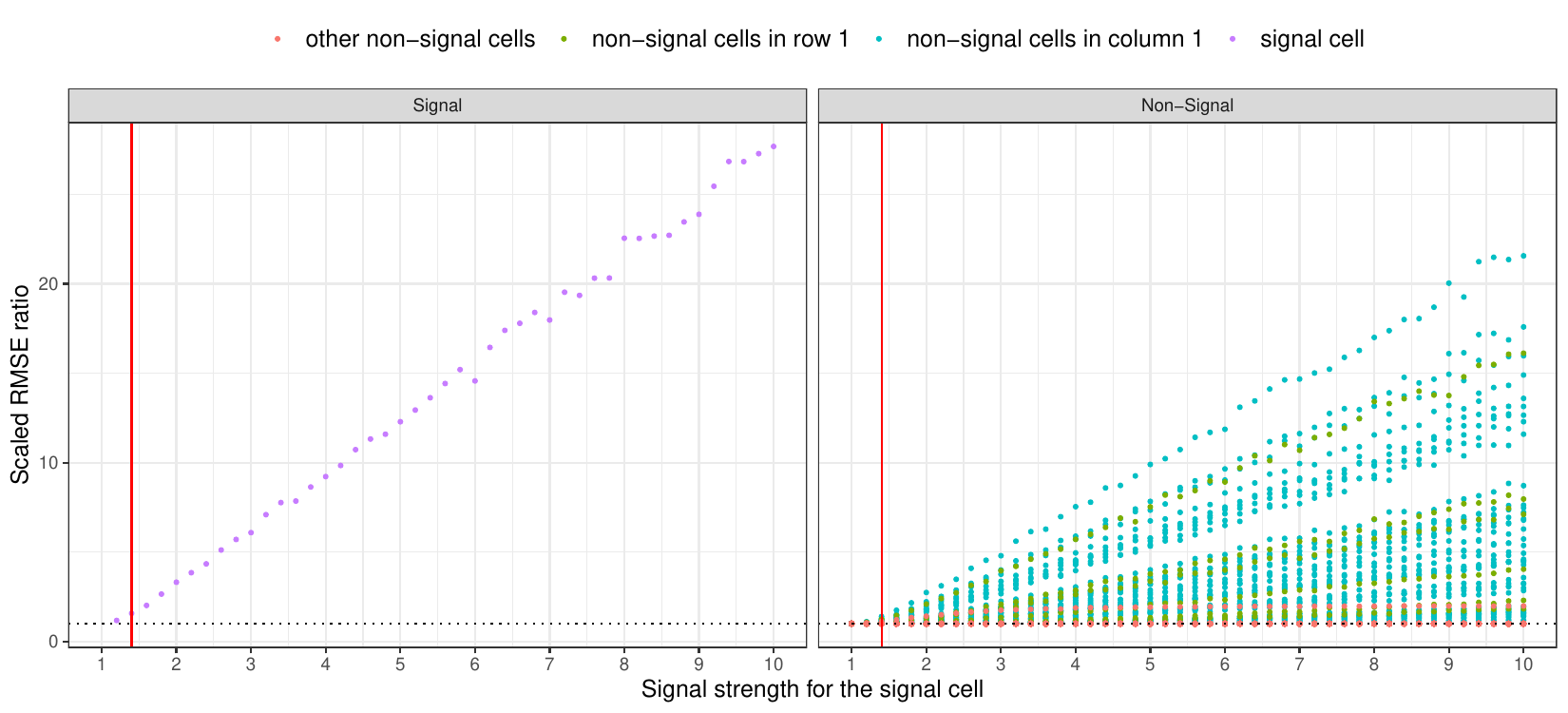}
    \caption{Results of estimating performance of the expected null baseline count (E) under simulation setting I (homogeneous signal strengths). The figures display the replication-based ratio: Scaled-RMSE($\hat E_{ij}$)/Scaled-RMSE($\tilde E_{ij}$) plotted along the vertical axes against the signal strength of the only signal cell $(1, 1)$ for the signal cell  (left panel) and non-signal cells (right panel). The dots are color-coded based on whether they correspond to the signal cell (purple), non-signal cells sharing the same row (green) or column (teal) as the signal cell, and other non-signal cells (orange). The dashed horizontal line indicates where the ratio equals 1, and the red vertical line denotes the minimum signal strength of the signal cell where the sufficient condition \eqref{eqn:sufficient_condition_thm2} holds. } \label{fig:Eestimators_comparison}
\end{figure}

As seen in Figure \ref{fig:Eestimators_comparison}, the estimator $\tilde{E}$ attains impressive efficiency gains over $\hat{E}$, with scaled RMSE ratios $\gg 1$ for the signal cell (left facet) while also performing noticeably better for many non-signal cells (right facet), particularly in settings where the underlying true signal strength $\lambda_{ij}^{\text{true}}$ for the signal cell is moderate or high (horizontal axis $\geq$ 2). This efficiency gain is highly prominent when the sufficient condition \eqref{eqn:sufficient_condition_thm2} holds (right sides of the red vertical lines in the two panels) but is also present, albeit to a lesser extent, when the sufficient condition \eqref{eqn:sufficient_condition_thm2} does not hold (left sides of the lines). Importantly, across all settings, there is little to no efficiency loss in $\tilde{E}$ compared to $\hat{E}$ for non-signal cells. This is evident from the lower ends of the scaled RMSE ratios for non-signal cells, which virtually never fall below 1.

\subsection{Results from simulation settings II and III (moderately heterogeneous and highly heterogeneous signal strengths)} \label{sec:simu_setting2}

Simulation settings II and III assess the signal discovery and signal strength estimation performances of the models under scenarios with moderately and highly heterogeneous $\lambda_{ij}$ for signal cells. Table \ref{table:simulation-setting-2} summarizes the levels of the various simulation scenarios/factors varied for setting II. For setting III, only two scenarios are considered: one with fixed truth and one with randomly perturbed truth, each having 6 signal cells with $\lambda_{ij}^{\text{true}} > 1$ in the entire table. For both settings, we focus below on results from the "fixed truth" $\{\lambda_{ij}^{\text{true}}\}$ scenarios. Comparable results with perturbations, which again demonstrate similar findings, are provided in Appendix S5.1.  We further concentrate here on the signal estimation performance; signal detection performances are displayed in Appendix S5.1, showing similar patterns as those observed in setting I and demonstrating similar conclusions that the proposed non-parametric empirical Bayes approaches perform comparably to the existing approaches. For the estimation performance assessment, we again utilize the max-Scaled-RMSE as defined in Equation \eqref{eqn:Average-Scaled-RMSE&Max-Scaled-RMSE}. Analogous results based on max-Scaled-MAE are provided in Appendix S5.2.

\begin{table}[]
\centering
\begin{tabular}{lll}
\hline
Factors        & Level & number of levels \\ \hline
Case & Case4            & 1 \\
Signal strength & 1.2, 1.4, 1.6, 2.0, 2.5, 3.0, 4.0 & 7 \\
Level of zero-inflation  & none, low, high           & 3 \\
Existence of variation in $\{\lambda^{\text{true}}\}$   & Y, N    & 2 \\\hline
\end{tabular}
\caption{Random table configurations in simulation setting II (moderately heterogeneous signal strengths)}
\label{table:simulation-setting-2}
\end{table}

\begin{figure}[htp]
    \centering
    \includegraphics[width=0.75\linewidth]{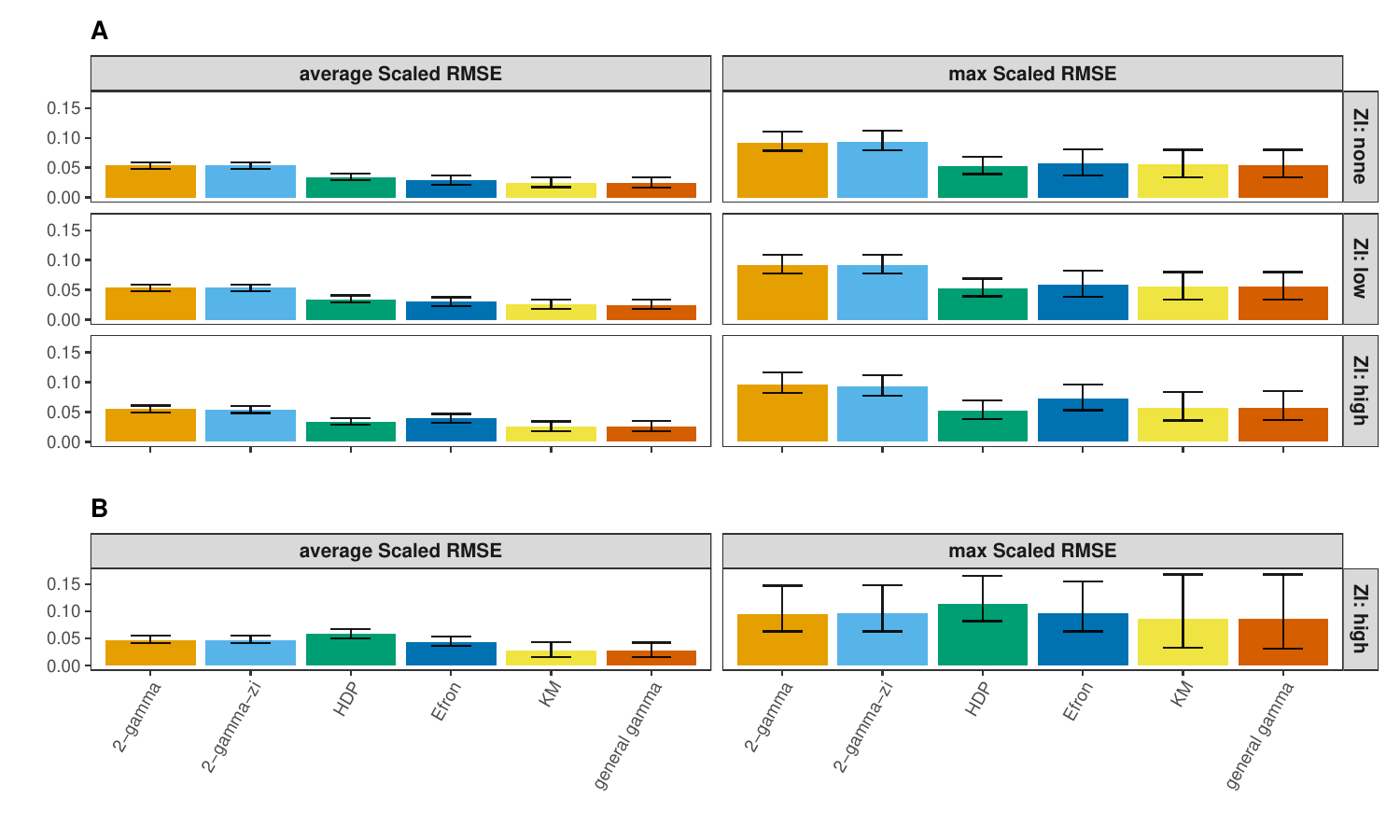}
    \caption{Results from simulation settings II (heterogeneous signal strength; panel A) and III (highly heterogeneous signal strength; panel B). The barplots show the replication-based means of average scaled RMSE and max scaled RMSE across signal strength levels across the replications. The error whiskers display the 5th and 95th quantiles computed over the replicated tables.}\label{fig:sett2_3_average}
\end{figure}

Figure \ref{fig:sett2_3_average} A and B show the signal strength estimation performances for the models through average and max Scaled-RMSE metrics computed over all signal cells and across different simulation scenarios. The findings largely resemble those observed in homogeneous signal $\lambda_{ij}$ settings (Figure 3D), except that both 2-gamma and 2-gamma-zi appear to have noticeably poorer performances here. This is unsurprising, as with heterogeneous signals, the effective underlying prior distribution for $\{\lambda_{ij}^{\text{true}}\}$ ends up with multiple ($\geq 4$) clusters/modes (one mode for structural zeros, one mode for non-signals, and at least two modes for signals), which the 2-gamma and 2-gamma-zi models, each allowing fewer than 4 prior modes, are unable to flexibly accommodate. The proposed non-parametric empirical Bayes approaches effectively address this multimodality, leading to highly accurate estimation of $\lambda_{ij}$, particularly when the heterogeneity is moderate (Panel A), which results in well-separated distinct clusters in the underlying prior distribution for $\lambda_{ij}$. When the $\lambda_{ij}$ values for the signal cells are highly heterogeneous as in Setting III, the underlying prior structure for the signal cells effectively becomes flat, with non-prominent modes/clusters and high variance. In this setting, any framework that is able to accommodate two clusters for $\lambda_{ij}$ values---one concentrated cluster for the non-signal cells with $\lambda_{ij}^{\text{true}} = 1$ and one diffuse cluster for the remaining cells  (both zero inflation and varying signal values for $\lambda_{ij}^{\text{true}} > 1$)---can reasonably well estimate the underlying prior for $\lambda_{ij}$. This is precisely what we observe in the panel B. However, the proposed non-parametric empirical Bayes approaches, particularly the KM and general-gamma models, still provide positive albeit moderate improvement over existing approaches on average in signal strength estimation. 

Collectively, the simulation results presented in this section demonstrate the high accuracy in the signal strength $\lambda_{ij}$ estimation permitted by the proposed non-parametric Bayesian approaches. This, in turn, provides sufficient confidence in the results and inferences obtained by employing these approaches in real-world datasets in applications.

\section{real data analysis} \label{sec:real-data-analysis}

In this section, we discuss the results obtained by employing the proposed non-parametric empirical Bayes approaches to the statin data discussed in Section~\ref{sec:statin-data-description}. We consider two subsets of the data: a large subset with 1,491 AEs (statin-1491) and a smaller subset with 46 AEs (statin-46). The former is used for exploratory identification of AEs of concern (signal detection) from the dataset, aimed at providing preliminary alarms. The latter dataset comprises a more focused subset of AEs previously determined to be of importance by the FDA through expert knowledge, and the interest here centers on identifying the major AEs of concern from this curated subset and determining their degree of relevance (i.e., both signal detection and signal strength estimation). 

In each analysis, we fit the proposed non-parametric empirical Bayes models---viz., the KM, Efron, and general-gamma models. To aid comparison, we also employed all the existing Bayesian/empirical Bayesian models/methods reviewed in Section~\ref{sec:review}. The KM and Efron models were fit using grids with $K=\min\{3000, 10 \times I \times J\}$ grid points. Additionally, for the Efron model, the model hyper-parameters parameters $p$ and $c_0$ were set to $120$ and $0.01$, respectively. For the general-gamma model, the Dirichlet prior parameter $\alpha$ was chosen based on LOOCV (see Section S4.3 for details), computed on models fitted with $\alpha \in \{0.01, 0.5, 0.75, 0.99\}$, and the optimal choice was determined to be $\alpha = 0.75$. The HDP model was fit using MCMC sampling (using JAGS\cite{plummer2012jags} to generate $4$ independent chains; each chain produced $5,000$ retained iterations after discarding the initial $5,000$ iterations used as burn-in). For the single-gamma and BCPNN models, FDR adjustments on the computed posterior probabilities of being signals were made before using them for signal detection. Computation times for the different models/methods for the two analyses are provided in Table~\ref{tab:runtime}.

\begin{table}[]
\centering
\begin{tabular}{lll}
\hline
method & Average elapsed time on statin-46& Average elapsed time on statin-1491\\\hline
general-gamma & 0.798 & 104.869\\
KM & 1.276 & 921.133\\
Efron & 1.768 & 178.405\\
HDP & 75.543 & 22311.099\\
2-gamma & 2.911 & 50.941\\
2-gamma-zi & 2.466 & 90.286\\
single-gamma & 0.000 & 0.000\\
BCPNN & 0.003 & 0.044\\
\hline
\end{tabular}
\caption{ Average Time (in seconds) needed to execute the different methods applied on the statin-46 table among 100 runs and on statin-1491 among 10 runs.}
\label{tab:runtime}
\end{table}

As expected, the empirical Bayes approaches appear substantially more efficient than the full Bayesian method (HDP) in terms of computation expenses, as noted in Table \ref{tab:runtime}. Specifically, HDP is over 75 and 200 times more computation-heavy than the proposed general-gamma non-parametric empirical Bayes approach in the statin-46 and statin-1491 datasets, respectively. The simpler approaches of BCPNN and the single-gamma model are essentially effortless in both datasets, thanks to their simple parametric natures. The general-gamma model is more efficient than the 2-gamma and 2-gamma-zi models in the smaller statin-46 dataset despite having many more gamma components in the underlying mixture prior. This is due to the computational overhead required to determine reasonable starting points for the underlying marginal maximum likelihood estimation of the mixture prior density in 2-gamma and 2-gamma-zi models. For these models, our implementation begins with multiple starting points obtained from different approaches and subsequently selects an optimal starting point (see Appendix S2). However, for the larger statin-1491 dataset, this overhead does not dominate the overall computation cost, and in this dataset, the general-gamma model has a higher implementation time. In both datasets, the general-gamma model enjoys significantly lower computation costs compared to the other two non-parametric empirical Bayes approaches, namely, the Efron and the KM models. The Computational costs for the KM and Efron models are somewhat comparable in the statin-46 dataset, but in the larger statin-1491 dataset, the KM model becomes substantially more computation-heavy due to the complexities of the underlying convex optimization routine in higher dimensions (caused by a substantially higher number of prior mixture components). We next focus on the signal detection performances of these approaches in the two datasets.

\begin{table}[htpb]
\centering
\begin{tabular}[t]{lrrrrrr}
\hline
Method & Atorvastatin & Fluvastatin & Lovastatin & Pravastatin & Rosuvastatin & Simvastatin\\
\hline
general-gamma & 172 & 74 & 66 & 165 & 244 & 244\\
KM & 172 & 75 & 66 & 167 & 243 & 243\\
Efron & 173 & 77 & 68 & 175 & 246 & 246\\
HDP & 208 & 92 & 91 & 227 & 305 & 305\\
2-gamma & 175 & 79 & 73 & 179 & 252 & 248\\
2-gamma-zi & 168 & 70 & 66 & 152 & 234 & 234\\
single-gamma & 239 & 110 & 114 & 261 & 336 & 345\\
BCPNN & 239 & 110 & 153 & 279 & 348 & 353\\
\hline
\end{tabular}
\caption{Numbers of signals detected for each drug by different models/methods on the statin-1491 datset}
\label{tab:signal-detect-statin-1491}
\end{table}

\begin{table}[htpb]
\centering
\begin{tabular}{lllllll}
\hline
method & Atorvastatin & Fluvastatin & Lovastatin & Pravastatin & Rosuvastatin & Simvastatin\\\hline
general-gamma & 31 & 13 & 10 & 17 & 25 & 28\\
KM & 31 & 13 & 10 & 17 & 25 & 28\\
Efron & 32 & 13 & 10 & 17 & 25 & 29\\
HDP & 30 & 13 & 10 & 15 & 28 & 28\\
2-gamma & 31 & 13 & 10 & 16 & 25 & 28\\
2-gamma-zi & 30 & 13 & 10 & 14 & 25 & 27\\
single-gamma & 32 & 14 & 11 & 19 & 30 & 29\\
BCPNN & 33 & 14 & 12 & 20 & 30 & 30\\
\hline
\end{tabular}
\caption{Numbers of signals detected for each drug by different models/methods on the statin-46 datset}
\label{tab:signal-detect-statin-46}
\end{table}

For signal detection, we used posterior probability-based approaches as discussed in Section \ref{sec:simulation}. We focus on the number of detected signals by each method for each drug in both the first and second dataset-specific analyses; the results are displayed in Tables \ref{tab:signal-detect-statin-46} and \ref{tab:signal-detect-statin-1491}, respectively. In the second analysis on the smaller statin-46 dataset, the number of detected signals appears consistent across models, although BCPNN detects a slightly larger number of signals for some drugs. However, significant differences are observed in the number of detected signals from the first analysis with the larger statin-1491 dataset.  In particular, Table \ref{tab:signal-detect-statin-1491} shows that BCPNN detects the highest number of signals across most statin drugs, especially for Rosuvastatin and Simvastatin. This aligns with our findings from the simulation experiments, where BCPNN exhibited higher FDR—sometimes exceeding the nominal level of $0.05$, despite FDR control during training—compared to the other methods. The fact that the single-gamma model also detects a large number of signals is not surprising: the estimated single-gamma prior used for all $\{\lambda_{ij}\}$ is disproportionately pulled toward the ``large'' $\lambda_{ij}$ values due to its inflexible parametric nature.  HDP also appears to detect a large number of signals; however, the results may be impacted by computational issues, specifically, potentially poor convergence/mixing of the MCMC chains, which is difficult to assess with $5,000 \times 4$ post burn-in iterations, given the complexity of the model. However, the extremely high computational cost of implementing this model makes substantially increasing the number of MCMC draws challenging.

On the other end of the spectrum, the 2-gamma-zi model appears to be the most conservative, detecting the fewest signals for each drug in the dataset. Models explicitly accounting for zero inflation often demonstrate such conservative signal detection, as noted in several studies \cite{chakraborty2022use, huang2017zero, hu2015signal}, and our results concur with the existing literature. By contrast, the 2-gamma model is much less conservative, detecting notably more signals than 2-gamma-zi (though substantially fewer than BCPNN, single-gamma, and HDP).  The general-gamma model, a flexible refinement of these two gamma mixture models that permits simultaneous acknowledgment of heterogeneity among the signal, non-signal, and zero-inflation (if present) $\lambda_{ij}$ values, leads to a detected number of signals between the number of signals determined by the 2-gamma and 2-gamma-zi models. The KM approach identifies nearly the same number of signals as general-gamma, aligning with our simulation-based observations that these two models perform similarly in applications (Section \ref{sec:simulation}). However, the KM model requires more computation, as shown in Table \ref{tab:runtime}. The Efron model also appears to perform roughly similarly to the KM and general-gamma models but detects somewhat more signals for each drug. This could be an artifact of the tuning parameter in the Efron model. However, as noted in Section~\ref{sec:efron}, there is currently no general approach suggested in the literature for tuning these parameters, making optimal implementation of the model somewhat challenging.

We next focus on signal strength inference for the $\{\lambda_{ij}\}$ parameters through the computed empirical Bayes posterior distributions, summarized by the associated posterior medians and equi-tailed 90\% posterior credible intervals. Statin drugs are known to be associated with muscular and renal disease\cite{LAW2006S52, wolfe2004dangers}. With a selection of muscle and kidney related AE, the resulting posterior medians and 90\% credible intervals is displayed as dot and error bar (forest) plots in Figure \ref{fig:real_data_analysis}. The results reveal substantial nuances, which are discussed below.

\begin{figure}[htp] 
    \centering
    \includegraphics[width=16cm]{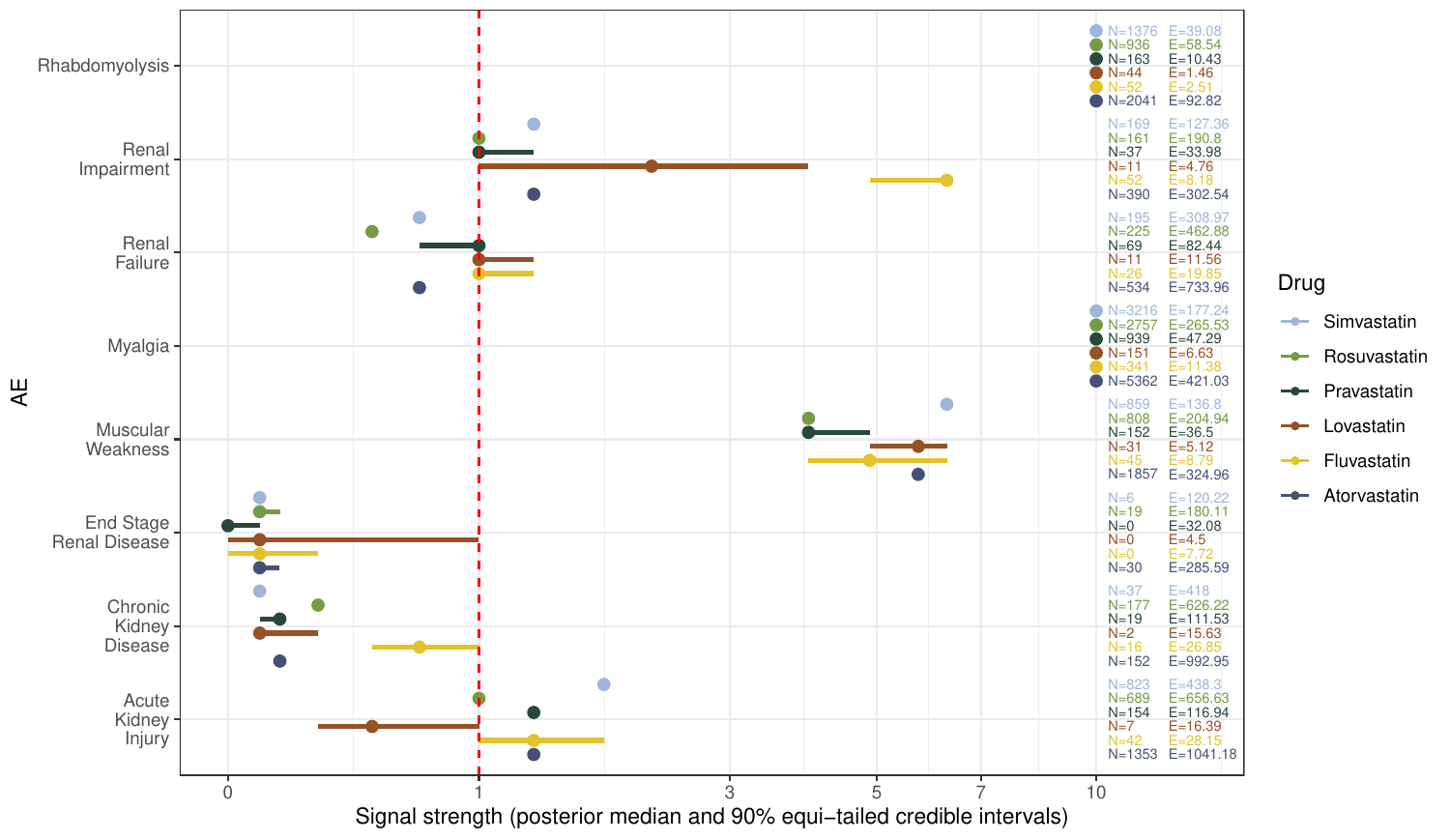}
    \caption{Forest plot visualizing empirical posterior inferences on $8$ prominent AEs across $6$ statin drugs through computed empricial Bayesian posterior distributions of signal strengths $\{\lambda_{ij}\}$ obtained from the general-gamma model employed on the statin-46 dataset. The points and bars represent the posterior medians and $90\%$ equi-tailed credible intervals for the corresponding AE-drug pair-sepcific $\{\lambda_{ij}\}$, and the results from different drugs are color-coded. The red dotted vertical line represents the value ``1''. The texts on the right provide the number of observations as well as the null baseline expected counts under independence for a AE-drug pair.} \label{fig:real_data_analysis}
\end{figure}

First, even among detected signal pairs $(i, j)$ with corresponding $\lambda_{ij} > 1$ with high posterior probability, there can be substantial heterogeneity in the magnitudes of the estimated signal strength parameters $\{\lambda_{ij}\}$. For example, muscular weakness appears to be a signal with a very high posterior probability of having $\lambda_{ij} > 1$ for the AE ($i$) Renal Impairment in each of the drugs ($j$) Simvastatin, Atorvastatin, and Fluvastatin. However, the degree of signal strength is much higher in Fluvastatin, where $\lambda_{ij} \geq 5$ with high posterior probability, compared to the other two drugs where $\lambda_{ij} \in (1, 2)$ with high posterior probabilities. This degree of relevance cannot be rigorously quantified in a pure signal detection analysis that only focuses on signal/non-signal dichotomizations.

Second, the detected signal pairs $(i, j)$ with a high posterior probability of $\lambda_{ij} > 1$ may still demonstrate different levels of uncertainties in their $\lambda_{ij}$ values, even when they have similar observed-to-expected (O/E = $N_{ij}/E_{ij})$ count ratios. For example, the O/E ratio is nearly identical for the AE Muscle Weakness in Rosuvastatin and Pravastatin; however, the posterior distribution for $\lambda_{ij}$ for this AE in Rosuvastatin has much less variability than for Pravastatin due to the considerably larger observed count $N_{ij}$ for the former. This uncertainty assessment is key to making rigorous statistical inferences, which the proposed non-parametric empirical Bayes framework provides.

Third, while permitting substantial flexibility and accommodating heterogeneity through the underlying mixture prior, the proposed general-gamma model incorporates a rigorous regularization/shrinkage mechanism for $\lambda_{ij}$ parameters that avoids overfitting (i.e., unregularized/noisy estimation). This is reflected in the posterior inferences for pairs that have very high observed-to-expected count ratios (O/E). For example, the posterior for $\lambda_{ij}$ in the Rhabdomyolysis-Lovastatin pair is shrunk to $10$ despite an O/E ratio of $44/1.46 \approx 30$.

Finally, notably different inferences can be obtained for the same AE $i$ across different drugs $\{j\}$, even when the same observed counts $\{N_{ij}\}$ are seen across drugs. For instance, Pravastatin, Lovastatin, and Fluvastatin each have $0$ observations corresponding to the AE "End Stage Renal Disease," but their inferred signal strengths are substantially different. While the posterior credible interval for Pravastatin concentrates entirely on $\lambda_{ij}$ values close to zero, for Fluvastatin, it lies within the region $(0, 0.4)$, and for Lovastatin, spans the interval $(0, 1)$. This implies a substantial posterior probability that the observed count for this AE corresponds to a structural zero (hence, non-signal) for Pravastatin. In contrast, for Fluvastatin and Lovastatin, there is less conclusive evidence of the AE-drug pairs being structural zeros but strong evidence of being non-signals. This is a consequence of the differences in the observed counts and corresponding null baseline expected counts for these AE-drug pairs, and the empirical Bayes posterior distributions for $\{\lambda_{ij}\}$ are able to appropriately reflect this.

\section{Discussion and Conclusion}\label{sec:discussion}

This paper extends the scope of existing SRS data mining, which typically focuses only on identifying whether an AE-drug pair is a signal, to a more comprehensive analysis involving signal strength estimation and uncertainty quantification. We leverage a non-parametric empirical Bayesian modeling framework, which permits rigorous and flexible signal detection and signal strength estimation within a single model while retaining computational traceability for model estimation. We introduce three non-parametric empirical Bayes models/methods for pharmacovigilance, namely the KM model, the Efron model, and the general-gamma (mixture) prior model. Each model enables rigorous and flexible modeling of the complex structures in SRS datasets, and we discuss and propose efficient implementation strategies. For enhanced signal strength estimation, we propose a novel estimator for the null baseline expected counts for AE-drug report counts and show that under mild regularity conditions, as the underlying true signal strength increases, the proposed estimator outperforms the existing estimator in terms of asymptotic mean squared error. 

Through extensive simulation, we show that the proposed approaches perform as well or better in signal discovery compared to existing Bayesian/semi-Bayesian models---including BCPNN, the gamma Poisson shrinker, and HDP---while demonstrating substantially better signal strength estimation performance in settings with heterogeneous signals and potential zero-inflation. Using detailed analyses of real SRS data, we illustrate how nuanced insights on AE signals can be obtained through the proposed methods in practice. The implications of these results are as follows. First, if one is solely interested in determining whether certain AE-drug pairs are signals/non-signals, several of the models discussed in the paper, including the existing approaches and the proposed approaches, may be used to achieve reasonable performance. However, much nuance is lost in such a dichotomized significant/non-significant worldview, which contemporary statistical practice argues against in general and which our real data analysis exemplifies. Second, when there is strong prior knowledge about the underlying signal structures that can be quantified through a parametric/semi-parametric prior model, one might use the corresponding parametric/semi-parametric Bayesian model to estimate signal strength for AE-drug pairs (e.g., the 2-gamma model for signal/non-signal structures and the 2-gamma-zi model for true-zero, non-signal, or signal structures). While these models may support reasonable signal discovery, they can fail to yield accurate signal strength estimates if the underlying model assumptions do not hold. The proposed non-parametric empirical Bayes approaches provide sufficient safeguards against model violation over a range of settings with varied signal heterogeneity and zero inflation, enabling robust signal detection \textit{and} signal strength estimation for a large class of real-world SRS datasets.

Several extensions of our approaches are possible. First, the proposed models have some tuning parameters, including the Dirichlet prior parameter $\alpha$ for the general-gamma model, the number of grid points for the KM model, and the $p$ and $c_0$ parameters for the Efron model. In this paper, we employ an approximate leave-one-out cross-validation information criterion-based selection for some of these parameters (e.g., $\alpha$) while using fixed choices for others. While our tuning approach appears reasonably robust for the general-gamma model, future research may investigate more comprehensive approaches to tuning (or estimating within a hierarchical Bayesian framework) the general-gamma model or, more importantly, for the latter two models. Second, the proposed methodology does not explicitly account for joint associations and dependencies between multiple drugs and/or multiple AEs. Future research will extend the current modeling framework to more rigorously address these dependencies. Finally, we aim to apply our current methods to more extensive real-world data for deeper scientific insights in application-focused studies.

\bibliographystyle{plain}

\clearpage 
\section*{APPENDIX}

\begin{appendix} 
\renewcommand{\thesection}{S\arabic{section}}

\section{Derivation of AMSEs for null value estimations}

Based on the conditional multinomial model (2) in  section 4.4.2:
\begin{equation}
    N_{ij} \mid N_{\bullet \bullet} \sim \text{Multinomial}\left(N_{\bullet \bullet}, p_{ij} = \frac{\lambda_{ij}p_{i\bullet}^*p^*_{\bullet j}}{\sum_{k=1}^I\sum_{l=1}^J \lambda_{kl}p_{k\bullet}^*p_{\bullet l}^*}\right),
\end{equation}
 where $p_{i\bullet}^*$ and $p_{\bullet j}^*$ are the null baseline marginal probabilities of AE-$i$ drug-$j$ under independence, respectively, and $N_{\bullet \bullet}$ is the total number of report counts (grand total) in the observed SRS dataset. In the following, we derive the AMSE of $\hat E_{ij}/ E_{ij}$. 

First, let us define $N_A = N_{ij}$, $N_B = \sum_{j'\neq j} N_{ij'}$ and  $N_C = \sum_{i'\neq i} N_{i'j}$ as well as $p_A = p_{ij}$, $p_B = \sum_{j'\neq j} p_{ij'}$,  $p_C = \sum_{i'\neq i} p_{i'j}$. Thus, we can write 
\[\frac{\hat E_{ij}}{E_{ij}} = N_{\bullet \bullet}\frac{N_A+N_B}{N_{\bullet \bullet}}\frac{N_A+N_C}{N_{\bullet \bullet}}/(N_{\bullet\bullet}p_{i\bullet}^*p_{\bullet j}^*) = \frac{N_A+N_B}{N_{\bullet \bullet}}\frac{N_A+N_C}{N_{\bullet \bullet}}/(p_{i\bullet}^*p_{\bullet j}^*).\]
By the central limit theorem, the asymptotic distribution of the vector random variable $\hat p_{ABC} = [N_A/N_{\bullet \bullet} \quad N_B/N_{\bullet \bullet} \quad N_C/N_{\bullet \bullet}]^T$ is:
\begin{align*}
    &\sqrt{N_{\bullet \bullet}}\left(\left[\frac{N_A}{N_{\bullet\bullet}}\; \frac{N_B}{N_{\bullet\bullet}}\;\frac{N_C}{N_{\bullet\bullet}}\right]^T - [p_A\; p_B\; p_C]^T\right) \quad \xrightarrow{d} \quad N(0,\Sigma),
\end{align*}
where
\[\Sigma=\begin{bmatrix}
p_A(1-p_A) & -p_A p_B & -p_A p_C \\
-p_A p_B & p_B(1-p_B) & -p_B p_C \\
-p_A p_C & -p_B p_C & p_C(1-p_C) 
\end{bmatrix}  \]
To derive the asymptotic distribution of  $\hat E_{ij}/ E_{ij} = \frac{N_A+N_B}{N_{\bullet \bullet}}\frac{N_A+N_C}{N_{\bullet \bullet}}/(p_{i\bullet}^*p_{\bullet j}^*)$, we define a function $g(p_A,p_B,p_C) = (p_A+p_B)(p_A+p_C)$; Its partial derivative vector is:
\[\nabla g = [2p_A+p_B+p_C\quad p_A+p_C \quad p_B+p_C]^T\]
By the Delta method, we have:
\[\sqrt{N_{\bullet \bullet}}\left(\frac{N_A+N_B}{N_{\bullet \bullet}}\frac{N_A+N_C}{N_{\bullet \bullet}} - (p_A+p_B)(p_A+p_C)\right) \quad \xrightarrow{d} \quad N(0,\nabla g^T\Sigma\nabla g),\]
where $\nabla g^T\Sigma\nabla g = (p_A+p_B)(p_A+p_C)[4p_A+p_B+p_C-4(p_A+p_B)(p_A+p_C)]$. Thus,
\[\sqrt{N_{\bullet \bullet}}\left(\frac{\hat E_{ij}}{E_{ij}} - \frac{(p_A+p_B)(p_A+p_C)}{p_{i\bullet}^*p_{\bullet j}^*}\right) \quad \xrightarrow{d} \quad N\left(0,\frac{\nabla g^T\Sigma\nabla g}{(p_{i\bullet}^*p_{\bullet j}^*)^2}\right),\]

The asymptotic MSE of $\hat E_{ij}/ E_{ij}$ is the asymptotic expectation of $(\hat E_{ij}/ E_{ij} -1)^2$.  Plugin in $p_A = p_{ij}$, $p_B = \sum_{j'\neq j} p_{ij'}$,  $p_C = \sum_{i'\neq i} p_{i'j}$, after some simplifications, we obtain the asymptotic MSE of $\hat E_{ij}/ E_{ij}$:
\begin{equation} \label{eqn:AMSE_hatE}
    \text{AMSE}\left( \frac{\hat E_{ij}}{E_{ij}} \right)\\ 
= \frac{\bar \lambda_{i\bullet} \bar \lambda_{\bullet j}}{\bar \lambda_{\bullet \bullet}^4N_{\bullet \bullet}}\left[\bar \lambda_{\bullet \bullet}\left( 2\lambda_{ij}+ \frac{\bar \lambda_{i \bullet}}{p_{i\bullet}^*} + \frac{\bar \lambda_{\bullet j}}{p_{\bullet j}^*}\right) - 4\bar \lambda_{i \bullet}\bar \lambda_{\bullet j} \right] + \left[ \frac{\bar \lambda_{i \bullet}\bar \lambda_{\bullet j}}{\bar \lambda_{\bullet \bullet}^2} -1\right]^2,
\end{equation}
where $\bar \lambda_{i \bullet} = \sum_{l=1}^J\lambda_{il}p_{\bullet l}^*$, $\bar \lambda_{\bullet j} = \sum_{k=1}^I\lambda_{kj}p_{k\bullet}^*$ and $\bar \lambda_{\bullet \bullet} = \sum_{k=1}^I\sum_{l=1}^J\lambda_{kl}p_{k\bullet}^*p_{\bullet l}^*$.

The AMSE of $\tilde E_{ij}/E_{ij}$ is derived similarly. We write 
\[\frac{\tilde E_{ij}}{E_{ij}} = \frac{\frac{N_{iJ}}{N_{\bullet\bullet}}\frac{N_{Ij}}{N_{\bullet\bullet}}}{\frac{N_{IJ}}{N_{\bullet\bullet}}(p_{i\bullet}^*p_{\bullet j}^*)}.\]
By the CLT, we have:
\begin{align*}
    &\sqrt{N_{\bullet \bullet}}\left(\left[\frac{N_{iJ}}{N_{\bullet\bullet}}\; \frac{N_{Ij}}{N_{\bullet\bullet}}\;\frac{N_{IJ}}{N_{\bullet\bullet}}\right]^T - [p_{Ij}\; p_{iJ}\; p_{IJ}]^T\right) \quad \xrightarrow{d} \quad N(0,\Sigma_2),
\end{align*}
where
\[\Sigma_2=\begin{bmatrix}
p_{Ij}(1-p_{Ij}) & -p_{Ij} p_{iJ} & -p_{Ij} p_{IJ} \\
-p_{Ij} p_{iJ} & p_{iJ}(1-p_{iJ}) & -p_{iJ} p_{IJ} \\
-p_{Ij} p_{IJ} & -p_{iJ} p_{IJ} & p_{IJ}(1-p_{IJ}) 
\end{bmatrix}  \]
To derive the asymptotic distribution of  $\tilde E_{ij}/ E_{ij} = \frac{\frac{N_{iJ}}{N_{\bullet\bullet}}\frac{N_{Ij}}{N_{\bullet\bullet}}}{\frac{N_{IJ}}{N_{\bullet\bullet}}(p_{i\bullet}^*p_{\bullet j}^*)}$, We define a function $g_2(p_{Ij},p_{iJ},p_{IJ}) = \frac{p_{Ij}p_{iJ}}{p_{IJ}}$; Its partial derivative vector is:
\[\nabla g_2 = [p_{iJ}/p_{IJ} \quad p_{Ij}/p_{IJ} \quad -p_{iJ}p_{Ij}/p_{IJ}^2]^T.\]
Following the same steps as $\hat E_{ij}$, the AMSE of $\tilde E_{ij}/E_{ij}$ is:
\begin{equation} \label{eqn:AMSE_tildeE}
   \text{AMSE} \left(\frac{\tilde E_{ij}}{E_{ij}} \right) =\frac{1}{N_{\bullet \bullet}}\left[\frac{1}{p_{I\bullet}^*p_{\bullet J}^*\bar \lambda_{\bullet \bullet}}+\frac{1}{p_{i\bullet}^*p_{\bullet J}^*\bar \lambda_{\bullet \bullet}}+\frac{1}{p_{I\bullet}^*p_{\bullet j}^*\bar \lambda_{\bullet \bullet}}-\frac{1}{\bar \lambda_{\bullet \bullet}^2} \right] +\left(\frac{1}{\bar \lambda_{\bullet \bullet}}-1\right)^2. 
\end{equation}

\subsection{proof of Theorem 2}

First, we define  $S_{ij}(\lambda_{ij}, \bar \lambda_{i\bullet}, \bar \lambda_{\bullet j}, \bar \lambda_{\bullet \bullet})$ as the scaled difference between the asymptotic MSE of $\hat E_{ij}/E_{ij}$ \eqref{eqn:AMSE_hatE} and $\tilde E_{ij}/E_{ij}$ \eqref{eqn:AMSE_tildeE}:
\begin{equation*} 
\begin{split}
\text{AMSE}(\hat E_{ij}/E_{ij}) - \text{AMSE}(\tilde E_{ij}/E_{ij}) = \frac{1}{N_{\bullet \bullet}\lambda_{\bullet \bullet}^4}S_{ij}(\lambda_{ij}, \bar \lambda_{i\bullet}, \bar \lambda_{\bullet j}, \bar \lambda_{\bullet \bullet}),\\
\end{split}
\end{equation*}
where $S_{ij}$ is given below and is deposited as the summation of $s_{ij1}$ (contains the first two terms of $S_{ij}$) and $s_{ij2}$ (contains the third term of $S_{ij}$):
\begin{equation*} 
\begin{split}
S_{ij}(\lambda_{ij}, \bar \lambda_{i\bullet}, \bar \lambda_{\bullet j}, \bar \lambda_{\bullet \bullet}) = &\bar \lambda_{i\bullet} \bar \lambda_{\bullet j} \bar \lambda_{\bullet \bullet}\left( 2\lambda_{ij} + \frac{\bar \lambda_{\bullet j}}{p_{\bullet j}^*} + \frac{\bar \lambda_{i\bullet}}{p_{i\bullet}^*} - \frac{4\bar \lambda_{i\bullet} \bar \lambda_{\bullet j}}{\bar \lambda_{\bullet \bullet}}\right) + \bar \lambda_{\bullet \bullet}^2 +  \\
&N_{\bullet \bullet}\left[ \bar \lambda_{i\bullet}^2 \bar \lambda_{\bullet j}^2 - 2\bar \lambda_{i\bullet} \bar \lambda_{\bullet j}\bar \lambda_{\bullet \bullet}^2 + (2 - \frac{a_{ij}}{N_{\bullet \bullet}})\bar \lambda_{\bullet \bullet}^3 - \bar \lambda_{\bullet \bullet}^2\right]\\
&= s_{ij1}(\lambda_{ij}, \bar \lambda_{i\bullet}, \bar \lambda_{\bullet j}, \bar \lambda_{\bullet \bullet}) + s_{ij2}(\bar \lambda_{i\bullet}, \bar \lambda_{\bullet j}, \bar \lambda_{\bullet \bullet}),
\end{split}
\end{equation*}

and $a_{ij} = \frac{1}{p_{I\bullet}^*p_{\bullet J}^*}+ \frac{1}{p_{i\bullet}^*p_{\bullet J}^*}+\frac{1}{p_{I\bullet}^*p_{\bullet j}^*}$.  Since both terms of $s_{ij1}$ are positive (see equation \eqref{eqn:AMSE_hatE}: the first term is the scaled asymptotic variance of $\hat E_{ij}$),  $s_{ij1}(\lambda_{ij}, \bar \lambda_{i\bullet}, \bar \lambda_{\bullet j}, \bar \lambda_{\bullet \bullet})$ 
 is always positive. Thus, if $s_{ij2}(\bar \lambda_{i\bullet}, \bar \lambda_{\bullet j}, \bar \lambda_{\bullet \bullet})$ is positive, then $S_{ij} > 0$ meaning that the proposed null value estimator $\tilde E_{ij}$ performs better then the original estimator $\hat E_{ij}$ in terms of scaled AMSE. In the following, we investigate the condition for  $\bar \lambda_{i\bullet}$, $\bar \lambda_{\bullet j}$ and $\bar \lambda_{\bullet \bullet}$ lead to positive $s_{ij2}(\bar \lambda_{i\bullet}, \bar \lambda_{\bullet j}, \bar \lambda_{\bullet \bullet})$. By the regularity condition (D), we establish a upper and lower bound for $a_{ij}$:
\begin{align}
    a_{ij} &= \frac{N_{\bullet\bullet}}{N_{\bullet\bullet}p_{I\bullet}^*p_{\bullet J}^*}+ \frac{N_{\bullet\bullet}}{N_{\bullet\bullet}p_{i\bullet}^*p_{\bullet J}^*}+\frac{N_{\bullet\bullet}}{N_{\bullet\bullet}p_{I\bullet}^*p_{\bullet j}^*} \nonumber\\
    & = N_{\bullet\bullet}\left(\frac{1}{E_{iJ}}+ \frac{1}{E_{Ij}}+\frac{1}{E_{IJ}} \right)\nonumber\\
    &\leq \frac{3N_{\bullet\bullet}}{E_{\min}}. \nonumber\\
    &\Rightarrow \frac{3}{N_{\bullet \bullet}}<\frac{a_{ij}}{N_{\bullet \bullet}}<\frac{3}{E_{\min}},
\end{align}
where $E_{\min}$ is the smallest expected count among reference cells. When $\bar \lambda_{i\bullet}\bar \lambda_{\bullet j} = \bar \lambda_{\bullet \bullet}^2$, it straightforward to show that $ s_{ij2}(\bar \lambda_{i\bullet}, \bar \lambda_{\bullet j}, \bar \lambda_{\bullet \bullet})<0$. The partial derivative $\frac{\partial s_{ij2}(\bar \lambda_{i\bullet}, \bar \lambda_{\bullet j}, \bar \lambda_{\bullet \bullet})}{\partial (\bar \lambda_{i\bullet}\bar \lambda_{\bullet j}) }$ shows that $s_{ij2}(\bar \lambda_{i\bullet}, \bar \lambda_{\bullet j}, \bar \lambda_{\bullet \bullet})$ is monotonically increasing in $\bar \lambda_{i\bullet}\bar \lambda_{\bullet j}$ if $\bar \lambda_{i\bullet}\bar \lambda_{\bullet j}> \bar \lambda_{\bullet \bullet}^2$. Now, consider substitution $\bar \lambda_{i\bullet}\bar \lambda_{\bullet j} = x$ and consider a fixed $\bar \lambda_{\bullet \bullet}$. Then, $s_{ij2}(x)$ is a quadratic function of $x > 1$:
\[s_{ij2}(x) = x^2 - 2\bar \lambda_{\bullet \bullet}^2 x +\bar \lambda_{\bullet \bullet}^2[(2 - a_{ij}/N_{\bullet \bullet})\bar \lambda_{\bullet \bullet} -1].\]
The square root determinant of the equation $s_{ij2}(x) = 0$ is $\sqrt{\Delta} = 2\bar \lambda_{\bullet \bullet}^2\sqrt{\left( 1- \frac{1}{\bar \lambda_{\bullet \bullet}}\right)^2 + \frac{a_{ij}}{N_{\bullet \bullet}\bar \lambda_{\bullet \bullet}}} > 0$, which guarantees that equation $s_{ij2}(x) = 0$ has two roots ($A_{ij}, B_{ij}$) in real number that are:
\[
A_{i^*j^*} =\left\{1 -  \sqrt{\left(1-\frac{1}{\bar \lambda_{\bullet \bullet}} \right)^2 + \left( \frac{1}{E_{i^* J}}+\frac{1}{E_{I j^*}}+\frac{1}{E_{IJ}}\right) \frac{1}{\bar \lambda_{\bullet \bullet}}}\right\} \bar \lambda_{\bullet \bullet}^2,
\]
and 
\[
B_{i^*j^*} =\left\{1 +  \sqrt{\left(1-\frac{1}{\bar \lambda_{\bullet \bullet}} \right)^2 + \left( \frac{1}{E_{i^* J}}+\frac{1}{E_{I j^*}}+\frac{1}{E_{IJ}}\right) \frac{1}{\bar \lambda_{\bullet \bullet}}}\right\} \bar \lambda_{\bullet \bullet}^2
\] 
The derivatives of $A_{ij}$ and $B_{ij}$ with respect to $\bar \lambda_{\bullet \bullet}$ show that $A_{ij}$ and $B_{ij}$ are monotonically increasing for $\bar \lambda_{\bullet \bullet} > 1$ and $\frac{3}{N_{\bullet \bullet}}<\frac{a_{ij}}{N_{\bullet \bullet}}<\frac{3}{E_{\min}}$. Then, the solution for $s_{ij2}(x) > 0$ is: $x \in [1, A_{ij}) \cup (B_{ij}, \infty)$ for $A_{ij} \geq 1$; $x \in (B_{ij}, \infty)$ for $A_{ij} < 1$.

\subsection{Corollary 1 and associated proofs}

\begin{corollary} \label{corollary1}
    Suppose that all the regularity conditions of Theorem 2 are satisfied. Then, we have the following for non-reference cells in $i^*$-th row $\{(i^*, j'): j'\neq J\}$ (similar result can be derived for non-reference cells in a column):
    \begin{enumerate}
        \item If the null baseline marginal probabilities $\{p_{i \bullet}^*\}$ and $\{p_{\bullet j}^*\}$ satisfy $0 < p_{i\bullet}^*, p_{\bullet j}^*< \frac{1}{2}$ for all $i = 1, \dots, (I-1)$ and for all $j = 1, \dots, (J-1)$ and the associated column-wise average signal strength equals the table-wise average signal strength: $\bar \lambda_{\bullet j'} = \bar \lambda_{\bullet \bullet}$, then, for a large enough $\lambda_{i^*j^*}$ for some $i^* \in \{1, \dots, I-1\}$ and $j^* \in \{1, \dots, J-1\}$ such that,
    \begin{align}\label{cor(1)}
        &\bar \lambda_{i^* \bullet} >\left\{1+ \sqrt{\left(1-\frac{1}{\bar \lambda_{\bullet \bullet}} \right)^2 + \frac{3}{E_{\min}\bar \lambda_{\bullet \bullet}}}\right\}\bar \lambda_{\bullet \bullet}, 
    \end{align}
    the estimating performance of $\tilde E$ is better than $\hat E$:
    \begin{align*}
        &\text{AMSE}(\hat E_{i^* j'}/E_{i^* j'}) >  \text{AMSE}(\tilde E_{i^* j'}/E_{i^* j'}) \quad  \text{for } j'\neq J.
    \end{align*}

    \item We further assume that the null baseline marginal probabilities $\{p_{i \bullet}^*\}$ and $\{p_{\bullet j}^*\}$ satisfy $0 < p_{i\bullet}^*, p_{\bullet j}^*< \frac{1}{7}$ for all $i = 1, \dots, (I-1)$ and for all $j = 1, \dots, (J-1)$. Then, if the the associated column-wise average signal strength equals 1: $\bar \lambda_{\bullet j'} = 1$. Among non-reference cells in $i^*$-th row $\{(i^*, j'): j'\neq J\}$, for a large enough $\lambda_{i^*j^*}$ for some $i^* \in \{1, \dots, I-1\}$ and $j^* \in \{1, \dots, J-1\}$ such that,
    \begin{align}\label{cor(2)}
        &\bar \lambda_{i^* \bullet} >\left\{1+ \sqrt{\left(1-\frac{1}{\bar \lambda_{\bullet \bullet}} \right)^2 + \frac{3}{E_{\min}\bar \lambda_{\bullet \bullet}}}\right\}\bar \lambda_{\bullet \bullet}^2, 
    \end{align}
    the estimating performance of $\tilde E$ is better than $\hat E$:
    \begin{align*}
        &\text{AMSE}(\hat E_{i^* j'}/E_{i^* j'}) >  \text{AMSE}(\tilde E_{i^* j'}/E_{i^* j'}) \quad  \text{for } j'\neq J.
    \end{align*}
    \end{enumerate}
\end{corollary}
\begin{proof}
    (1) and (2) can be simply obtained by inequality manipulations based on Theorem 2. Based on the solution for $s_{ij2}(x) > 0$, (1) and (2) of the Corollary 1 are obtained as follows:

(1)If there exists a large enough $\lambda_{i^*j^*}$ such that
\begin{align*}
    &\bar \lambda_{i^* \bullet} >\left\{1+ \sqrt{\left(1-\frac{1}{\lambda_{\bullet \bullet}} \right)^2 + \frac{3}{E_{\min}\lambda_{\bullet \bullet}}}\right\} \bar\lambda_{\bullet \bullet},
\end{align*}
then, for all the non-reference cells in $i^*$-th row $\{(i^*,j'): j' \neq J\}$, we have:
\begin{align*}
    \bar \lambda_{i^* \bullet} \bar \lambda_{ \bullet j'} &> \left\{1+ \sqrt{\left(1-\frac{1}{\bar \lambda_{\bullet \bullet}} \right)^2 + \frac{3}{E_{\min}\bar \lambda_{\bullet \bullet}}}\right\}\bar \lambda_{\bullet \bullet}\bar \lambda_{ \bullet j'} \\
    &>\left\{1+ \sqrt{\left(1-\frac{1}{\bar \lambda_{\bullet \bullet}} \right)^2 + \frac{3}{E_{\min}\bar \lambda_{\bullet \bullet}}}\right\}\bar \lambda_{\bullet \bullet}^2 \\
    &= B_{i^*j'}
\end{align*}
Therefore, $\text{AMSE}(\hat E_{i^* j'}/E_{i^* j'}) > \text{AMSE}(\tilde E_{i^* j'}/E_{i^* j'}) \quad \text{for } j'\neq J.$ Similar conclusion can be established for non-reference cells that are in the same column. Next, let us check the existence of $\bar \lambda_{i \bullet} > \frac{B_{i^*j'}}{\bar \lambda_{\bullet \bullet}}$

In a simplified case where the cell $(i, j)$ is the only signal cell, we have $\bar \lambda_{\bullet \bullet} = 1 + p^*_{i \bullet}p^*_{\bullet j}(\lambda_{ij}-1)$ and $\bar \lambda_{i^* \bullet} = 1 + p^*_{i \bullet}(\lambda_{ij}-1)$. Then, 
\[\bar \lambda_{i^* \bullet} = \frac{\bar \lambda_{\bullet \bullet}}{p^*_{\bullet j}} +1 - \frac{1}{p^*_{\bullet j}}.\]
We provide a upper bound of $B_{ij}$ for $\bar \lambda_{\bullet \bullet}>1$:
\begin{align*}
    B_{ij} &= \left(1+ \sqrt{\left( 1- \frac{1}{\bar \lambda_{\bullet \bullet}}\right)^2 + \frac{a_{ij}}{N_{\bullet \bullet}\bar \lambda_{\bullet \bullet}}}\right)\bar \lambda_{\bullet \bullet}^2\\
    &= \bar \lambda_{\bullet \bullet}^{2}\ +\ \bar \lambda_{\bullet \bullet} \sqrt{\left(\bar \lambda_{\bullet \bullet}-1\right)^{2}+\frac{a_{ij}\bar \lambda_{\bullet \bullet}}{N_{\bullet \bullet}}} \\
    &< \bar \lambda_{\bullet \bullet}^{2}\ +\ \bar \lambda_{\bullet \bullet} \sqrt{\left(\bar \lambda_{\bullet \bullet}-1\right)^{2}+\bar \lambda_{\bullet \bullet}}\\
    &=\bar \lambda_{\bullet \bullet}^{2}\ +\ \bar \lambda_{\bullet \bullet} \sqrt{\bar \lambda_{\bullet \bullet}^2-\bar \lambda_{\bullet \bullet}+1} \\
    &< \bar \lambda_{\bullet \bullet}^{2}\ + \ \bar \lambda_{\bullet \bullet}^{2} = 2\bar \lambda_{\bullet \bullet}^{2}
\end{align*}
If $\bar \lambda_{i^* \bullet}$ is greater than $ 2\bar \lambda_{\bullet \bullet}$, it is greater than $\frac{B_{i^*j'}}{\bar \lambda_{\bullet \bullet}}$.
\begin{align}\label{eqn:check_existence1}
    \bar \lambda_{i^* \bullet} &= \frac{\bar \lambda_{\bullet \bullet}}{p^*_{\bullet j}} +1 - \frac{1}{p^*_{\bullet j}} > 2\bar \lambda_{\bullet \bullet} \nonumber\\
    &\Rightarrow (\frac{1}{p^*_{\bullet j}} - 2)\bar \lambda_{\bullet \bullet} -\frac{1}{p^*_{\bullet j}} +1 > 0.
\end{align}
We note that $\frac{1}{p^*_{\bullet j}} - 2 > 0$. 
Thus, for a large enough $\lambda_{ij}$, we have $\bar \lambda_{i \bullet} > \frac{B_{i^*j'}}{\bar \lambda_{\bullet \bullet}}$. Under a general situation in the context of SRS data, $\bar \lambda_{\bullet \bullet} = c_{\bullet \bullet} + p^*_{i \bullet}p^*_{\bullet j}(\lambda_{ij}-1)$ and $\bar \lambda_{i^* \bullet} = c_{i \bullet} + p^*_{i \bullet}(\lambda_{ij}-1)$, where $c_{\bullet \bullet} = \sum_{i'\neq i}\sum_{j' \neq j}\lambda_{i'j'}p^*_{i' \bullet}p^*_{\bullet j'} + p^*_{i \bullet}p^*_{\bullet j}\geq 1$ and $c_{i \bullet} = \sum_{j' \neq j}\lambda_{ij'}p^*_{\bullet j'} + p^*_{\bullet j} \geq 1$, similar condition as equation \eqref{eqn:check_existence1} can be established.

(2) If there exists a large enough $\lambda_{i^*j^*}$ such that
\begin{align*}
    &\bar \lambda_{i^* \bullet} >\left\{1+ \sqrt{\left(1-\frac{1}{\lambda_{\bullet \bullet}} \right)^2 + \frac{3}{E_{\min}\lambda_{\bullet \bullet}}}\right\} \bar\lambda_{\bullet \bullet}^2,
\end{align*}
then, for all the non-reference cells in $i^*$-th row $\{(i^*,j'): j' \neq J\}$, we have:
\begin{align*}
    \bar \lambda_{i^* \bullet} \bar \lambda_{ \bullet j'} &> \left\{1+ \sqrt{\left(1-\frac{1}{\bar \lambda_{\bullet \bullet}} \right)^2 + \frac{3}{E_{\min}\bar \lambda_{\bullet \bullet}}}\right\}\bar \lambda_{\bullet \bullet}^2 \bar \lambda_{ \bullet j'} \\
    &>\left\{1+ \sqrt{\left(1-\frac{1}{\bar \lambda_{\bullet \bullet}} \right)^2 + \frac{3}{E_{\min}\bar \lambda_{\bullet \bullet}}}\right\}\bar \lambda_{\bullet \bullet}^2 \\
    &= B_{i^*j'}
\end{align*}
In a simplified case where the cell $(i, j)$ is the only signal cell, we have $\bar \lambda_{\bullet \bullet} = 1 + p^*_{i \bullet}p^*_{\bullet j}(\lambda_{ij}-1)$ and $\bar \lambda_{i^* \bullet} = 1 + p^*_{i \bullet}(\lambda_{ij}-1)$. Then, 
\[\bar \lambda_{i^* \bullet} = \frac{\bar \lambda_{\bullet \bullet}}{p^*_{\bullet j}} +1 - \frac{1}{p^*_{\bullet j}}.\]
If $\bar \lambda_{i^* \bullet}$ is greater than $ 2\bar \lambda_{\bullet \bullet}^{2}$, it is greater than $B_{ij}$.
\begin{align}\label{eqn:check_existence}
    \bar \lambda_{i^* \bullet} &= \frac{\bar \lambda_{\bullet \bullet}}{p^*_{\bullet j}} +1 - \frac{1}{p^*_{\bullet j}} > 2\bar \lambda_{\bullet \bullet}^{2} \nonumber\\
    &\Rightarrow 2\bar \lambda_{\bullet \bullet}^{2} - \frac{\bar \lambda_{\bullet \bullet}}{p^*_{\bullet j}} +\frac{1}{p^*_{\bullet j}} -1 < 0.
\end{align}
The associated determinant is greater than 0: $\Delta(p^*_{\bullet j}) = \frac{1}{p^{*2}_{\bullet j}} - \frac{8}{p^*_{\bullet j}} + 8 > 0$ for $p_{\bullet j}^* < \frac{1}{7}$. Then, the solution for the inequality \eqref{eqn:check_existence} is:
\begin{equation} \label{eqn:sol_check_existence}
    \bar \lambda_{\bullet \bullet} \in \left(\frac{\frac{1}{p^*_{\bullet j}} - \sqrt{\Delta(p^*_{\bullet j})}}{4}, \frac{\frac{1}{p^*_{\bullet j}} + \sqrt{\Delta(p^*_{\bullet j})}}{4} \right).
\end{equation}
Thus, for a large enough $\lambda_{ij}$ such that the condition \eqref{eqn:sol_check_existence} is satisfied, we have $\bar \lambda_{i \bullet} > B_{ij}$. For a general situation in the context of SRS data, $\bar \lambda_{\bullet \bullet} = c_{\bullet \bullet} + p^*_{i \bullet}p^*_{\bullet j}(\lambda_{ij}-1)$ and $\bar \lambda_{i^* \bullet} = c_{i \bullet} + p^*_{i \bullet}(\lambda_{ij}-1)$, where $c_{\bullet \bullet} = \sum_{i'\neq i}\sum_{j' \neq j}\lambda_{i'j'}p^*_{i' \bullet}p^*_{\bullet j'} + p^*_{i \bullet}p^*_{\bullet j}\geq 1$ and $c_{i \bullet} = \sum_{j' \neq j}\lambda_{ij'}p^*_{\bullet j'} + p^*_{\bullet j} \geq 1$. Similar condition as \eqref{eqn:sol_check_existence} can be established.  
\end{proof}

\section{Improvement of numeric parameter estimation and initial value selection for 2-gamma and 2-gamma-zi model}

\subsection{Improved numeric parameter estimation}
In the "openEBGM" package, parameter estimation of $\{\alpha_1, \beta_1, \alpha_2, \beta_2, \omega\}$ is performed by directly optimizing the joint log marginal likelihood with build-in optimizers in R, e.g. "nlminb", "nlm" or "optim". We note that the parameters $\{\alpha_1, \beta_1, \alpha_2, \beta_2, \omega\}$ in the prior distribution are bounded--- $\alpha_1, \beta_1, \alpha_2, \beta_2 > 0$ and $0 < \omega < 1$. Optimizing with bounded parameters, optimizers can struggle near the boundary of the parameter space, leading to poor gradient estimation and unstable behavior. In our implementation of the 2-gamma and 2-gamma-zi model, we reparameterize parameters by $\alpha_1 = \log(\alpha_1^*)$; $\beta_1 = \log(\beta_1^*)$; $\alpha_2 = \log(\alpha_2^*)$; $\beta_2 = \log(\beta_2^*)$; and $p = expit(p^*)$, where $\{ \alpha_1^*, \beta_1^*, \alpha_2^*, \beta_2^*, p^*\} \in R^5$ in our implementation, which leads to better conditioning of the optimization problem and prevents issues such as large gradients or vanishing gradients, that are common in constrained spaces.

\subsection{Kmeans-based method: KM1}

The 2-gamma model assumes a mixture of two gamma distribution components prior to $\{\lambda_{ij}\}$. Although we $\{\lambda_{ij}\}$ are unobservable, we can use the empirical $O/E$ values $\tilde \lambda_{ij} = N_{ij}/E_{ij}$ as a proxy to reveal the latent bimodal structure of $\{\lambda_{ij}\}$. Applying the K-means algorithm with two clusters to $\{\tilde \lambda_{ij}\}$ helps partition $\{\tilde \lambda_{ij}\}$ values into two groups, aligning with the two gamma components of the prior distribution, thereby facilitating parameter initialization.

Let $\bar{x}_1$ be the mean of $\{\tilde \lambda_{ij}\}$ labeled as cluster 1 and $\tilde{s}_1$ is the associated the standard deviation of cluster 1. Similarly, $\bar{x}_2$ and $\tilde{s}_2$ are the mean and standard deviation of $\{\tilde \lambda_{ij}\}$ labeled as cluster 2 by the Kmeans algorithm. Then, the initial values of $\{\alpha_1, \beta_1, \alpha_2, \beta_2, \omega\}$ are given by:

\begin{equation} \label{eqn:km1}
\begin{split}
&\alpha_1 = \bar{x}_1^2/\tilde{s}_1; \quad \beta_1 = \bar{x}_1/\tilde{s}_1;\\
&\alpha_2 = \bar{x}_2^2/\tilde{s}_2; \quad \beta_2 = \bar{x}_2/\tilde{s}_2; \\
& \omega = \frac{\#\{ \tilde \lambda_{ij} :\text{ labeled as 1}\}}{\# \{\tilde \lambda_{ij}\}}.
\end{split}
\end{equation}

\subsection{Kmeans-based method: KM2}

The KM2 method takes structural zeros into account, which commonly appear in SRS data. By adopting the zero-inflated Poisson model,
$$z_{ij^*}\sim Bernoulli(\omega_{j^*});$$
\[    
N_{ij}|z_{ij} = z\sim\left\{
\begin{array}{ll}
      \delta_0,  & z=1 \\
      \text{Poisson($\lambda_{ij}\times E_{ij}$)}, & z=0, \\
\end{array} 
\right. \]
we estimate the number of structural zeros $n_0$ using a profile likelihood maximization technique proposed by Chakraborty et al. \cite{chakraborty2022use} under this model. Then, we compute the empirical $O/E$ values $\tilde \lambda_{ij} = N_{ij}/E_{ij}$ as the KM1 method and remove $n_0$ zeros from $\{\tilde \lambda_{ij}\}$. We apply the K-means algorithm with two clusters to the rest of $\{\tilde \lambda_{ij}\}$. The initial values of $\{\alpha_1, \beta_1, \alpha_2, \beta_2, \omega\}$ are computed from the Kmeans clustering results in the same way as the KM1 method (see equation \eqref{eqn:km1}).

\subsection{Method of Moment  approach}

The method of moments approach provides an initial set of values, $\{\alpha_1, \beta_1, \alpha_2, \beta_2, \omega\}$, for the 2-gamma model by solving a system of equations derived from equating the scaled two negative binomial mixture (marginal distribution) factorial moments with their sample version moments. We note that the $m$-th order scaled factorial moment is invariant of $(i, j)$-th observation and is given as:
\[E\left(\frac{(N_{ij})_m}{E_{ij}^m}\right) = \omega\beta_1^m\prod_{k=0}^{m-1}(\alpha_1+k)+(1-\omega)\beta_2^m\prod_{k=0}^{m-1}(\alpha_2+k),\]
where $(N_{ij})_m = \frac{N_{ij}!}{(N_{ij}-m)!}$. When $N_{ij}<m$, $(N_{ij}-m)!$ is undefined. Therefore, we excluded the observations ($\{N_{ij}\}$) that are smaller than $m$  from the computation of the corresponding scaled sample factorial moment, and we define $(N_{ij})_m = 0$ for $N_{ij}<m$. Then, the sample version of the $m$-th order scaled factorial moment is:
$${\frac{(\bar{N})_m}{E^m}}=\frac{1}{IJ-\#\{N_{ij}<m\}}\sum_i\sum_j\frac{(N_{ij})_m}{E_{ij}^m}.$$
We obtain a set of $\{\alpha_1, \beta_1, \alpha_2, \beta_2, \omega\}$ values by solving the following system of 5 equations:
$$\frac{(\bar{N})_m}{E^m} = E\left(\frac{(N_{ij})_m}{E_{ij}^m}\right), \quad \text{for $m=1,2,3,4,5$}.$$
In our implementation, this system of equations is solved by an R package "nleqslv".

\section{Details of the sparse general gamma model}

The details of two types of data augmentation mentioned in the section 4.4.1 are provided below.

The first data augmentation is used to estimate $\{\omega_k, h_k\}$ and introduces the unobserved allocation variable $S = \{S_{11}, S_{12}, \dots, S_{IJ}\}$ where $S_{ij} = (S_{ijk}: k = 1, \dots, K)$ a $K$-component $0-1$ vector that indicates which mixture component the $(i, j)$-th observation belongs to; $S_{ijk} = 1$ if observation $(i, j)$ belongs to component $k$ and $S_{ijk} = 0$, otherwise. We denote the complete data as $X_1 = \{S, N\}$ with the corresponding complete-data likelihood:
\[
p(X_1) = \frac{\Gamma(K\alpha)}{[\Gamma(\alpha)]^K}\prod_{k=1}^K \omega_k^{\alpha-1} \prod_{i=1}^I\prod_{j=1}^J\prod_{k=1}^K\left( \omega_k\frac{\Gamma(N_{ij}+r_k)}{\Gamma(r_k)\Gamma(N_{ij}+1)}(1-\theta_{ijk})^{N_{ij}}\theta_{ijk}^{r_k}\right)^{S_{ijk}}.
\]
The E-step of the algorithm entails obtaining the current iterate of the objective expected log complete-data likelihood function $Q$, which takes the following shape in $u$-th iteration:
\begin{align} \label{eqn:Q1}
    Q_1(\phi\mid \phi^{(u)}) &= \sum_{i=1}^I \sum_{j=1}^J\sum_{k=1}^K \tau_{ijk}^{(u)} \left[\log\Gamma(N_{ij}+r_k)-\log\Gamma(r_k)+N_{ij}\log(1-\theta_{ijk})+r_k\log(\theta_{ijk})\right] \nonumber \\
    &+ \log\Gamma(K\alpha) - K\log\Gamma(\alpha)+ (\alpha-1)\sum_{k=1}^K\log(\omega_k)+\sum_{i=1}^I\sum_{j=1}^J\sum_{k=1}^K (\alpha-1)\tau_{ijk}^{(u)}\log\omega_k,    
\end{align}
where $\phi^{(u)} = \{ \Omega^{(u)}, R^{(u)}, H^{(u)}\}$ and $\tau_{ijk}^{(u)} = p(S_{ij}=k\mid N,\phi^{(u)}) \propto \omega_k^{(u)} f_{NB}(N_{ij}\mid r_k^{(u)},\theta_{ijk}^{(u)})$. The associated M step involves maximizing $Q_1$ with respect to $\{ \omega_k, h_k \}$ to obtain their next iterates. This leads to the following updating rule for $\{\omega_k\}$:
\[
\omega_k^{(u+1)} = \max\left\{0,\frac{\alpha-1 + \sum_{i=1}^I\sum_{j=1}^J\tau_{ijk}^{(u)}}{I*J+K(\alpha-1)}\right\}.
\]
The update for $\{h_k\}$ is obtained by solving the following equation:
\[
\frac{\partial Q_1(\phi_1 \mid \phi_1^{(u)})}{\partial h_k} = \sum_{i=1}^I\sum_{j=1}^J \tau_{ijk}^{(u)}\left[ \frac{N_{ij}}{h_k} - \frac{E_{ij}(N_{ij}+r_k)}{1+E_{ij} h_k}\right]=0.
\]
While a closed-form solution for $h_k$ is not available from the above equation, a simple iterative updating scheme may be employed (see Appendix S3.1). 

For estimation of $\{r_k\}$, we consider a separate data augmentation leveraging the Poisson sum of logarithmic series representation of negative binomial random variables (Quenouille (1946)\cite{quenouille1949relation} ).  Following  Adamidis (1999)\cite{adamidis1999theory} we consider this data augmentation and use the following complete data triplet $X_2 = \{S, M, Y\}$ where
\[
M_{ij} \mid S_{ij}=k \sim \text{Poisson}(\lambda_{ijk}=-r_k\log\theta_{ijk}), \text{and} 
\]
\[
Y_{ijl} \mid S_{ij}=k \sim \text{Logarithmic}(\theta_{ijk}), \quad N_{ij} = \sum_{l=1}^{m_{ij}}Y_{ijl}.
\]
The complete likelihood given $X_2$ is:
\[
p(X_2) = \frac{\Gamma(K\alpha)}{[\Gamma(\alpha)]^K}\prod_{k=1}^K \omega_k^{\alpha-1}\prod_{i=1}^I\prod_{j=1}^J\prod_{k=1}^K \left[ \omega_k^{\alpha-1} \frac{\lambda_{ijk}^{m_{ij}}\exp(-\lambda_{ijk})}{\Gamma(m_{ij}+1)}\prod_{l=1}^{m_{ij}}\frac{1}{-\log(1-\theta_{ijk})}\frac{\theta_{ijk}^{y_{ijl}}}{y_{ijl}}\right]^{S_{ij}}.
\]
The $u$-th iteration of the corresponding expected log complete-data likelihood function $Q_2$ is then:
\begin{align}
    Q_2(\phi \mid \phi^{(u)}) &= \sum_{i=1}^I\sum_{j=1}^J\sum_{k = 1}^K \tau_{ijk}^{(u)}\left[ \delta_{ijk}^{(u)} \log\lambda_{ijk} - \lambda_{ijk} + \delta_{ijk}^{(u)} \log \left( \frac{1}{-\log(1-\theta_{ijk})}\right) +N_{ij}\log\theta_{ijk} \right] \nonumber \\
    &+ \log\Gamma(K\alpha) - K\log\Gamma(\alpha)+ (\alpha-1)\sum_{k=1}^K\log(\omega_k)+\sum_{i=1}^I\sum_{j=1}^J\sum_{k=1}^K (\alpha-1)\tau_{ijk}^{(u)}\log\omega_k,
\end{align}
where $\phi^{(u)} = \{ \Omega^{(u)}, R^{(u)}, H^{(u)}\}$, $\tau_{ijk}^{(u)} = p(S_{ij}=k\mid N,\phi^{(u)}) \propto \omega_k^{(u)} f_{\text{NB}} (N_{ij}\mid r_k^{(u)},\theta_{ijk}^{(u)})$, $\delta_{ijk}^{(u)} = E(M_{ij}\mid S_{ij}=k,\phi^{(u)}) = r_k^{(u)}[\Psi(N_{ij}+r_k^{(u)})-\Psi(r_k^{(u)})].$  The corresponding CM-step of the algorithm produces the following closed-form updating rule for $\{r_k\}$, given $\{h_k\}$:
\begin{equation*}
 r_k^{(u+1)} = \frac{\sum_{i=1}^I\sum_{j=1}^J\tau_{ijk}^{(u)}\delta_{ijk}^{(u)}}{\sum_{i=1}^I\sum_{j=1}^J\tau_{ijk}^{(u)}\log\theta_{ijk}}.
\end{equation*}

\subsection{A iterative way to update $h_k$}

In CM2-step, we propose a inner iterative algorithm to update $\{h_k\}$ given $r_k$ and $\{ \tau_{ijk} \}$.

\begin{algorithm}
\caption{The inner loop algorithm to update $\{h_k\}$ given $\{r_k\}$} \label{alg:inner-loop} 
\begin{algorithmic}
\Require Current iteration of $\{ r_k^{(u+1)}\}$, $\{\tau_{ijk}^{(u+1)}\}$ and a tolerance parameter $h_{\text{tol}}$.

$h_k^{\text{old}} = h_k^{u}$ for $k=1,\dots,K$;

\textbf{Loop forever} 

$\quad$ $h_k^{\text{new}} = \frac{\sum_i^I \sum_j^J \tau_{ijk}^{(u+1)} N_{ij}}{\sum_i^I \sum_j^J\frac{\tau_{ijk}^{(u+1)} E_{ij}(N_{ij}+r_k^{(u+1)})}{1+E_{ij}h_k^{\text{old}}}} $ 
 for $k=1,\dots,K$;
 
$\quad$\textbf{if} $\max\{|h_k^{\text{new}} - h_k^{\text{old}}| : k=1,\dots,K\}<h_{\text{tol}}$

       $\quad$$\quad$ \textbf{break};

    $\quad$\textbf{end if}
    
    $\quad$$h_k^{\text{old}} = h_k^{\text{new}}$
    
\textbf{end loop}

\textbf{return} $h_k^{\text{new}}$

\end{algorithmic}
\end{algorithm}

\subsection{Proof of Theorem 1}

Define the following notations (see also Dempster et al. (1977)\cite{dempster1977maximum}).

Denote the conditional density of complete data ($x$) given observed data ($y$) and parameter $\phi$ be $f(x|y,\phi)$ which is constructed by the ratio between the complete data density and observed data density given the parameter $\phi$.
\begin{equation*}
    f(x|y,\phi) = f(x|\phi)/f(y|\phi).
\end{equation*}
Denote $L(\phi)$ as the log-likelihood of observed data given parameter $\phi$:
\begin{equation*}
    L(\phi) = \log f(y|\phi)
\end{equation*}
In our Algorithm 1, $x_1 = \{ S,N\}$,   $x_2 = \{ S,M,Y,Z\}$, $y = N$ and $\phi = \{ \Omega, R,H\}$ with two different $Q$ functions:
\begin{equation*}
    Q_1(\phi'|\phi) = E[\log f(x_1|\phi') | y, \phi]
\end{equation*}
\begin{equation*}
    Q_2(\phi'|\phi) = E[\log f(x_2|\phi') | y, \phi]
\end{equation*}
Define $H$ functions by the conditional expectation of log conditional likelihood $(f(x|y,\phi'))$ given the observed data ($y$) and parameter ($\phi$):
\begin{equation*}
    H(\phi'|\phi) = E[\log f(x|y,\phi') | y, \phi]
\end{equation*}

Lemma 1 from Dempster et al. (1977), given below, ensures that $H$ function decreases as we update parameters in M-step of an EM algorithm and is used to estiabilish the convergence of our Algorithm 1.
\begin{theorem}
For any pair $(\phi',\phi)$ in $\Omega \times \Omega$,
\[ H(\phi'|\phi)\leq H(\phi|\phi), \]
with equality if and only if $f(x|y,\phi') = f(x|y,\phi)$ almost everywhere.  
\end{theorem}

The connection between $H(\phi'|\phi)$,  log likelihood of observed data $L(\phi)$ and the $Q(\phi'|\phi)$ function is described by the following Lemma.
\begin{theorem}
For any pair $(\phi',\phi)$ in $\Omega \times \Omega$,
\[ Q(\phi'|\phi) = L(\phi') + H(\phi'|\phi). \] 
\end{theorem}
The complete proof of Theorem 1 is provided below.

\begin{proof}

\end{proof}

In the $i$-th iteration of proposed algorithm,
\begin{equation}
    \begin{split}
        L(\phi^{(u)} = \{\omega^{(u)}, r^{(u)},h^{(u)}\}) &= Q_2(\phi^{(u)}|\phi^{(u)}) - H_2(\phi^{(u)}|\phi^{(u)})\\
        &\myeqa Q_2(\omega^{(u+1)}, r^{(u+1)},h^{(u)}|\phi^{(u)}) -  H_2(\phi^{(u)}|\phi^{(u)})\\
        &\myeqb Q_2(\omega^{(u+1)}, r^{(u+1)},h^{(u)}|\phi^{(u)}) -  H_2(\omega^{(u+1)}, r^{(u+1)},h^{(u)}|\phi^{(u)})\\
        &= L(\omega^{(u+1)}, r^{(u+1)},h^{(u)})\\
        &= Q_1(\omega^{(u+1)}, r^{(u+1)},h^{(u)}|\phi^{(u)}) - H_1(\omega^{(u+1)}, r^{(u+1)},h^{(u)}|\phi^{(u)})\\
        &\myeqaa Q_1(\omega^{(u+1)}, r^{(u+1)},h^{(u+1)}|\phi^{(u)}) -  H_1(\omega^{(u+1)}, r^{(u+1)},h^{(u)}|\phi^{(u)})\\
        &\myeqb Q_1(\omega^{(u+1)}, r^{(u+1)},h^{(u+1)}|\phi^{(u)}) - H_1(\omega^{(u+1)}, r^{(u+1)},h^{(u+1)}|\phi^{(u)})\\
        & = L(\phi^{(u+1)}),
    \end{split}
\end{equation}

which shows that for each iteration of our algorithm, the log-likelihood of observed data is guaranteed to increase.

\subsection{Selection of the Dirchlet hyperparameter $\alpha$}

To select the Dirichlet hyperparameter $\alpha$, we begin by specifying a set of candidate values for $\alpha$, i.e. $\{0.01, 0.25, 0.5, 0.75, 0.99\}$. Given an SRS data table, we fit the general gamma model with $\alpha$ values respectively. For each fitted model, posterior draws are generated, and their corresponding log model likelihoods are calculated. To assess model fit, we compute the approximate leave-one-out cross-validation information criterion (LOOIC) for each model using R package "LOO"\cite{gabry2024package}. The $\alpha$ value that minimizes the LOOIC is then selected as the optimal hyperparameter.

\section{Connection between Wasserstein distance and posterior RMSE}

Let $f_{ij}$ and $F_{ij}$ be the posterior and cumulative posterior density of $\lambda_{ij}$. Let $\lambda_{ij}^{\text{true}}$ be a degenerate random variable at $a$. Then, the Wasserstein-$p$ distance between $f_{ij}$ and $f_{\lambda_{ij}^{\text{true}}}$ is:

\begin{equation}
    \begin{split}
        \text{Wasserstein-}p(f_{ij}, f_{\lambda_{ij}^{\text{true}}})& = \left[\int_0^1|F^{-1}_{ij}(q) - F^{-1}_{\lambda_{ij}^{\text{true}}}(q)|^pdq\right]^{1/p}\\
        & = \left[\int_0^1|F^{-1}_{ij}(q) - a|^pdq\right]^{1/p}\\
        &= \left[\int_0^{\infty}|\lambda - a|^pf_{ij}d\lambda\right]^{1/p}
    \end{split}
\end{equation}

Therefore, when $p=1,2$, the corresponding Wasserstein distances are equivalent to the posterior absolute error and the root posterior squared error. 

\subsection{Exact form Wasserstein-1 distance Computation}

It's not hard to show that root posterior squared error or the Wasserstein-2 distance is constructed by the first and the second non-central moment of the posterior distribution which is easy to obtain. For discrete non-parametric empirical Bayes models, the analytical computation of Wasserstein distances can be obtained directly. 

The computation of Wasserstein-1 distance or the posterior absolute error of the general gamma-model is provided as follows.

Suppose our posterior distribution is a mixture of gamma distributions with $K$ components,
$$f(x) = \sum_{j=1}^K p_j\frac{1}{\Gamma(\alpha_j)\beta_j^{\alpha_j}}x^{\alpha_j-1}e^{-x/\beta_j}.$$

Denote the cumulative pdf of gamma as
$$\Phi(x|\alpha_j,\beta_j) = \int_0^{x} \frac{1}{\Gamma(\alpha_j)\beta_j^{\alpha_j}}x^{\alpha_j-1}e^{-x/\beta_j}dx.$$

Then, we have
$$ \int^{\infty}_{\lambda^{true}} xf(x)dx = \sum_{j=1}^K p_j\alpha_j\beta_j[1-\Phi(\lambda^{true}|\alpha_j+1,\beta_j)],$$

$$ \int_{0}^{\lambda^{true}} xf(x)dx = \sum_{j=1}^K p_j\alpha_j\beta_j\Phi(\lambda^{true}|\alpha_j+1,\beta_j).$$

Therefore,

$$ Wasserstein-1(f_{ij}, f_{\lambda^{true}}) = \sum_{j=1}^K p_j\alpha_j\beta_j[1-2\Phi(\lambda^{true}|\alpha_j+1,\beta_j)]+2\lambda^{true}\sum_{j=1}^K p_j\Phi(\lambda^{true}|\alpha_j,\beta_j)-\lambda^{true}.$$

\section{Simulation results}

\subsection{RMSE results}

\begin{table}[]
\centering
\begin{tabular}{llll}
\hline
Simulation setting &Case      &number of signals & Positions of signal\\ \hline
I &1&1 &(1,1)            \\
&2&3& (1,1); (7,1); (9,1) \\
&3&6&  (1,1); (7,1); (9,1) \\
 && &(29,5); (38,5); (39,5)  \\\hline
 II &4 &6 &(1,1); (7,1); (9,1)\\
 && &(29,5); (38,5); (39,5) \\\hline
 III &5 &12 &(1,1); (9,5); (9,1);\\
 &  & &(1,4); (25,1);  (29,1); \\ 
 && &(29,5); (38,5); (39,5); \\ 
 &  & &(39,1); (41,1);(29,6) \\\hline
\end{tabular}
\caption{Positions of signals of Simulation setting I, II and III}
\label{tab:my-table}
\end{table}

\begin{figure}[htp]
    \centering
    \includegraphics[width=16cm]{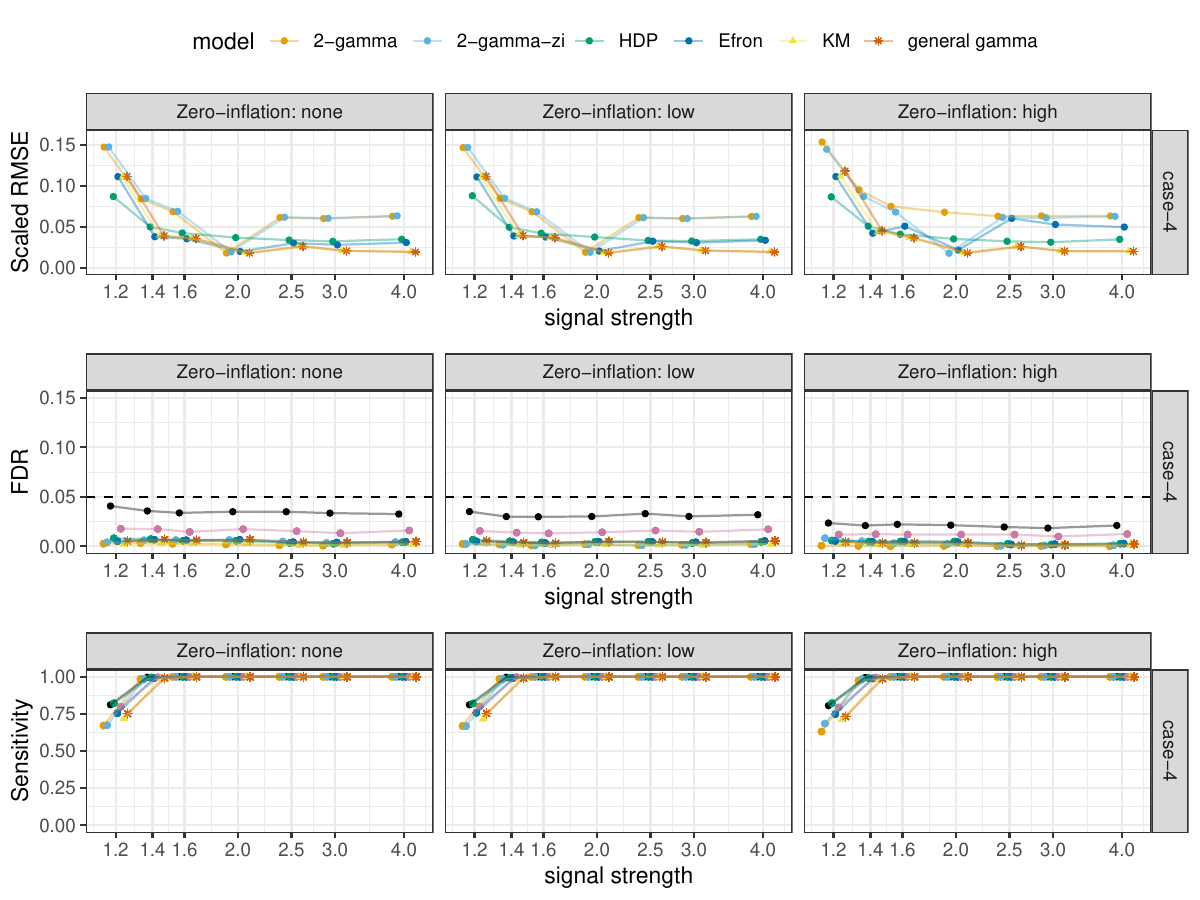}
    \caption{Simulation II (heterogeneous signal strengths): Max-scaled-RMSEs (top panel), FDRs (middle panel), and sensitivities (bottom panel) versus signal strength for Bayesian models (color-coded) and different zero-inflation level. Non-parametric Bayesian models (HDP, Efron, KM, and general-gamma) perform better than parametric Bayesian models (2-gamma and 2-gamma-zi) in signal estimation. All methods control FDR under 0.05. Methods also appear to have similar overall patterns for sensitivity, except for 2-gamma and 2-gamma-zi which show somewhat lower sensitivities than the other in low signal strength.}
\end{figure}

\begin{figure}[htp]
    \centering
    \includegraphics[width=16cm]{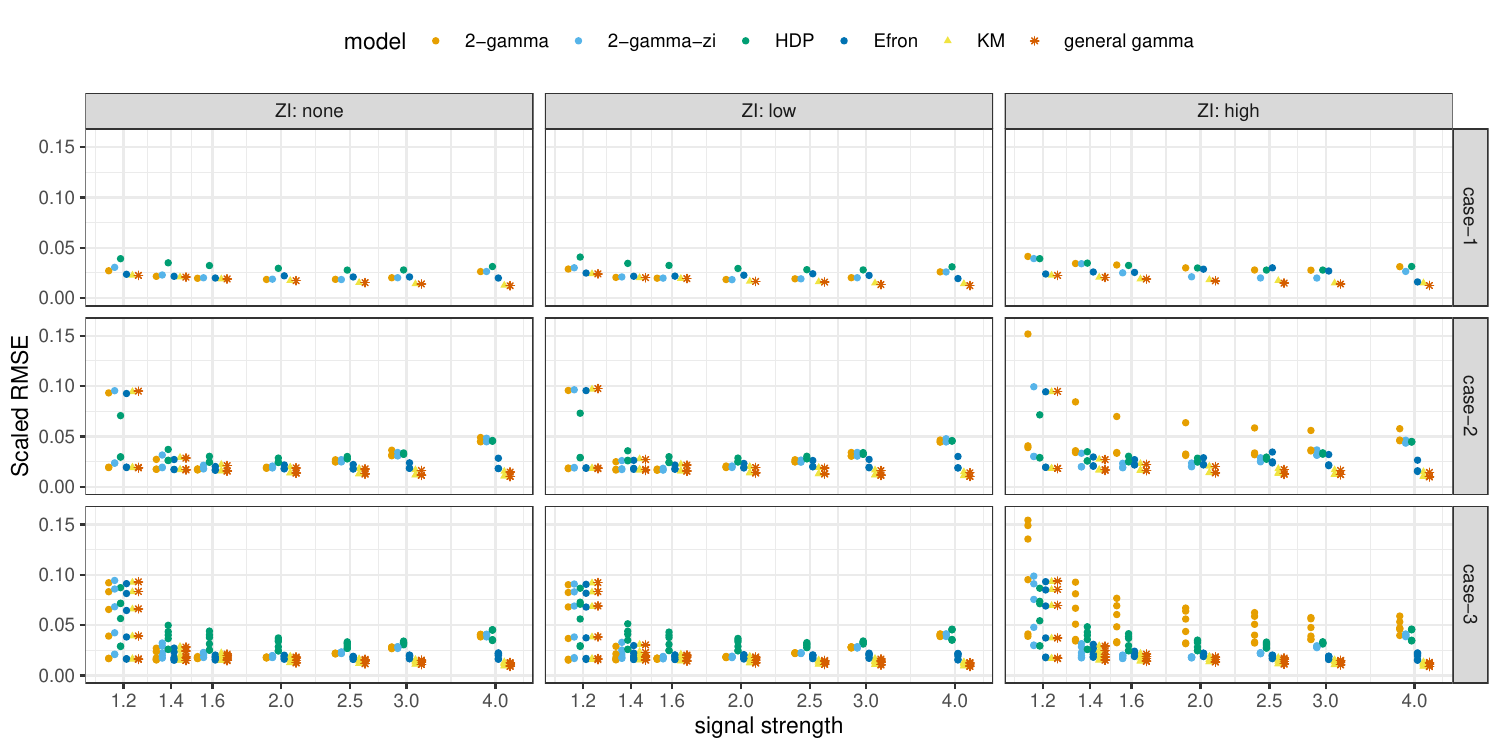}
    \caption{Simulation I (homogeneous signal strengths): Scaled RMSEs versus signal strength for Bayesian models (color-coded) with fixed truth in signal strength. The general-gamma and the KM perform better than other methods across zero-inflation levels and cases.}
\end{figure}

\begin{figure}[htp]
    \centering
    \includegraphics[width=16cm]{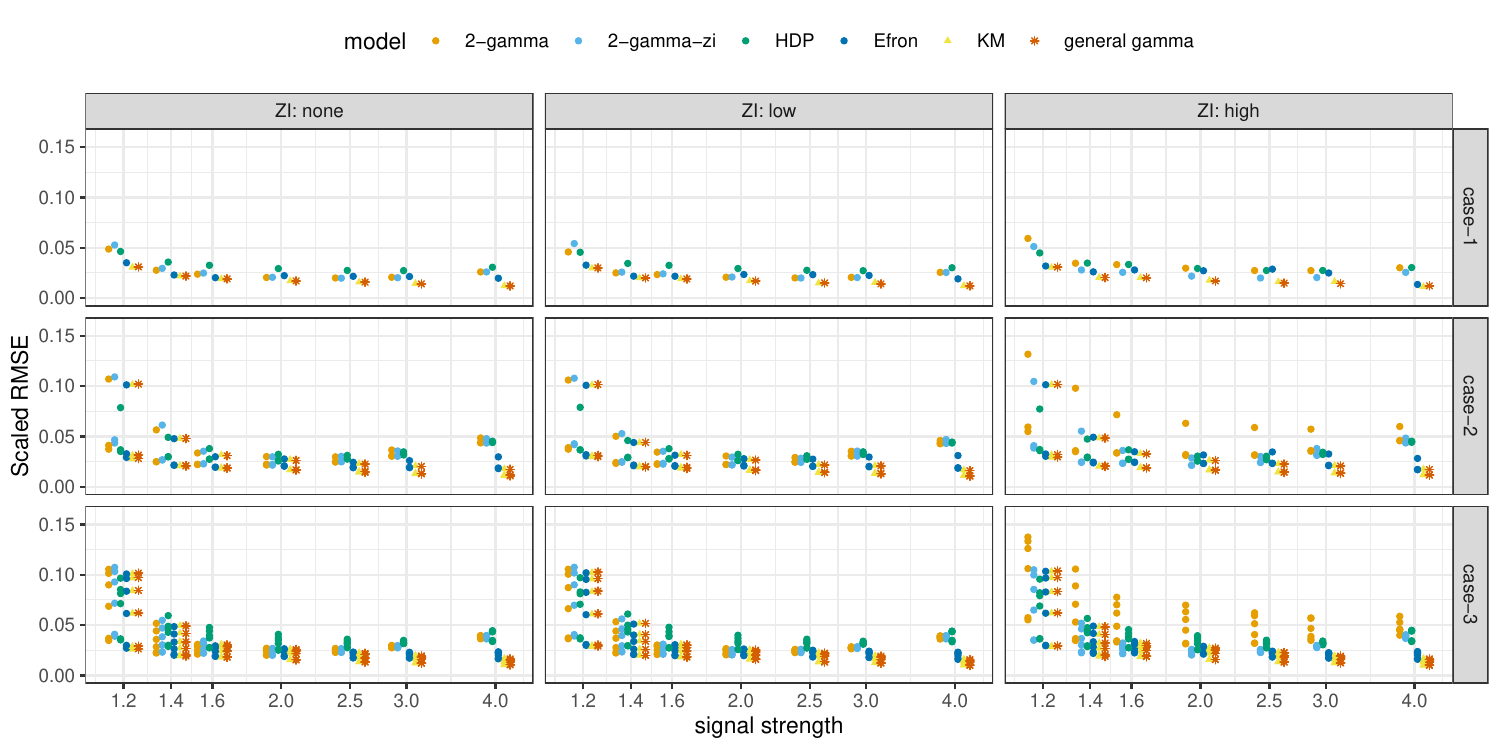}
    \caption{Simulation I (homogeneous signal strengths): Scaled RMSEs versus signal strength for Bayesian models (color-coded) with randomly perturbed truth in signal strength. The general-gamma and the KM perform better than other methods across zero-inflation levels and cases.}
\end{figure}

\begin{figure}[htp]
    \centering
    \includegraphics[width=16cm]{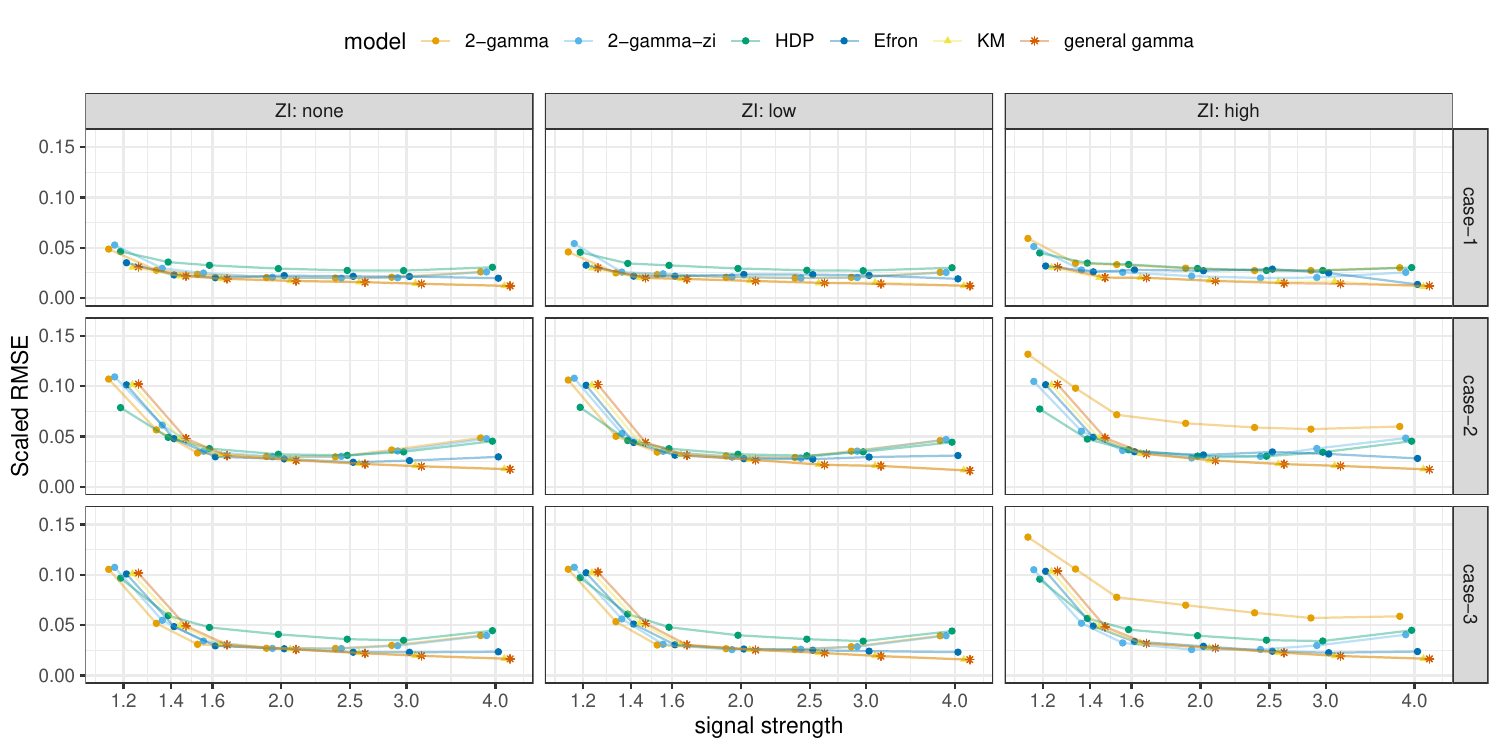}
    \caption{Simulation I (homogeneous signal strengths): Max-Scaled-RMSEs versus signal strength for Bayesian models (color-coded) with fixed truth in signal strength. The general-gamma and the KM perform better than other methods across zero-inflation levels and cases.}
\end{figure}

\begin{figure}[htp]
    \centering
    \includegraphics[width=16cm]{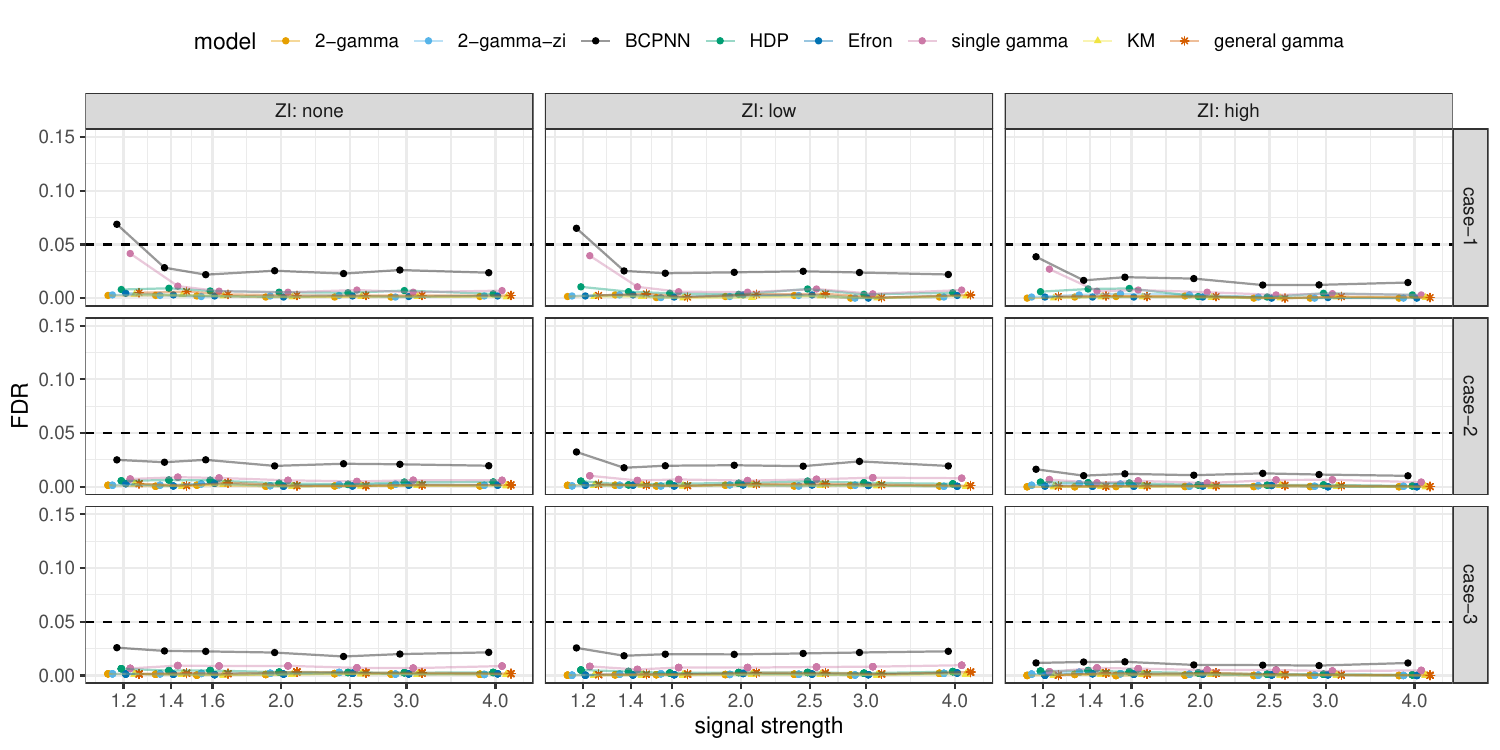}
    \caption{Simulation I (homogeneous signal strengths): FDRs versus signal strength for Bayesian models (color-coded) with randomly perturbed truth in signal strength. All methods control FDR under 0.05, except for the BCPNN, which has some situations and a small level of signal strength, and the computed FDR exceeds 0.05.}
\end{figure}

\begin{figure}[htp]
    \centering
    \includegraphics[width=16cm]{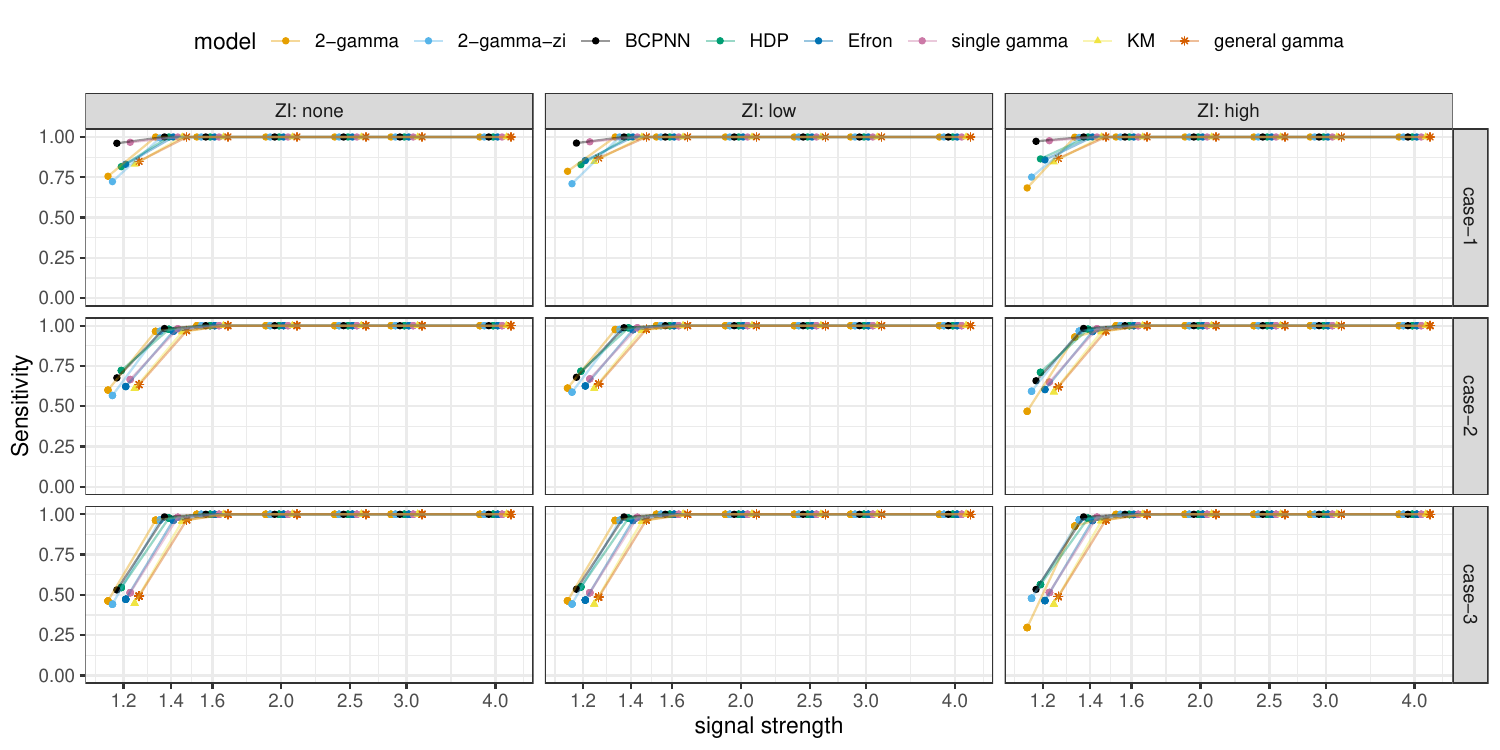}
    \caption{Simulation I (homogeneous signal strengths): sensitivities versus signal strength for Bayesian models (color-coded) with randomly perturbed truth in signal strength. Methods appear similar overall pattern for sensitivity, except for 2-gamma which shows somewhat lower sensitivities than the other in low signal strength.}
\end{figure}

\begin{figure}[htp]
    \centering
    \includegraphics[width=16cm]{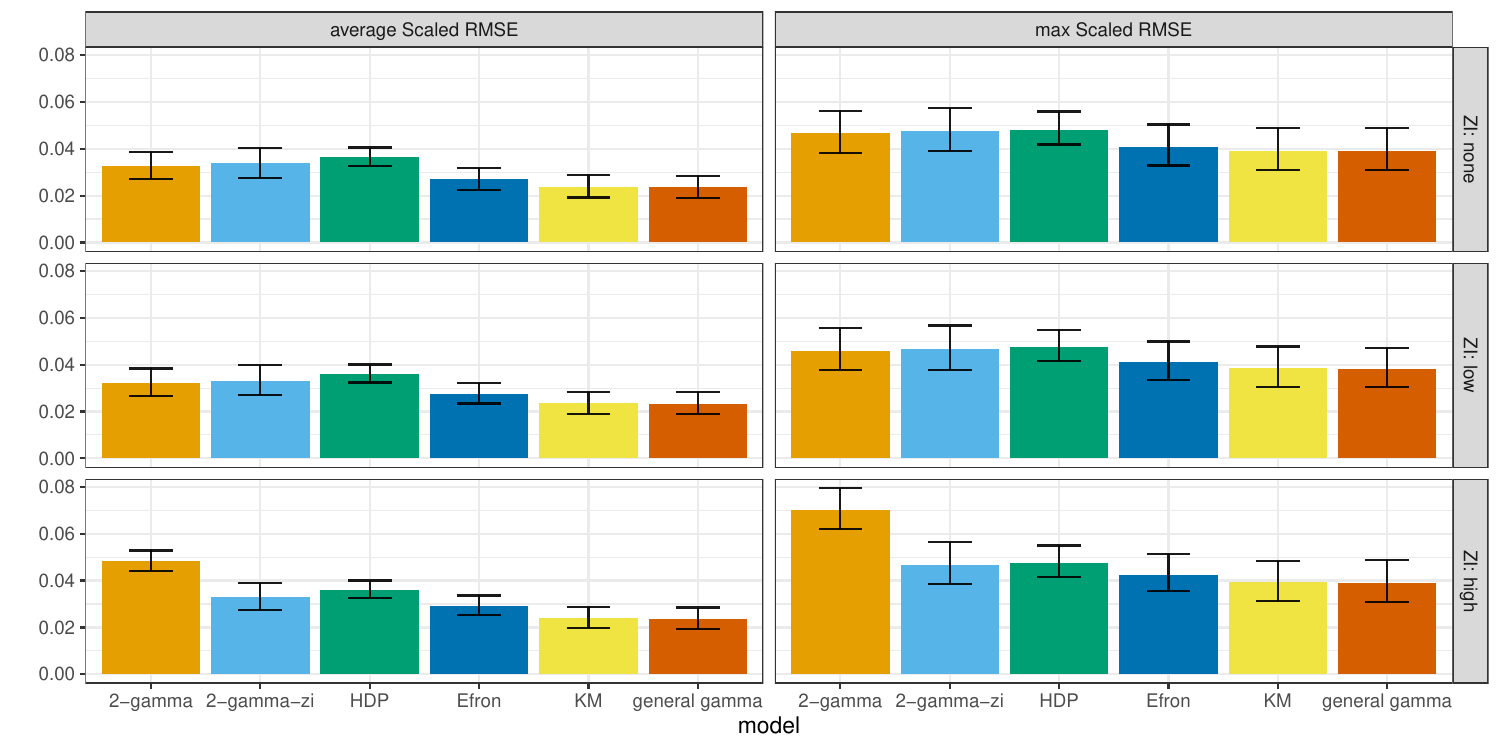}
    \caption{Simulation I (homogeneous signal strengths): The overall means of average (left column) and maximum (right column) scaled RMSEs obtained across all signal strength values and all choices of the number of signal cells are plotted as vertical bars for different levels of zero inflation (along rows) for each bayesian models. Each bar represents the replication-based mean of the average (left column) or the maximum (right column) scaled RMSE of a specific method computed over the signal cells of an entire table; the associated error whiskers represent the 5th and 95th percentile points across replicates. The 2-gamma model performs worse than other models in high levels of zero-inflation.}
\end{figure}

\begin{figure}[htp]
    \centering
    \includegraphics[width=16cm]{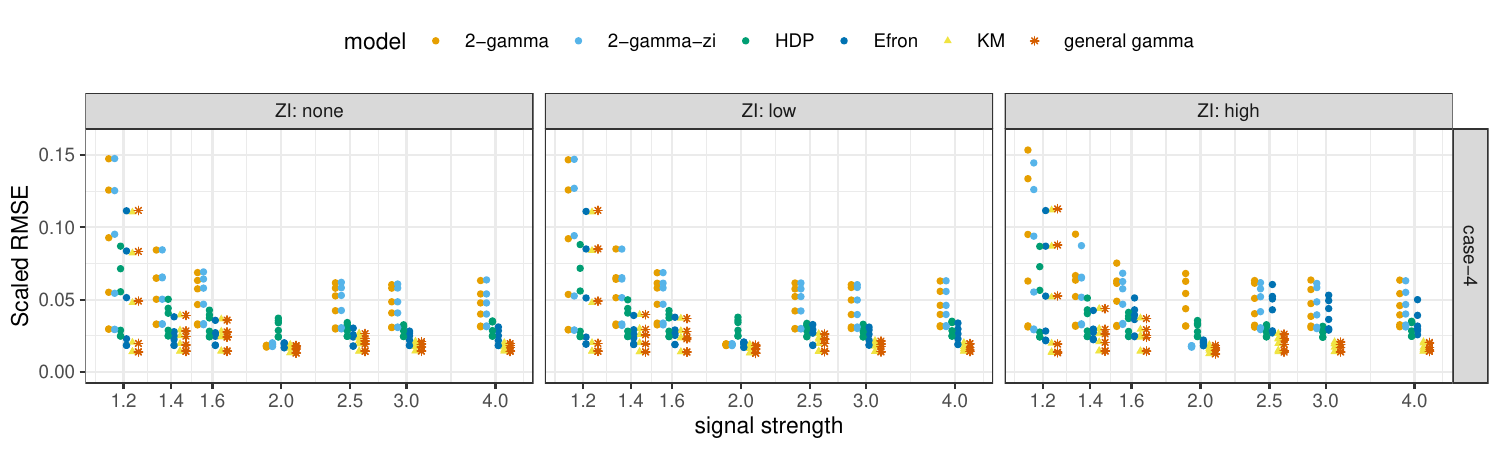}
    \caption{Simulation II (heterogeneous signal strengths): Scaled RMSEs versus signal strength for Bayesian models (color-coded) with fixed truth in signal strength. The general-gamma and the KM perform better than other methods across zero-inflation levels and cases.}
\end{figure}

\begin{figure}[htp]
    \centering
    \includegraphics[width=16cm]{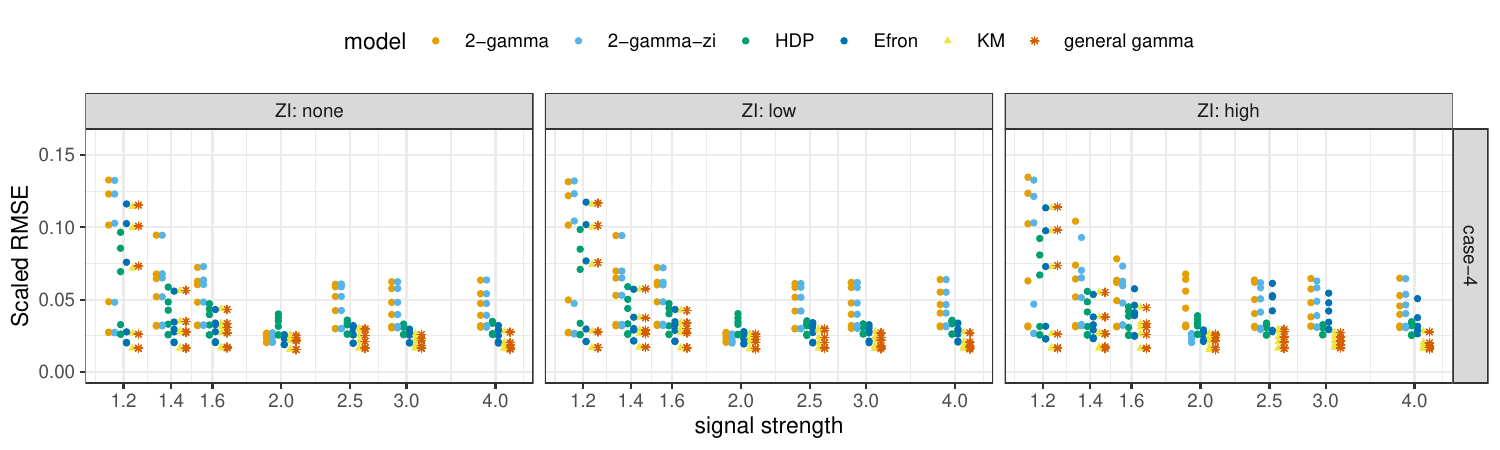}
    \caption{Simulation II (heterogeneous signal strengths): Scaled RMSEs versus signal strength for Bayesian models (color-coded) with randomly perturbed truth in signal strength. The general-gamma and the KM perform better than other methods across zero-inflation levels and cases.}
\end{figure}

\begin{figure}[htp]
    \centering
    \includegraphics[width=16cm]{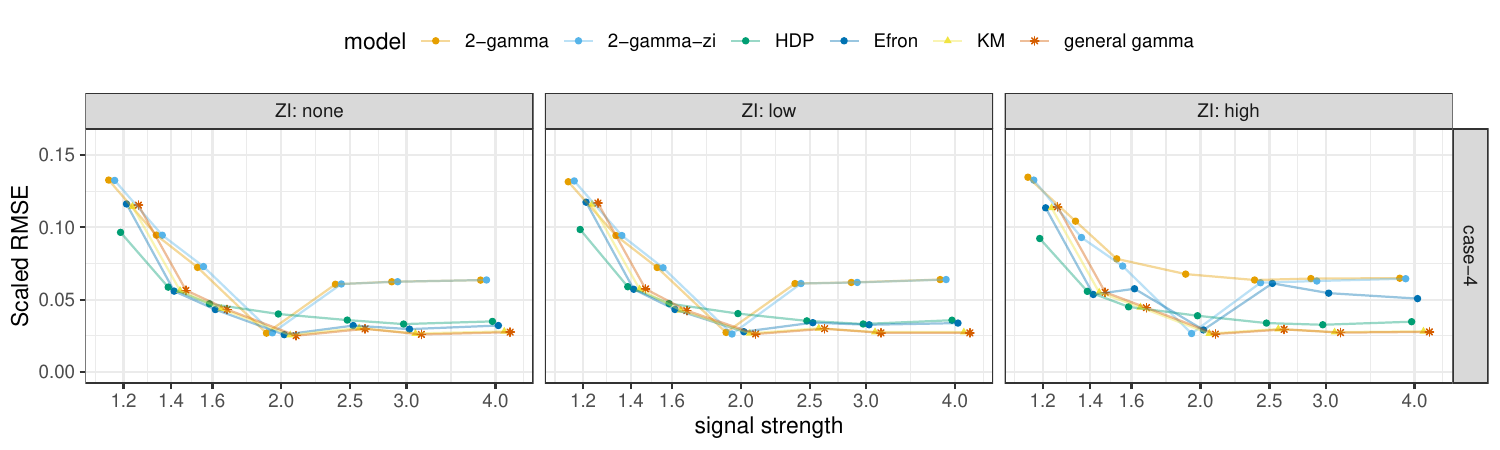}
    \caption{Simulation II (heterogeneous signal strengths): Max-Scaled-RMSEs versus signal strength for Bayesian models (color-coded) with fixed truth in signal strength. The general-gamma and the KM perform better than other methods across zero-inflation levels and cases.}
\end{figure}

\begin{figure}[htp]
    \centering
    \includegraphics[width=16cm]{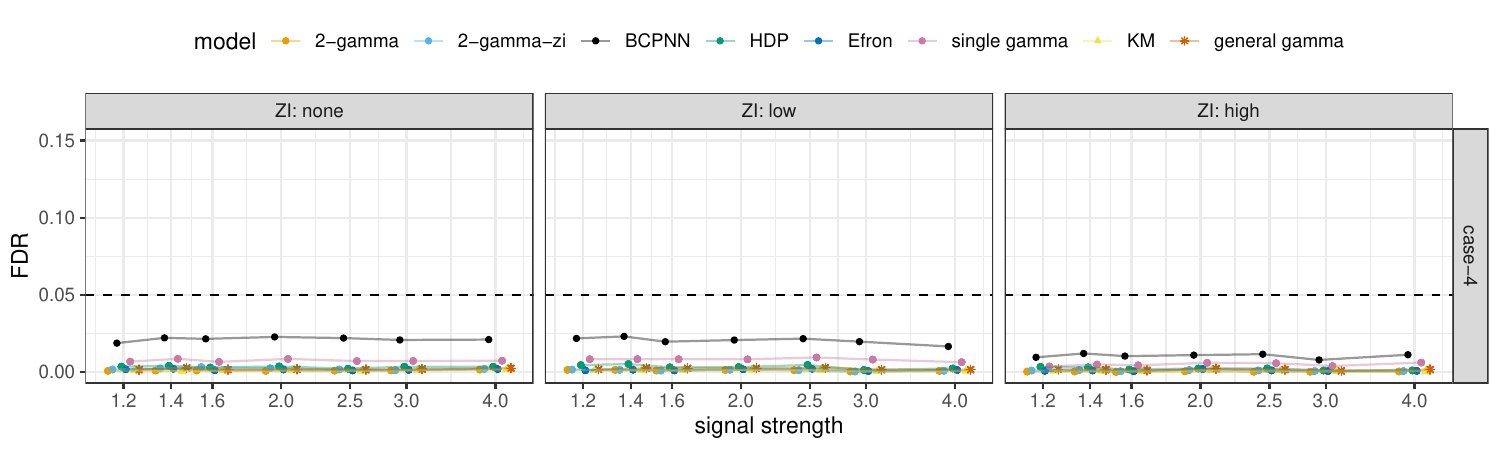}
    \caption{Simulation II (heterogeneous signal strengths): FDRs versus signal strength for Bayesian models (color-coded) with randomly perturbed truth in signal strength. All methods control FDR under 0.05.}
\end{figure}

\begin{figure}[htp]
    \centering
    \includegraphics[width=16cm]{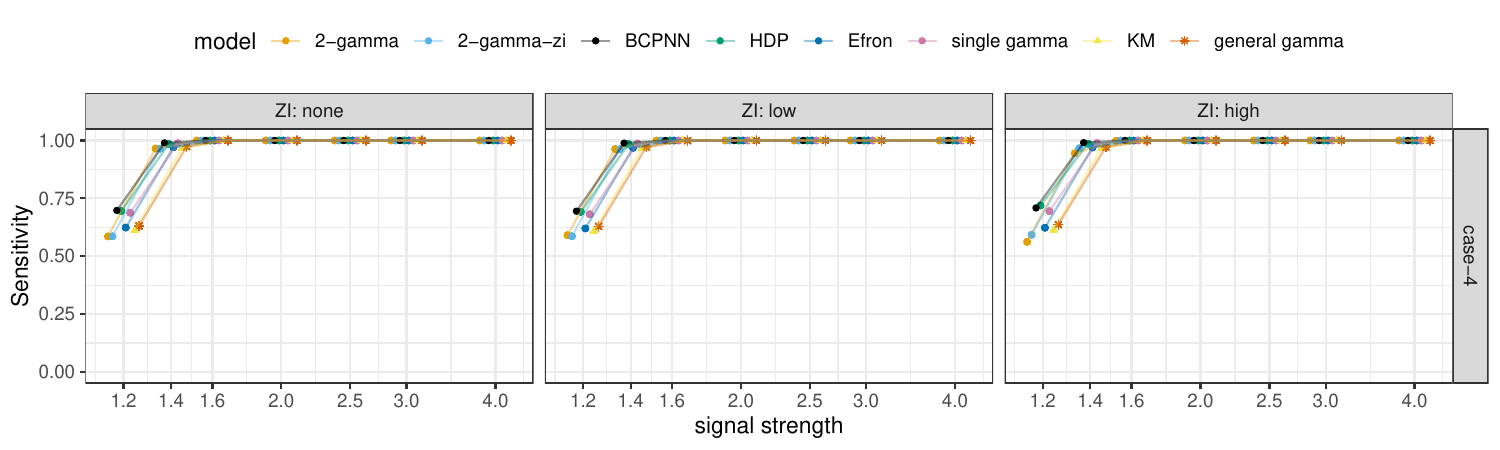}
    \caption{Simulation II (heterogeneous signal strengths): sensitivities versus signal strength for Bayesian models (color-coded) with randomly perturbed truth in signal strength. Methods appear to have similar overall patterns for sensitivity, except for 2-gamma and 2-gamma-zi which show somewhat lower sensitivities than the other in low signal strength.}
\end{figure}

\begin{figure}[htp]
    \centering
    \includegraphics[width=16cm]{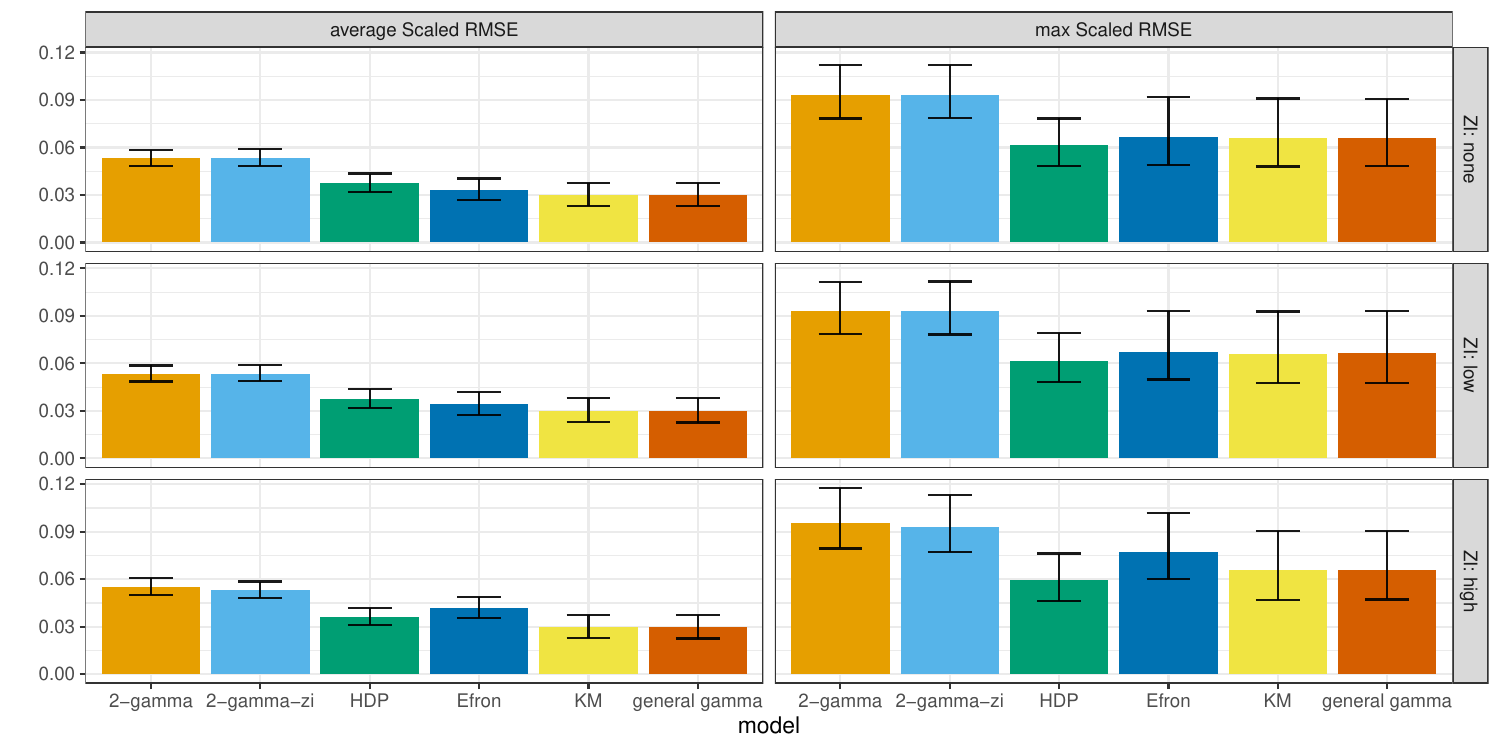}
    \caption{Simulation II (heterogeneous signal strengths): The overall means of average (left column) and maximum (right column) scaled RMSEs obtained across all signal strength values and all choices of the number of signal cells are plotted as vertical bars for different levels of zero inflation (along rows) for each bayesian models. Each bar represents the replication-based mean of the average (left column) or the maximum (right column) scaled RMSE of a specific method computed over the signal cells of an entire table; the associated error whiskers represent the 5th and 95th percentile points across replicates. Non-parametric Bayesian methods (KM, general-gamma, HDP, and Efron) perform better than parametric Bayesian methods (2-gamma and 2-gamma-zi).}
\end{figure}

\begin{figure}[htp]
    \centering
    \includegraphics[width=16cm]{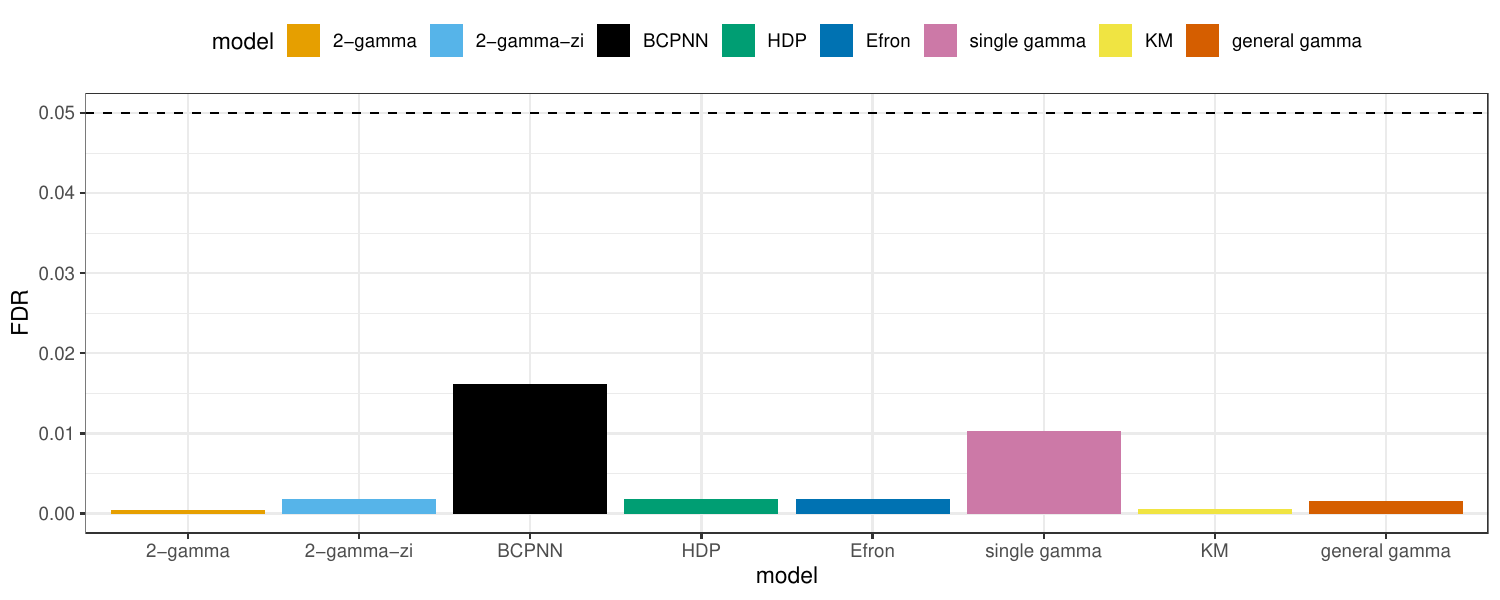}
    \caption{Simulation III (highly heterogeneous signal strengths): FDRs for Bayesian models with fixed truth in signal strength. All methods control FDR under 0.05.}
\end{figure}

\clearpage

\subsection{MAE results}

\begin{figure}[htp]
    \centering
    \includegraphics[width=16cm]{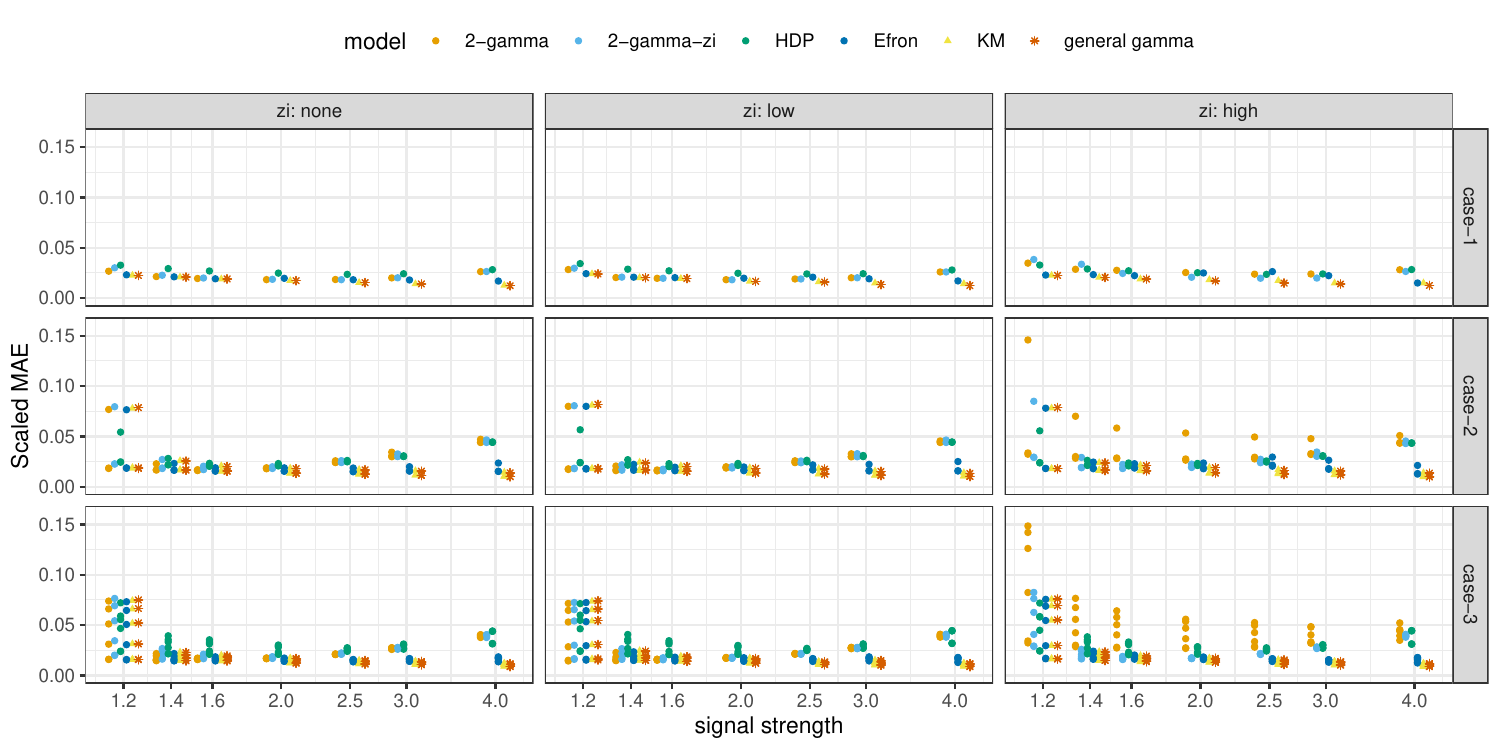}
    \caption{Simulation I (homogeneous signal strengths): Scaled MAEs versus signal strength for Bayesian models (color-coded) with fixed truth in signal strength. The general-gamma and the KM perform better than other methods across zero-inflation levels and cases.}
\end{figure}

\begin{figure}[htp]
    \centering
    \includegraphics[width=16cm]{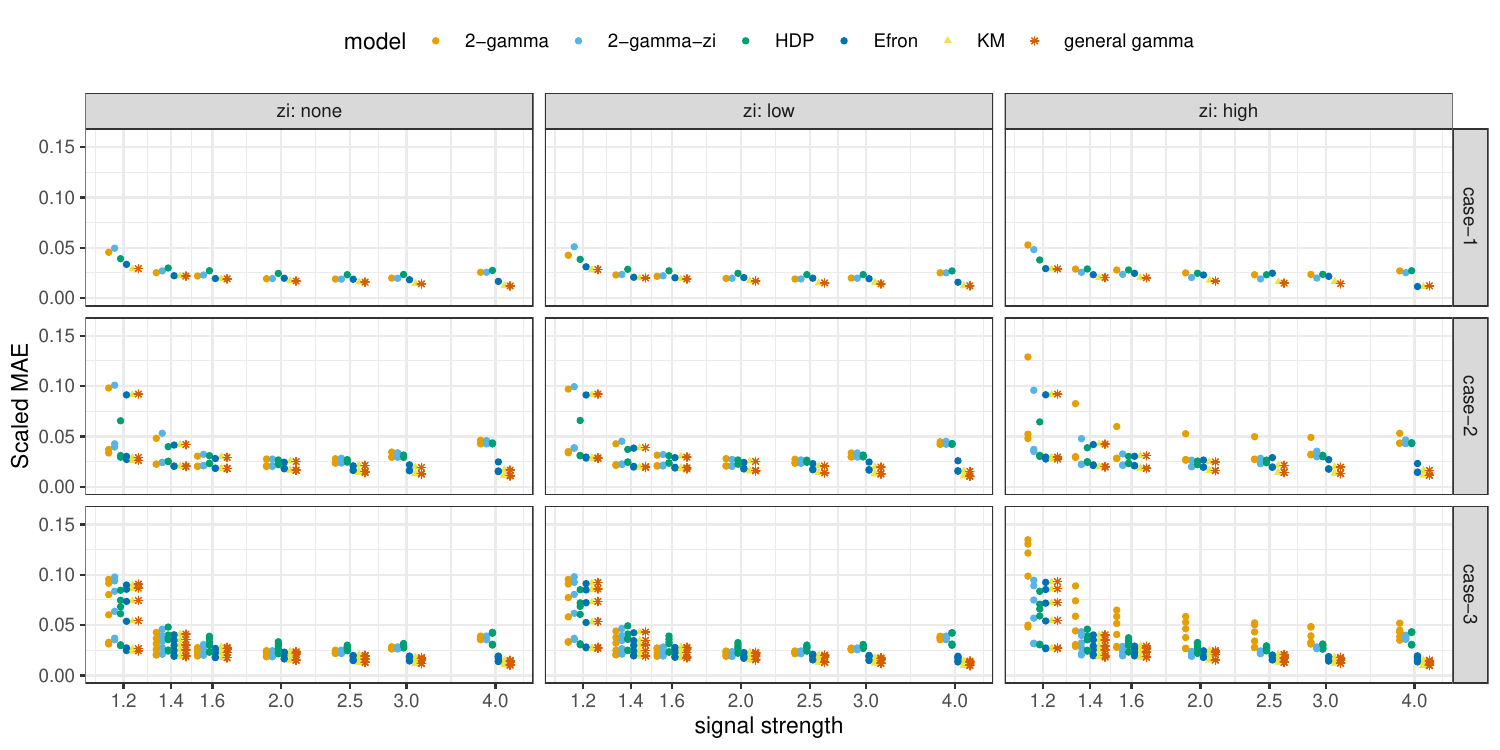}
    \caption{Simulation I (homogeneous signal strengths): Scaled MAEs versus signal strength for Bayesian models (color-coded) with randomly perturbed truth in signal strength. The general-gamma and the KM perform better than other methods across zero-inflation levels and cases.}
\end{figure}

\begin{figure}[htp]
    \centering
    \includegraphics[width=16cm]{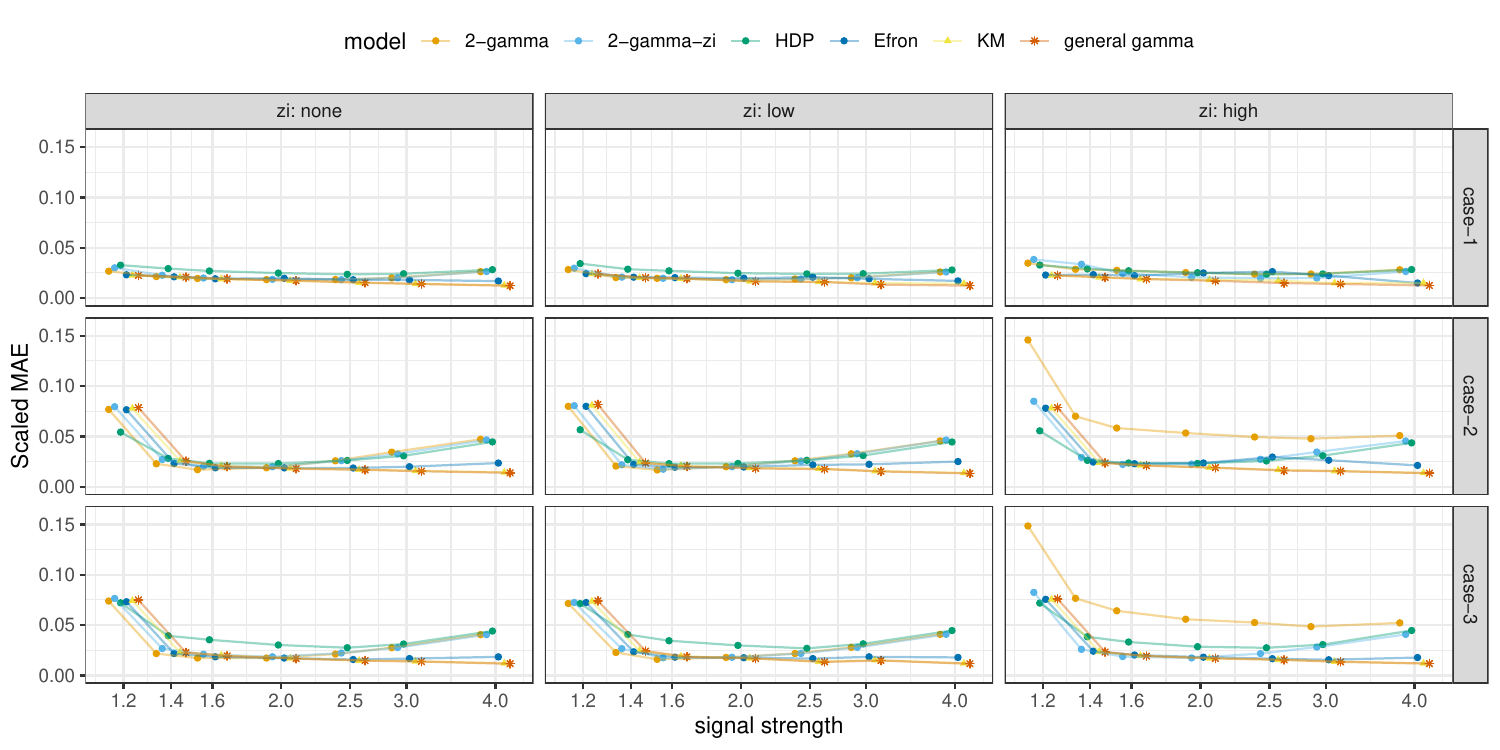}
    \caption{Simulation I (homogeneous signal strengths): Max-Scaled-MAEs versus signal strength for Bayesian models (color-coded) with fixed truth in signal strength. The general-gamma and the KM perform better than other methods across zero-inflation levels and cases.}
\end{figure}

\begin{figure}[htp]
    \centering
    \includegraphics[width=16cm]{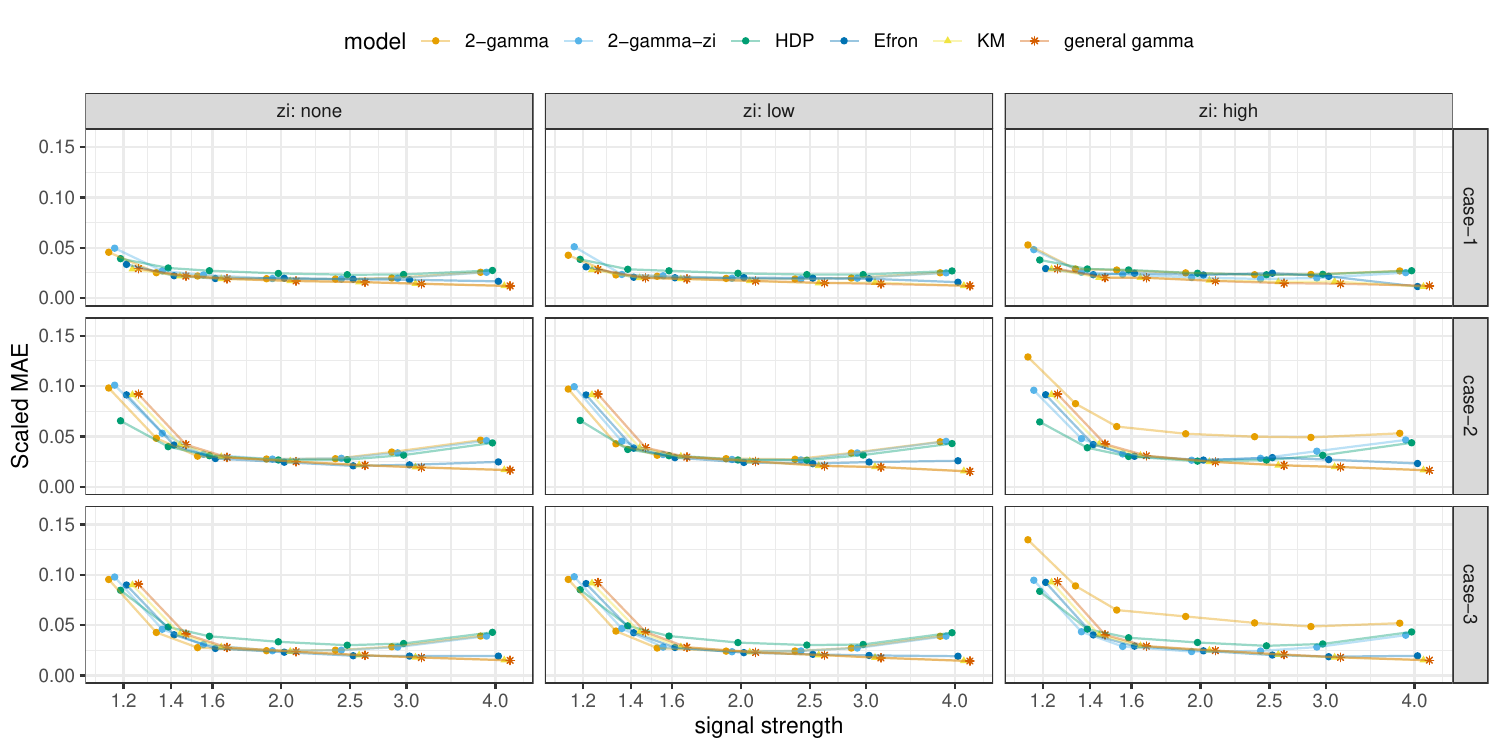}
    \caption{Simulation I (homogeneous signal strengths): Max-Scaled-MAEs versus signal strength for Bayesian models (color-coded) with randomly perturbed truth in signal strength. The general-gamma and the KM perform better than other methods across zero-inflation levels and cases.}
\end{figure}

\begin{figure}[htp]
    \centering
    \includegraphics[width=16cm]{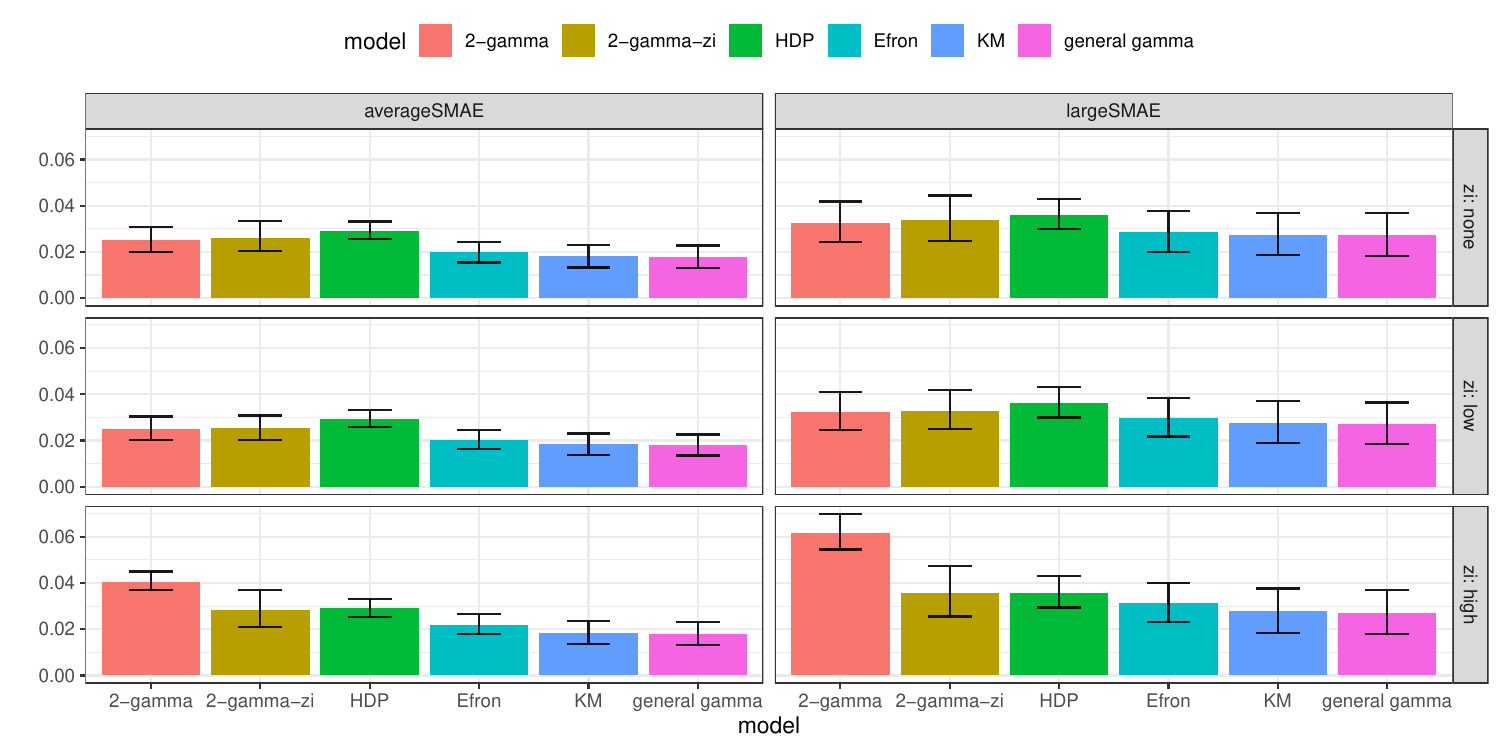}
    \caption{Simulation I (homogeneous signal strengths): The overall means of average (left column) and maximum (right column) scaled MAEs obtained across all signal strength values and all choices of the number of signal cells with fixed true in signal strength are plotted as vertical bars for different levels of zero inflation (along rows) for each bayesian models. Each bar represents the replication-based mean of the average (left column) or the maximum (right column) scaled MAE of a specific method computed over the signal cells of an entire table; the associated error whiskers represent the 5th and 95th percentile points across replicates. The 2-gamma model performs worse than other models in high levels of zero-inflation.}
\end{figure}

\begin{figure}[htp]
    \centering
    \includegraphics[width=16cm]{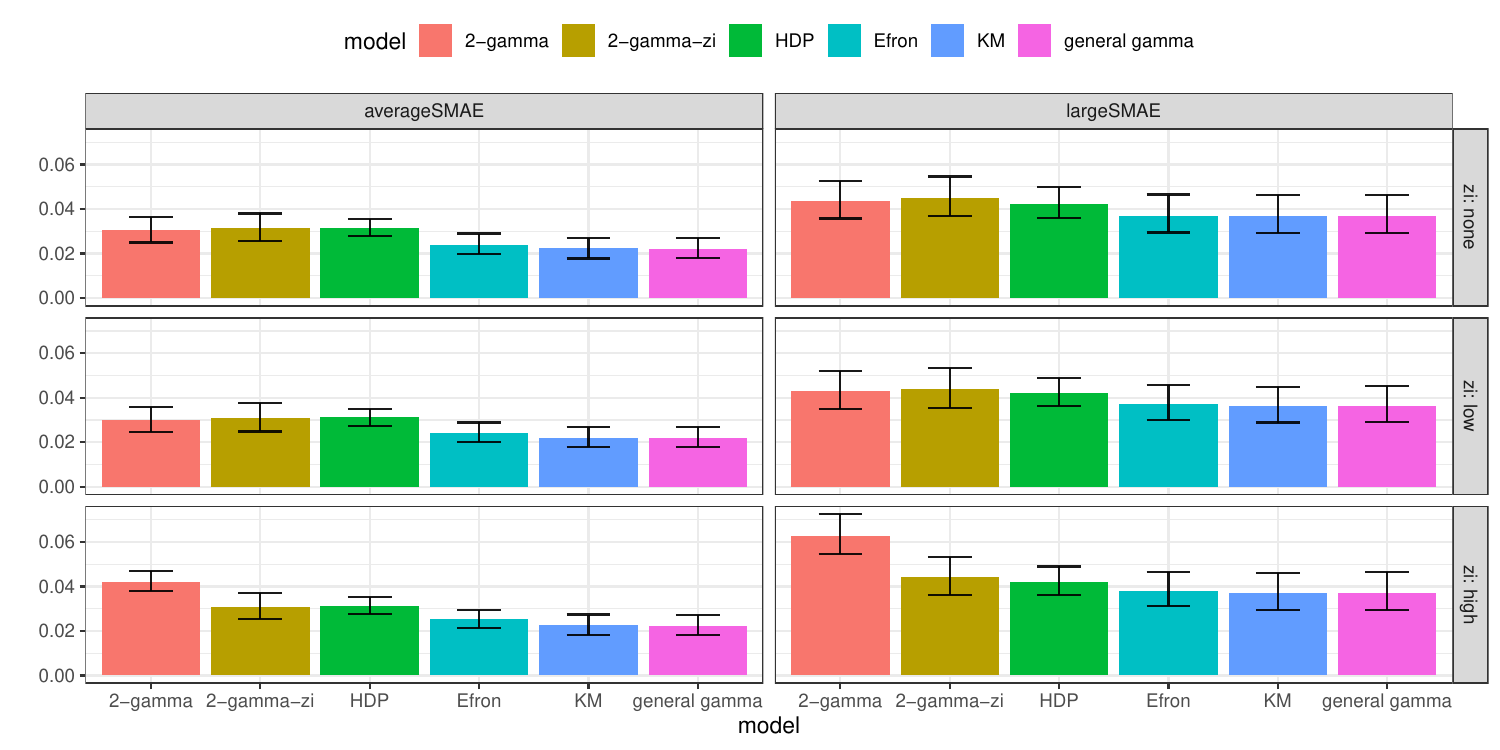}
    \caption{Simulation I (homogeneous signal strengths): The overall means of average (left column) and maximum (right column) scaled MAEs obtained across all signal strength values and all choices of the number of signal cells with randomly perturbed truth in signal strength are plotted as vertical bars for different levels of zero inflation (along rows) for each bayesian models. Each bar represents the replication-based mean of the average (left column) or the maximum (right column) scaled MAE of a specific method computed over the signal cells of an entire table; the associated error whiskers represent the 5th and 95th percentile points across replicates. The 2-gamma model performs worse than other models in high levels of zero-inflation.}
\end{figure}

\begin{figure}[htp]
    \centering
    \includegraphics[width=16cm]{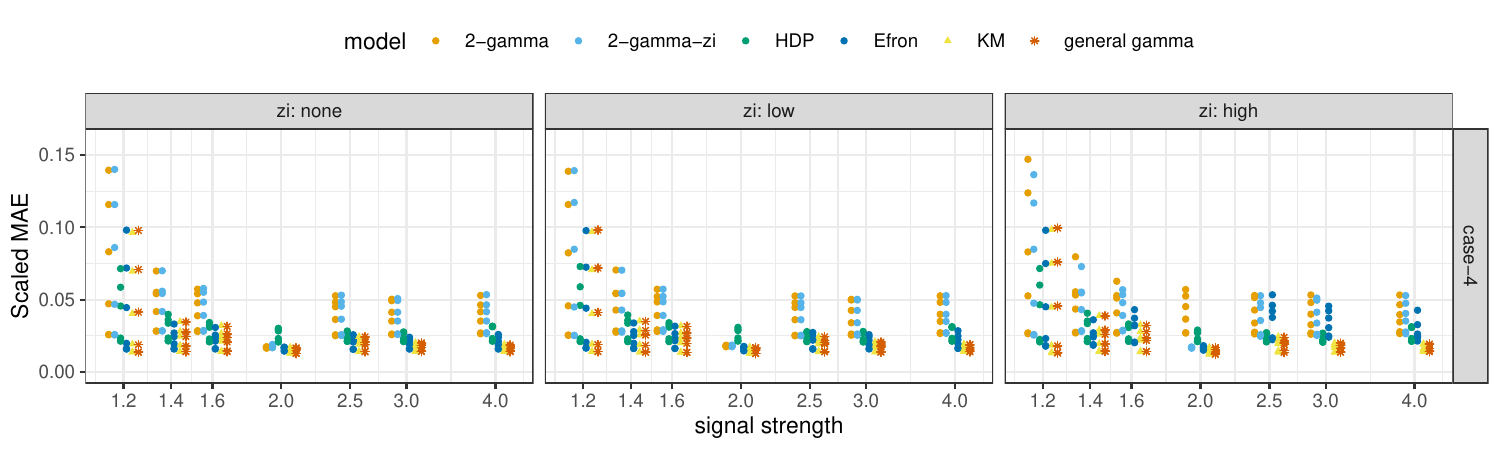}
    \caption{Simulation II (heterogeneous signal strengths): Scaled MAEs versus signal strength for Bayesian models (color-coded) with fixed truth in signal strength. The general-gamma and the KM perform better than other methods across zero-inflation levels and cases.}
\end{figure}

\begin{figure}[htp]
    \centering
    \includegraphics[width=16cm]{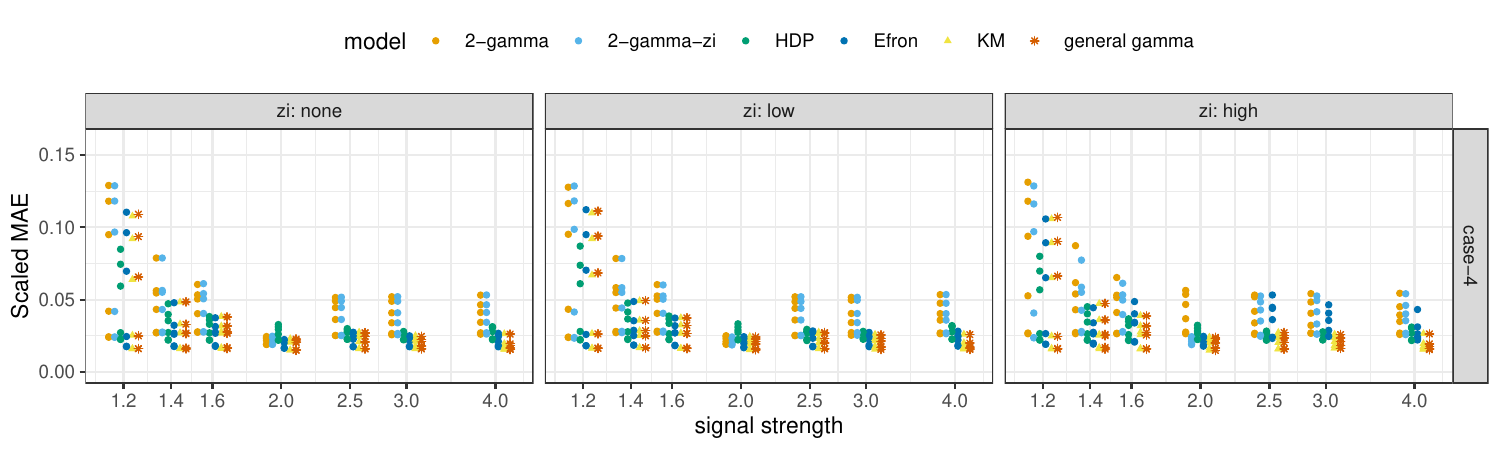}
    \caption{Simulation II (heterogeneous signal strengths): Scaled MAEs versus signal strength for Bayesian models (color-coded) with randomly perturbed truth in signal strength. The general-gamma and the KM perform better than other methods across zero-inflation levels and cases.}
\end{figure}

\begin{figure}[htp]
    \centering
    \includegraphics[width=16cm]{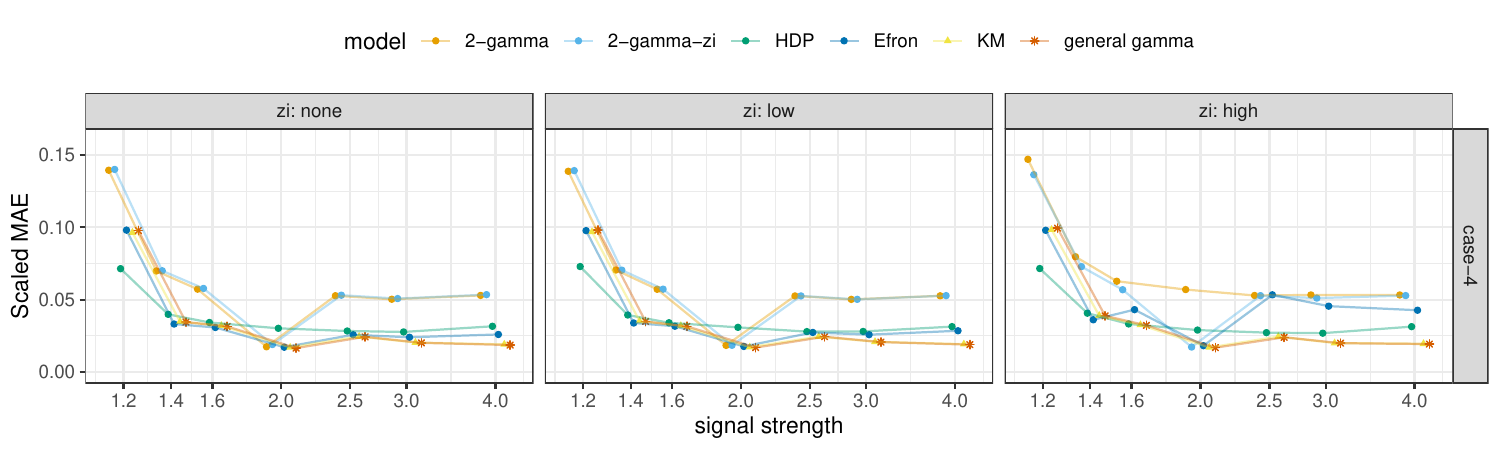}
    \caption{Simulation II (heterogeneous signal strengths): Max-Scaled-MAEs versus signal strength for Bayesian models (color-coded) with fixed truth in signal strength. The general-gamma and the KM perform better than other methods across zero-inflation levels and cases.}
\end{figure}

\begin{figure}[htp]
    \centering
    \includegraphics[width=16cm]{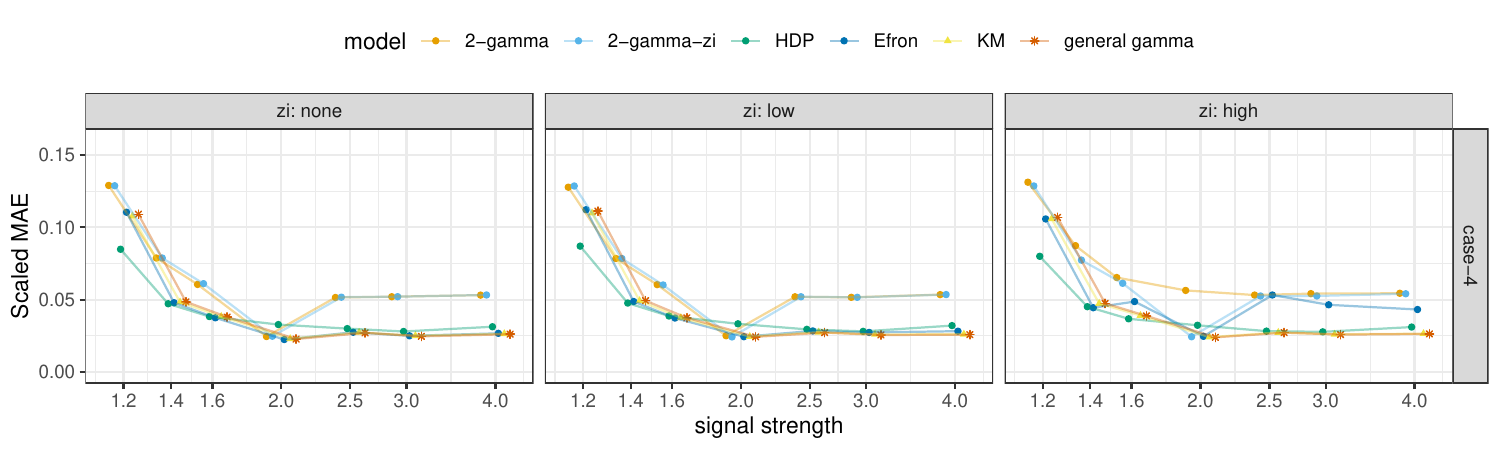}
    \caption{Simulation II (heterogeneous signal strengths): Max-Scaled-MAEs versus signal strength for Bayesian models (color-coded) with randomly perturbed truth in signal strength. The general-gamma and the KM perform better than other methods across zero-inflation levels and cases.}
\end{figure}

\begin{figure}[htp]
    \centering
    \includegraphics[width=16cm]{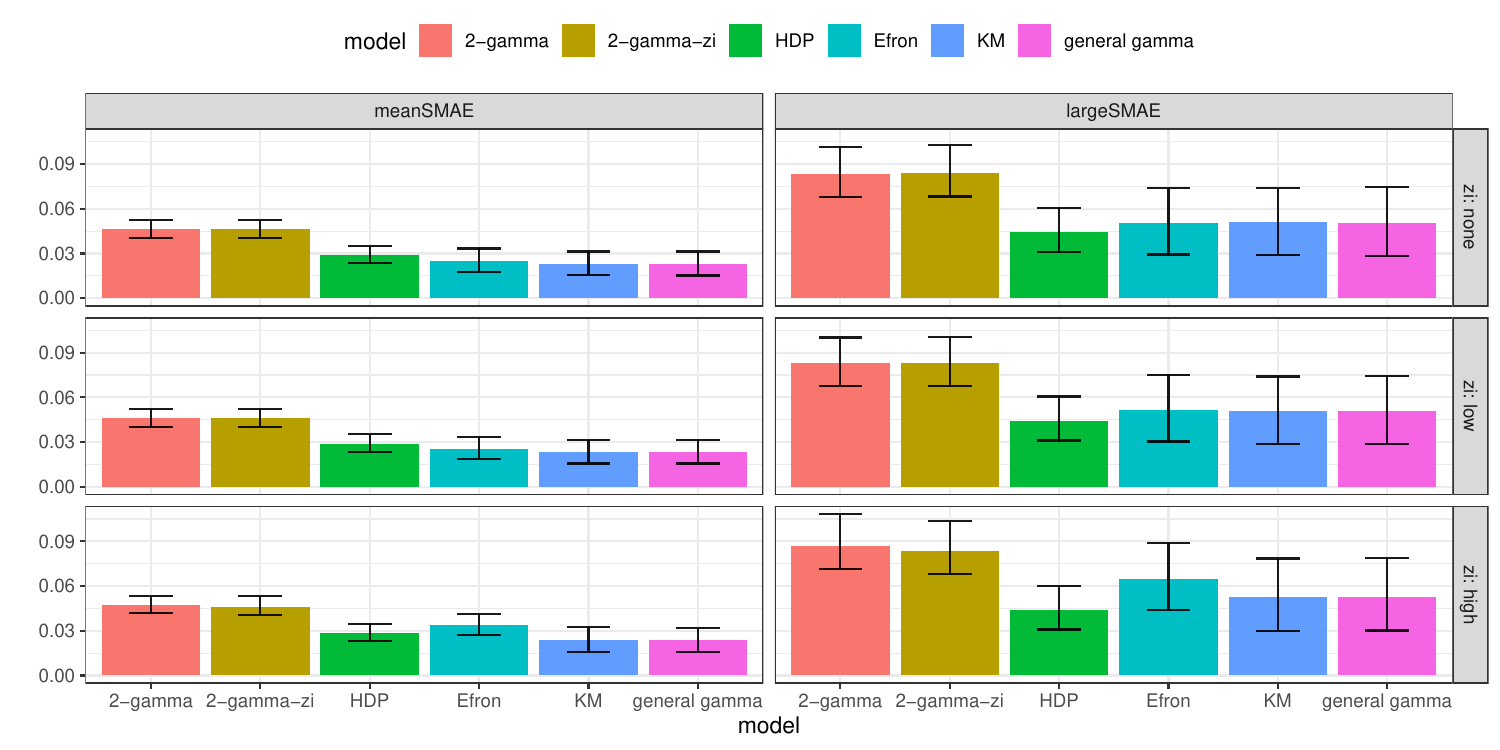}
    \caption{Simulation II (heterogeneous signal strengths): The overall means of average (left column) and maximum (right column) scaled MAEs obtained across all signal strength values and all choices of the number of signal cells with fixed truth in signal strength are plotted as vertical bars for different levels of zero inflation (along rows) for each bayesian models. Each bar represents the replication-based mean of the average (left column) or the maximum (right column) scaled MAE of a specific method computed over the signal cells of an entire table; the associated error whiskers represent the 5th and 95th percentile points across replicates. Non-parametric Bayesian methods (KM, general-gamma, HDP, and Efron) perform better than parametric Bayesian methods (2-gamma and 2-gamma-zi).}
\end{figure}

\begin{figure}[htp]
    \centering
    \includegraphics[width=16cm]{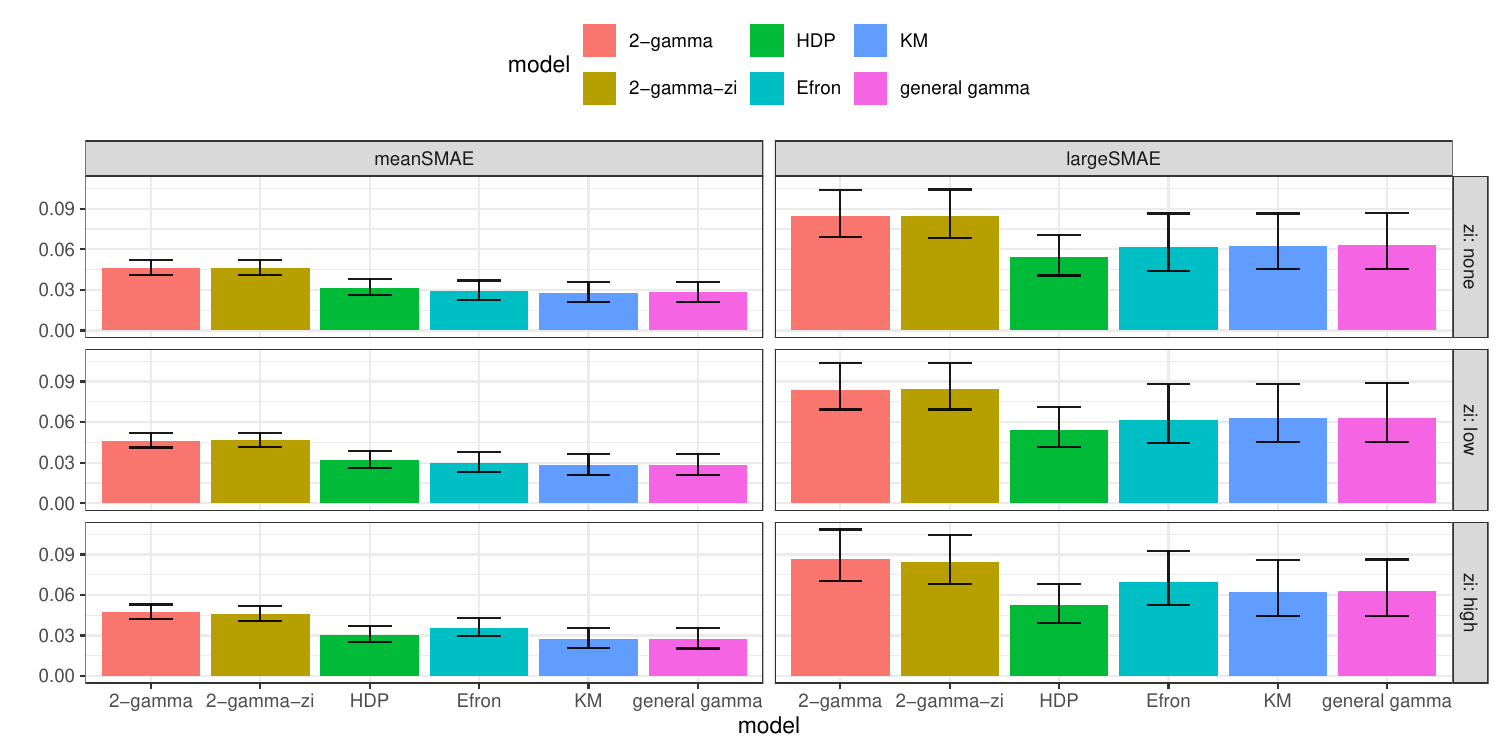}
    \caption{Simulation II (heterogeneous signal strengths): The overall means of average (left column) and maximum (right column) scaled MAEs obtained across all signal strength values and all choices of the number of signal cells with random perturbed truth in signal strength are plotted as vertical bars for different levels of zero inflation (along rows) for each bayesian models. Each bar represents the replication-based mean of the average (left column) or the maximum (right column) scaled MAE of a specific method computed over the signal cells of an entire table; the associated error whiskers represent the 5th and 95th percentile points across replicates. Non-parametric Bayesian methods (KM, general-gamma, HDP, and Efron) perform better than parametric Bayesian methods (2-gamma and 2-gamma-zi).}
\end{figure}

\begin{figure}[htp]
    \centering
    \includegraphics[width=16cm]{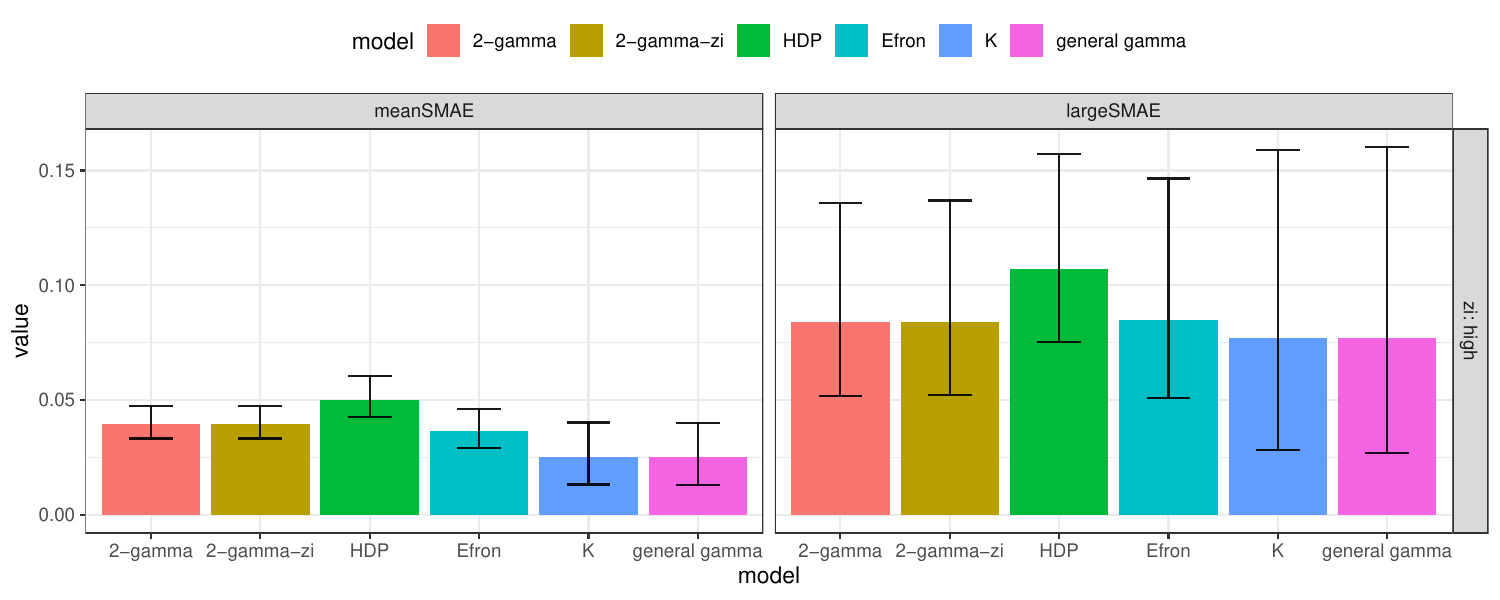}
    \caption{Simulation III (highly heterogeneous signal strengths): The overall means of average (left column) and maximum (right column) scaled MAEs obtained across all signal strength values and all choices of the number of signal cells with fixed truth in signal strength are plotted as vertical bars for different levels of zero inflation (along rows) for each bayesian models. Each bar represents the replication-based mean of the average (left column) or the maximum (right column) scaled MAE of a specific method computed over the signal cells of an entire table; the associated error whiskers represent the 5th and 95th percentile points across replicates.}
\end{figure}

\clearpage

\section{Statin 46 and Statin 42 tables}

We obtain the Statin 42 table from the Statin 46 table by excluding four adverse events with fewer than 70 counts in the 'Other Drugs' reference column, which are Blood Creatine Phosphokinase Mm Increased, Myoglobin Blood Present, Myoglobin Urine Present and Myoglobinaemia.

\begin{table}
\centering
\resizebox{13cm}{!}{
\begin{tabular}{|lrrrrrrr|}
\hline
  & Atorvastatin & Fluvastatin & Lovastatin & Pravastatin & Rosuvastatin & Simvastatin & Other\\
\hline
Acute Kidney Injury & 1353 & 42 & 7 & 154 & 689 & 823 & 355651\\
\hline
Anuria & 71 & 0 & 0 & 2 & 43 & 62 & 10403\\
\hline
Blood Calcium Decreased & 14 & 2 & 0 & 0 & 110 & 17 & 15918\\
\hline
Blood Creatine Phosphokinase Abnormal & 34 & 0 & 0 & 0 & 8 & 11 & 261\\
\hline
Blood Creatine Phosphokinase Increased & 1175 & 125 & 32 & 200 & 562 & 768 & 23805\\
\hline
Blood Creatine Phosphokinase Mm Increased & 2 & 0 & 0 & 0 & 9 & 0 & 14\\
\hline
Blood Creatinine Abnormal & 27 & 0 & 0 & 0 & 5 & 3 & 3385\\
\hline
Blood Creatinine Increased & 227 & 10 & 0 & 17 & 210 & 97 & 74742\\
\hline
Chromaturia & 340 & 10 & 6 & 33 & 174 & 114 & 19294\\
\hline
Chronic Kidney Disease & 152 & 16 & 2 & 19 & 177 & 37 & 339179\\
\hline
Compartment Syndrome & 53 & 0 & 0 & 1 & 12 & 12 & 2644\\
\hline
Creatinine Renal Clearance Decreased & 6 & 0 & 0 & 2 & 124 & 6 & 7768\\
\hline
Diaphragm Muscle Weakness & 14 & 0 & 0 & 0 & 8 & 1 & 94\\
\hline
Electromyogram Abnormal & 2 & 0 & 0 & 2 & 0 & 0 & 132\\
\hline
End Stage Renal Disease & 30 & 0 & 0 & 0 & 19 & 6 & 97553\\
\hline
Glomerular Filtration Rate Abnormal & 8 & 0 & 0 & 0 & 1 & 0 & 1069\\
\hline
Glomerular Filtration Rate Decreased & 59 & 1 & 0 & 8 & 39 & 29 & 13190\\
\hline
Hypercreatininaemia & 0 & 0 & 0 & 0 & 8 & 0 & 648\\
\hline
Hypocalcaemia & 36 & 0 & 0 & 16 & 8 & 18 & 23102\\
\hline
Muscle Disorder & 291 & 2 & 7 & 21 & 191 & 87 & 7329\\
\hline
Muscle Enzyme Increased & 48 & 1 & 0 & 0 & 13 & 9 & 410\\
\hline
Muscle Fatigue & 85 & 0 & 2 & 16 & 30 & 39 & 4257\\
\hline
Muscle Haemorrhage & 24 & 0 & 0 & 5 & 13 & 4 & 3806\\
\hline
Muscle Necrosis & 68 & 2 & 0 & 1 & 10 & 20 & 662\\
\hline
Muscle Rupture & 181 & 25 & 0 & 61 & 36 & 120 & 3219\\
\hline
Muscular Weakness & 1857 & 45 & 31 & 152 & 808 & 859 & 111003\\
\hline
Musculoskeletal Discomfort & 137 & 18 & 15 & 25 & 187 & 93 & 19931\\
\hline
Musculoskeletal Disorder & 56 & 3 & 0 & 9 & 65 & 73 & 25881\\
\hline
Musculoskeletal Pain & 420 & 3 & 2 & 38 & 324 & 228 & 82576\\
\hline
Myalgia & 5362 & 341 & 151 & 939 & 2757 & 3216 & 143819\\
\hline
Myasthenic Syndrome & 1 & 0 & 9 & 0 & 3 & 6 & 643\\
\hline
Myoglobin Blood Increased & 71 & 4 & 0 & 4 & 16 & 39 & 1003\\
\hline
Myoglobin Blood Present & 2 & 0 & 0 & 0 & 0 & 0 & 8\\
\hline
Myoglobin Urine Present & 0 & 0 & 0 & 0 & 1 & 2 & 70\\
\hline
Myoglobinaemia & 15 & 0 & 0 & 0 & 0 & 0 & 62\\
\hline
Myoglobinuria & 26 & 4 & 0 & 1 & 1 & 9 & 296\\
\hline
Myopathy & 849 & 64 & 45 & 145 & 377 & 544 & 6695\\
\hline
Myopathy Toxic & 31 & 0 & 0 & 1 & 2 & 21 & 457\\
\hline
Myositis & 219 & 8 & 10 & 28 & 62 & 141 & 7482\\
\hline
Necrotising Myositis & 279 & 0 & 0 & 2 & 10 & 52 & 278\\
\hline
Oliguria & 52 & 0 & 0 & 4 & 24 & 37 & 7590\\
\hline
Renal Failure & 534 & 26 & 11 & 69 & 225 & 195 & 250710\\
\hline
Renal Impairment & 390 & 52 & 11 & 37 & 161 & 169 & 103343\\
\hline
Renal Tubular Necrosis & 40 & 0 & 0 & 12 & 10 & 24 & 12762\\
\hline
Rhabdomyolysis & 2041 & 52 & 44 & 163 & 936 & 1376 & 31707\\
\hline
Tendon Discomfort & 9 & 0 & 0 & 3 & 22 & 10 & 794\\
\hline
Other Pt & 180699 & 4886 & 2845 & 20296 & 113960 & 76068 & 61724222\\
\hline
\end{tabular}}
\caption{ Statin 46}
\end{table}

\begin{table}
\centering
\resizebox{13cm}{!}{
\begin{tabular}{|lrrrrrrr|}
\hline
  & Atorvastatin & Fluvastatin & Lovastatin & Pravastatin & Rosuvastatin & Simvastatin & Other\\
\hline
Acute Kidney Injury & 1353 & 42 & 7 & 154 & 689 & 823 & 355651\\
\hline
Anuria & 71 & 0 & 0 & 2 & 43 & 62 & 10403\\
\hline
Blood Calcium Decreased & 14 & 2 & 0 & 0 & 110 & 17 & 15918\\
\hline
Blood Creatine Phosphokinase Abnormal & 34 & 0 & 0 & 0 & 8 & 11 & 261\\
\hline
Blood Creatine Phosphokinase Increased & 1175 & 125 & 32 & 200 & 562 & 768 & 23805\\
\hline
Blood Creatinine Abnormal & 27 & 0 & 0 & 0 & 5 & 3 & 3385\\
\hline
Blood Creatinine Increased & 227 & 10 & 0 & 17 & 210 & 97 & 74742\\
\hline
Chromaturia & 340 & 10 & 6 & 33 & 174 & 114 & 19294\\
\hline
Chronic Kidney Disease & 152 & 16 & 2 & 19 & 177 & 37 & 339179\\
\hline
Compartment Syndrome & 53 & 0 & 0 & 1 & 12 & 12 & 2644\\
\hline
Creatinine Renal Clearance Decreased & 6 & 0 & 0 & 2 & 124 & 6 & 7768\\
\hline
Diaphragm Muscle Weakness & 14 & 0 & 0 & 0 & 8 & 1 & 94\\
\hline
Electromyogram Abnormal & 2 & 0 & 0 & 2 & 0 & 0 & 132\\
\hline
End Stage Renal Disease & 30 & 0 & 0 & 0 & 19 & 6 & 97553\\
\hline
Glomerular Filtration Rate Abnormal & 8 & 0 & 0 & 0 & 1 & 0 & 1069\\
\hline
Glomerular Filtration Rate Decreased & 59 & 1 & 0 & 8 & 39 & 29 & 13190\\
\hline
Hypercreatininaemia & 0 & 0 & 0 & 0 & 8 & 0 & 648\\
\hline
Hypocalcaemia & 36 & 0 & 0 & 16 & 8 & 18 & 23102\\
\hline
Muscle Disorder & 291 & 2 & 7 & 21 & 191 & 87 & 7329\\
\hline
Muscle Enzyme Increased & 48 & 1 & 0 & 0 & 13 & 9 & 410\\
\hline
Muscle Fatigue & 85 & 0 & 2 & 16 & 30 & 39 & 4257\\
\hline
Muscle Haemorrhage & 24 & 0 & 0 & 5 & 13 & 4 & 3806\\
\hline
Muscle Necrosis & 68 & 2 & 0 & 1 & 10 & 20 & 662\\
\hline
Muscle Rupture & 181 & 25 & 0 & 61 & 36 & 120 & 3219\\
\hline
Muscular Weakness & 1857 & 45 & 31 & 152 & 808 & 859 & 111003\\
\hline
Musculoskeletal Discomfort & 137 & 18 & 15 & 25 & 187 & 93 & 19931\\
\hline
Musculoskeletal Disorder & 56 & 3 & 0 & 9 & 65 & 73 & 25881\\
\hline
Musculoskeletal Pain & 420 & 3 & 2 & 38 & 324 & 228 & 82576\\
\hline
Myalgia & 5362 & 341 & 151 & 939 & 2757 & 3216 & 143819\\
\hline
Myasthenic Syndrome & 1 & 0 & 9 & 0 & 3 & 6 & 643\\
\hline
Myoglobin Blood Increased & 71 & 4 & 0 & 4 & 16 & 39 & 1003\\
\hline
Myoglobinuria & 26 & 4 & 0 & 1 & 1 & 9 & 296\\
\hline
Myopathy & 849 & 64 & 45 & 145 & 377 & 544 & 6695\\
\hline
Myopathy Toxic & 31 & 0 & 0 & 1 & 2 & 21 & 457\\
\hline
Myositis & 219 & 8 & 10 & 28 & 62 & 141 & 7482\\
\hline
Necrotising Myositis & 279 & 0 & 0 & 2 & 10 & 52 & 278\\
\hline
Oliguria & 52 & 0 & 0 & 4 & 24 & 37 & 7590\\
\hline
Renal Failure & 534 & 26 & 11 & 69 & 225 & 195 & 250710\\
\hline
Renal Impairment & 390 & 52 & 11 & 37 & 161 & 169 & 103343\\
\hline
Renal Tubular Necrosis & 40 & 0 & 0 & 12 & 10 & 24 & 12762\\
\hline
Rhabdomyolysis & 2041 & 52 & 44 & 163 & 936 & 1376 & 31707\\
\hline
Tendon Discomfort & 9 & 0 & 0 & 3 & 22 & 10 & 794\\
\hline
Other Pt & 180699 & 4886 & 2845 & 20296 & 113960 & 76068 & 61724222\\
\hline
\end{tabular}}
\caption{ Statin 42}
\end{table}

\end{appendix}

\end{document}